\documentclass[11pt]{article}

\usepackage{amsmath,rotating,subcaption}
\usepackage{times}
\usepackage{bm}
\usepackage{natbib}
\usepackage{algorithm}
\usepackage{xargs}[2008/03/08]
\usepackage{amssymb}
\usepackage{mathrsfs}
\usepackage{graphicx}
\usepackage[flushleft]{threeparttable}
\usepackage{makecell}
\usepackage{appendix,chngcntr,apptools}

\usepackage{mathtools}

\usepackage{breakcites}

\usepackage{multirow}
\usepackage[colorlinks,
linkcolor=blue,
anchorcolor=blue,
citecolor=blue,
breaklinks=true
]{hyperref}

\usepackage{star}   

\def\T{{ \mathrm{\scriptscriptstyle T} }}

\def\##1\#{\begin{align}#1\end{align}}
\def\$#1\${\begin{align*}#1\end{align*}}

\def\T{{ \mathrm{\scriptscriptstyle T} }} 

\newcommand{\Rom}[1]{\text{\uppercase\expandafter{\romannumeral #1\relax}}}

\usepackage{geometry}
\def\En{\widehat{\mathbb{E}}_{(\sigma^\star, \bm{\alpha}^\star, \bm{\beta}^\star)}}
\def\Ea{\widehat{\mathbb{E}}_{(\sigma, \bm{\alpha}, \bm{\beta})}}
\def\Pn{\widehat{\mathbb{P}}_{(\sigma^\star, \bm{\alpha}^\star, \bm{\beta}^\star)}}
\def\Pa{\widehat{\mathbb{P}}_{(\sigma, \bm{\alpha}, \bm{\beta})}}

\begin{document}

\title{ \LARGE  Bayesian Factor-adjusted Sparse Regression}     

\author{Jianqing Fan\footnotemark[1], ~Bai Jiang\thanks{Department of Operations Research and Financial Engineering, Princeton University, Princeton, NJ 08544; E-mail: \texttt{jqfan@princeton.edu, baij@princeton.edu}.}, ~ and ~Qiang Sun\thanks{Department of Statistical Sciences, University of Toronto, Toronto, ON M5S 3G3; E-mail: \texttt{qsun@utstat.toronto.edu}.}}

\date{January 31, 2019}

\maketitle

\vspace{-0.25in}

\begin{abstract}

This paper investigates the high-dimensional linear regression with highly correlated covariates. In this setup, the traditional sparsity assumption on the regression coefficients often fails to hold, and consequently many model selection procedures do not work. To address this challenge, we model the variations of covariates by a factor structure. Specifically, strong correlations among covariates are explained by common factors and the remaining variations are interpreted as idiosyncratic components of each covariate. This leads to a factor-adjusted regression model with both common factors and idiosyncratic components as covariates. We generalize the traditional sparsity assumption accordingly and assume that all common factors but only a small number of idiosyncratic components contribute to the response. A Bayesian procedure with a spike-and-slab prior is then proposed for parameter estimation and model selection. Simulation studies show that our Bayesian method outperforms its lasso analogue, manifests insensitivity to the overestimates of the number of common factors, pays a negligible price in the no correlation case, and scales up well with increasing sample size, dimensionality and sparsity. Numerical results on a real dataset of U.S. bond risk premia and macroeconomic indicators lend strong support to our methodology.
\end{abstract}
\noindent
{keywords}: factor model, Bayesian sparse regression, posterior convergence rate, model selection.

\section{Introduction}

High-dimensional linear models are useful for a wide arrays of economic problems \citep{fan2011sparse, belloni2012sparse}. These models typically assume the sparsity of regression coefficients, that is, only a small number of covariates have significant effects on the response. However, the explanatory variables in the panel of an economic dataset are often highly correlated due to the influence of latent common factors, rendering the sparsity assumption unreasonable and restrictive. To address this issue, this paper proposes a general regression model with a factor-adjusted sparsity assumption, and develops a Bayesian method for this model.

To motivate the factor-adjusted model and its corresponding methodology, we start with the standard linear regression model
\begin{equation} \label{model1}
\mbf{Y}_{n \times 1} = \mbf{X}_{n \times p}\bm{\beta}_{p \times 1} + \sigma \bm{\varepsilon}_{n \times 1},
\end{equation}
where $\mbf{Y}_{n \times 1}=(y_1,\ldots, y_n)^\T$ is an $n\times 1$ response vector, $\mbf{X}_{n \times p}=(\bm{x}_1,\ldots, \bm{x}_n)^\T = [\mbf{X}_1,\dots,\mbf{X}_p]$ is a design matrix of $n$ observations and $p$ covariates, $\bm{\beta} = (\beta_1,\dots,\beta_p)^\T$ is a $p$-dimensional vector of regression coefficients, $\sigma$ is an unknown standard deviation, and $\bm \varepsilon$ is an $n$-dimensional standard Gaussian random vector, independent of $\mbf{X}$. Without loss of generality, we assume $\mbb{E}\mbf{X}_j = \mbf{0}$ and include no intercept term in the model. Of interest is the high-dimensional regime in which the dimensionality $p$ is much larger than the sample size $n$.

This model has attracted intensive interests in the frequentist community (\citealp{tibshirani1996regression, fan2001variable, candes2007dantzig, fan2008sure, zhang2008sparsity, su2016slope}, among others). All of these methods hinge on at least two basic assumptions. The first one assumes that the correlations between explanatory variables are sufficiently weak. Examples of this assumption are the mutual coherence condition \citep{donoho2001uncertainty, donoho2003optimally, donoho2006stable, bunea2007sparsity}, the irrepresentable condition \citep{zhao2006model}, the restricted eigenvalue condition \citep{bickel2009simultaneous, fan2018lamm} and the uniform compatibility condition \citep[page 157]{buhlmann2011statistics}. The second one, referred to as the sparsity assumption, assumes that only a small number $s$ of covariates contribute to the response. Formally, the sparsity, defined as $s := |\{j: \beta_j \neq 0\}|$, is much smaller than the dimensionality $p$.

Nevertheless, the weak correlation conditions do not necessarily hold in many applications, especially those in economic and financial studies. In an economic or financial dataset, the explanatory variables, e.g., stock returns or macroeconomic indicators over a period of time, are often influenced by similar economic fundamentals and are thus heavily correlated due to the existence of co-movement patterns \citep{forbes2002no, stock2002forecasting}. In the presence of such strong correlations introduced by common factors, one naturally expects strong effects of common factors on the response. If this is true, many covariates would have non-negligible effects on the response, rendering the traditional sparsity assumption in the standard regression model \eqref{model1} ideologically unreasonable.

The above argument shows the necessity to take the correlation structure of explanatory variables into account and adjust the sparsity assumption accordingly. For this purpose, we consider using factor models \citep{stock2002forecasting, bai2003inferential, bai2006confidence, fan2013large} and assume that each datum (row) $\bm{x}_i\in\mathbb{R}^p$ of the data matrix $\mbf{X}$ exhibits a decomposition of form
\begin{equation}\label{model2}
\bm{x}_i = \mbf{B} \bm{f}_i + \bm{u}_i,
\end{equation}
where $\mbf{B} = [\bm{b}_1,\dots,\bm{b}_p]^\T$ is a $p \times k$ unknown matrix of factor loading coefficients, $\bm{f}_i$ is a $k$-dimensional random vector of common factors, and $\bm{u}_i$ is a $p$-dimensional random vector of weakly-correlated idiosyncratic components, uncorrelated with $\bm{f}_i$. Without loss of generality, we assume $\mbb{E}\bm{f}_i = \mbf{0}$, $\mbb{E}\bm{u}_i = \mbf{0}$, and $\Cov(\bm{f}_i) = \mbf{I}$. Both common factors and idiosyncratic components are latent, but they are often estimated by using principal component analysis (PCA) \citep{bai2003inferential,fan2013large,wang2015asymptotics}. Model \eqref{model2} embraces the well-known CAPM model \citep{sharpe1964capital, lintner1975valuation} and Fama-French model \citep{fama1993common} as its special cases, with observable common factors. Let $\mbf{F}_{n \times k} = [\bm{f}_1,\dots,\bm{f}_n]^\T = [\mbf{F}_1,\dots,\mbf{F}_k]$ be the matrix of common factors, and $\mbf{U}_{n \times p} = [\bm{u}_1,\dots,\bm{u}_n]^\T = [\mbf{U}_1,\dots,\mbf{U}_p]$ be the matrix of idiosyncratic components. Then a more compact matrix form reads as
\begin{align} \label{model3}
\mbf{X} = \mbf{F} \mbf{B}^\T + \mbf{U}.
\end{align}
Each covariate (column) $\mbf{X}_j$ in $\mbf{X}$ can be decomposed as a sum of two components $\mbf{F}\bm{b}_j$ and $\mbf{U}_j$, reflecting the influence of common factors and idiosyncratic variations respectively.

Utilizing this factor structure \eqref{model3}, we generalize the standard sparse regression model \eqref{model1} to a factor-adjusted sparse regression model of the form
\begin{equation} \label{model4}
\mbf{Y}_{n \times 1} = \mbf{F}_{n \times k} \bm{\alpha}_{k \times 1} + \mbf{U}_{n \times p}\bm{\beta}_{p \times 1} + \sigma \bm{\varepsilon}_{n \times 1},
\end{equation}
where $\bm{\alpha}$ and $\bm{\beta}$ are regression coefficient vectors of $\mbf{F}$ and $\mbf{U}$, respectively. We assume that $\bm{\alpha}$ is dense (as it is usually low-dimensional) but $\bm{\beta}$ is sparse. That is, all common factors but only a small number of idiosyncratic components of the original explanatory variables contribute to the response. A non-zero $\beta_j$ indicates that the covariate $\mbf{X}_j$, excluding the strong correlation with other covariates, has a \textit{specific} effect on the response. Compared to the traditional sparsity assumption, this factor-adjusted sparsity assumption is more tenable as the idiosyncratic components are weakly-correlated.

We remark that our generalized factor-adjusted regression model \eqref{model4} covers the standard regression model \eqref{model1} as a special case by restricting the side constraint that $\bm{\alpha} = \mbf{B}^\T \bm{\beta}$. Under this constraint, the factor-adjusted sparsity assumption imposed on regression coefficients of idiosyncratic components in model \eqref{model4} coincides with the traditional sparsity assumption in model \eqref{model1}. Thus any statistical method for estimating model \eqref{model4} would estimate model \eqref{model1}. Of course, when such a constraint is not enforced, model \eqref{model4} provides more flexibility in the regression analysis than model \eqref{model1}.

Model \eqref{model4} is similar but different from the factor-augmented regression or the augmented principal component regression of \citet{stock2002forecasting, bai2006confidence}. In the factor-augmented models, the factors are usually extracted from a large panel of data via PCA and used as a part of covariates, yet the other variables are introduced from outside of the panel. These models are typically low-dimensional. In contrast, model \eqref{model4} takes idiosyncratic components as covariates, which are created internally from the panel of the data. This allows to explore additional explanatory power of the data. Our analyses of model \eqref{model4} in the high-dimensional fashion are applicable to the low-dimensional factor-augmented regression models in the literature, as model \eqref{model4} can easily incorporate external variables in the part of $\mbf{F}$ and/or $\mbf{U}$. For simplicity of presentation, we omit the details.

\citet{kneip2011factor} gave an insightful discussion on the limitation of the traditional sparse assumption in model \eqref{model1} with factor-structured covariates and proposed a factor-augmented regression model. Nevertheless, they still need the weak correlation condition on the original covariates, which is unlikely to hold for factor-structured covariates. See equation (5.5) of \citet{kneip2011factor}. \citet{fan2016decorrelation} pointed out the failure of classical frequentist methods dealing with model \eqref{model1} with factor-structured covariates, and proposed a frequentist method for estimating model \eqref{model4}. Specially, they estimated the latent common factors and idiosyncratic components, and then run frequentist sparse regression methods (e.g., lasso) on estimated common factors and idiosyncratic components. Similar to ours, they impose the weak correlation condition on idiosyncratic components, instead of the original covariates. See Example 3.2 of \cite{fan2016decorrelation}.

This paper focuses on Bayesian solutions to model \eqref{model4}. As shown in Section \ref{sec:2}, the fully Bayesian procedure cannot work easily due to the involvement of latent common factors and idiosyncratic components in the posterior computation. Inspired by \citet{fan2016decorrelation}, we consider estimating these latent variables by PCA and running a Bayesian sparse regression method on their estimates. The arsenal of Bayesian sparse regression methods, including those exploiting shrinkage priors (e.g. \citealp{park2008bayesian, polson2012local, armagan2013generalized, bhattacharya2015dirichlet, song2017nearly}) and those exploiting spike-and-slab priors (among others, \citealp{ishwaran2005spike, narisetty2014bayesian, castillo2015bayesian, rovckova2018spike}), has been developed in parallel to the frequentist methods. However, it is unclear whether these methods would work on estimated common factors and idiosyncratic components in model \eqref{model4}. When it does work, it remains unknown whether the factor model estimation would incur any loss to the convergence rate or model selection consistency of the Bayesian sparse regression method. Given theoretical results in the frequentist setting, these questions are still challenging, because the definitions of estimation errors and technical conditions of frequentist and Bayesian methods are significantly different \citep{castillo2015bayesian}. Even if a Bayesian sparse regression method is theoretically sound, it is unclear whether it performs better or worse than the frequentist methods on finite sample data. We would like answer these questions in the current paper.

Specifically, our Bayesian method imposes a slab prior on the regression coefficients of estimated common factors, and a spike-and-slab prior on the regression coefficients of estimated idiosyncratic components. This procedure results in a pseudo-posterior distribution, which differs from the exact posterior distribution obtained by a Bayesian regression on exact common factors and idiosyncratic components. Interestingly, the pseudo-posterior distribution achieves the $\ell_2$ contraction rate $\sqrt{s\log p/n}$ of the regression coefficients, which matches that of the exact posterior distribution. Byproducts of our analyses include the adaptivity to the unknown sparsity $s$ and the unknown standard deviation $\sigma$. We only need a type of sparse eigenvalue condition on the idiosyncratic components to overcome the non-identifiability issue of the parameters. This is easy to hold due to the weak correlation among idiosyncratic components. Moreover, by assuming a beta-min condition that is frequently used in the high-dimensional regression literature, we prove that our method consistently selects the support of the true sparse regression coefficients.

The rest of this paper proceeds as follows. In Section \ref{sec:2}, we propose the Bayesian methodology for the factor-adjusted regression model \eqref{model4}. Section \ref{sec:3} establishes the contraction rates and model selection consistency of the pseudo-posterior distribution. These theoretical results rely on a high-level condition concerning the estimation of factor models, which is examined by Section \ref{sec:4}. Section \ref{sec:5} presents experimental results on simulation datasets. Section \ref{sec:6} applies our method to a real dataset of U.S. bond risk premia and macroeconomic indicators. Section \ref{sec:7} is devoted to discussions. All technical proofs and algorithmic implementation are detailed in the appendices.

\textbf{Notation.} We write $\diag(a_1,\dots,a_m)$ for a diagonal matrix of elements $a_1,\dots,a_m$. For a symmetric matrix $\mbf{A}$, we write its largest eigenvalue as $\lambda_{\max}(\mbf{A})$ and its smallest eigenvalue as $\lambda_{\min}(\mbf{A})$. For a matrix $\mbf{A}_{m_1 \times m_2} = [a_{ij}]_{1 \le i \le m_1, 1 \le j \le m_2}$, we write $\mbf{A}_j$ to denote its $j$-th column, and lowercase $\bm{a}_i$ to denote its $i$-th row. For a index set $\xi \subseteq \{1,\dots,m_2\}$, $\mbf{A}_\xi = [\mbf{A}_j: j \in \xi]$ is the sub-matrix of $\mbf{A}$ assembling the columns indexed by $\xi$. Let $\Vert \mbf{A}\Vert_{\max} = \max_{i,j} |a_{ij}|$ be the element-wise maximum norm of $\mbf{A}$, let $\Vert \mbf{A} \Vert_\text{F}$ be its Frobenius norm. For a vector $\bm{v}$, let $\bm{v}_\xi$ denote its sub-vector assembling components indexed by $\xi$, and let $\Vert \bm{v} \Vert$ denote its $\ell_2$ norm. For two sequences $a_n$ and $b_n$, $a_n \prec b_n$ or $b_n \succ a_n$ means $a_n = \smallo(b_n)$.

\section{Model and Methodology}\label{sec:2}
Our goal is to study the factor-adjusted regression model \eqref{model4}, in which both common and idiosyncratic components $[\mbf{F}, \mbf{U}]$ are unobserved, but $\mbf{X}$ are observed through \eqref{model3}. Each datum (row) $\bm{x}_i$ in $\mbf{X}$ admits the factor structure \eqref{model2} with $\{(\bm{f}_i, \bm{u}_i)\}_{1 \le i \le n}$ therein identically distributed as $(\bm{f}, \bm{u})$. Note that $\{(\bm{f}_i, \bm{u}_i)\}_{1 \le i \le n}$ are not necessarily independently distributed. The dimension $k$ of $\bm{f}$ is fixed, but the dimension $p$ of $\bm{u}$ may grow as $n$ increases. By decomposition, $\bm{f}$ and $\bm{u}$ are uncorrelated. Without loss of generality, we assume that $\mbb{E}\bm{f} = \mbf{0}$, $\mbb{E}\bm{u} = \mbf{0}$ and $\Cov(\bm{f}) = \mbf{I}$. The regression coefficient vector $\bm{\beta}$ of $\mbf{U}$ is sparse in the sense that $s = |\{j: \beta_j \neq 0\}|$ is small. We allow $s$ to grow as $n$ increases, but require $s \prec n / \log p$ so that the desired $\ell_2$ contraction rate $\sqrt{s\log p/n} \to 0$ as $n \to \infty$. The Gaussian errors $\bm{\varepsilon}$ are independent from $\mbf{F}$ and $\mbf{U}$.

An inherent difficulty for estimating model \eqref{model4} is that both common factors and idiosyncratic components are unobserved. Therefore the first step is to estimate these unobserved variables. We follow \citet{bai2003inferential,fan2013large,wang2015asymptotics} and use PCA for this task. Let $\widehat{\lambda}_1 \ge \dots \ge \widehat{\lambda}_n$ be the $n$ eigenvalues of $\mbf{X}\mbf{X}^\T/n$. A natural estimator of $\mbf{F}$ is the concatenation of the $k$ square-root-$n$-scaled eigenvectors corresponding to the top $k$ eigenvalues of $\mbf{X}\mbf{X}^\T/n$, denote by $\widehat{\mbf{F}}$. That is,
$$\frac{\mbf{X}\mbf{X}^\T}{n} \widehat{\mbf{F}} = \widehat{\mbf{F}} \widehat{\mbf{\Lambda}}, ~~~\frac{\widehat{\mbf{F}}^\T\widehat{\mbf{F}}}{n} = \mbf{I}, ~~~\widehat{\mbf{B}} = \mbf{X}^\T \widehat{\mbf{F}}/n,$$
where $\widehat{\mbf{\Lambda}} = \diag(\widehat{\lambda}_1,\dots,\widehat{\lambda}_k).$
Then we estimate $\mbf{U}$ by
$$
\widehat{\mbf{U}} = \mbf{X} -\widehat{\mbf{F}}\widehat{\mbf{B}}^\T = (\mbf{I} - \widehat{\mbf{F}}\widehat{\mbf{F}}^\T/n)\mbf{X}.$$
If $k$ is unknown, we may estimate $k$ by
\begin{equation} \label{khat}
\widehat{k} = \argmax_{k \le k_{\max}} \frac{\text{$k$-th eigenvalue of $\mbf{X}^\T\mbf{X}/n$}}{\text{$(k+1)$-th eigenvalue of $\mbf{X}^\T\mbf{X}/n$}},
\end{equation}
where $k_{\max}$ is any prescribed upper bound for $k$ \citep{luo2009contour, lam2012factor, ahn2013eigenvalue}.

After estimating unobserved variables, we propose a Bayesian sparse regression method for tasks of parameter estimation and model selection. Suppose we are given data $(\mbf{X}, \mbf{Y})$ generated from true parameter $(\sigma^\star, \bm{\alpha}^\star, \bm{\beta}^\star)$. Let $(\sigma, \bm{\alpha}, \bm{\beta})$ be its running parameter. Let $\xi = \{j: \beta_j \ne 0\}$ and $\xi^\star = \{j: \beta^\star_j \ne 0\}$ be the support of $\bm{\beta}$ and $\bm{\beta}^\star$, respectively. We consider a hierarchical prior $\pi(\sigma^2, \bm{\alpha}, \bm{\beta})$ with a slab prior on the coefficients of common factors $\mbf{F}$ and a spike-and-slab prior on the coefficients of idiosyncratic components $\mbf{U}$ as follows:
\begin{equation} \label{prior}
\begin{split}
\sigma^2 &\sim g(\sigma^2),\\
\bm{\alpha}|\sigma^2 &\sim \prod_{j=1}^k \frac{1}{\sigma}h\left(\frac{\alpha_j}{\sigma}\right),\\
1\{j \in \xi\} &\sim \texttt{Bernoulli}(s_0/p),\\
\bm{\beta}_\xi|\sigma^2 &\sim \prod_{j \in \xi} \frac{1}{\tau_j\sigma}h\left(\frac{\beta_j}{\tau_j\sigma}\right), ~~~\bm{\beta}_{\xi^c}|\sigma^2 = 0,
\end{split}
\end{equation}
where $g$ is a positive continuous density function on $(0,\infty)$, e.g., the inverse-gamma density; $h$ is a ``slab'' density function on $(-\infty,+\infty)$ in the sense that $-\log [\inf_{|z| \le t} h(z)]  = \bigO(t^2)$ as $t \to \infty$, e.g., the Gaussian density $e^{-z^2/2}/\sqrt{2\pi}$ and the Laplace density $e^{-|z|/2}/2$; hyperparameters $\tau_1,\dots,\tau_p$ control the scales of running coefficients $\beta_1,\dots,\beta_p$; and, hyperparameter $s_0$ controls the sparsity of running models $\xi$. For the scaling hyperparameters, we set $\tau_j^{-1} = \Vert \widehat{\mbf{U}}_j \Vert/\sqrt{n}$ so that the effects of possibly heterogeneous scales of $\widehat{\mbf{U}}_j$'s are appropriately adjusted. For the sparsity hyperparameter, we simply set $s_0 = 1$ in the simulation experiments. When dealing with a real dataset, one could choose an informative $s_0$ according to expertise knowledges in the specific area, or tune $s_0$ by sophisticated cross-validation or empirical Bayes procedures

The Bayesian sparse regression on response $\mbf{Y}$ and regressors $\widehat{\mbf{F}}$, $\widehat{\mbf{U}}$ with prior \eqref{prior} obtains a pseudo-posterior distribution
\begin{equation} \label{posterior}
\widehat{\pi}(\sigma^2, \bm{\alpha}, \bm{\beta}|\mbf{X}, \mbf{Y})\\
= \widehat{\pi}(\sigma^2, \bm{\alpha}, \bm{\beta}|\widehat{\mbf{F}}, \widehat{\mbf{U}}, \mbf{Y})\\
\propto \pi(\sigma^2, \bm{\alpha}, \bm{\beta}) \mc{N}(\mbf{Y}|\widehat{\mbf{F}}\bm{\alpha} + \widehat{\mbf{U}}\bm{\beta}, \sigma^2 \mbf{I}),
\end{equation}
where $\mc{N}(\mbf{Y}|\bm{\mu}, \sigma^2 \mbf{I})$ denotes the $n$-dimensional normal distribution with mean $\bm{\mu}_{n \times 1}$ and covariance $\sigma^2 \mbf{I}$. We call it a ``pseudo-posterior'' distribution and put a hat over $\pi$ to emphasize that it differs from the exact posterior distributions $\pi(\sigma^2, \bm{\alpha}, \bm{\beta}|\mbf{F}, \mbf{U}, \mbf{Y})$, obtained by a Bayesian regression on observed $[\mbf{F}, \mbf{U}]$, and $\pi(\sigma^2, \bm{\alpha}, \bm{\beta}|\mbf{X}, \mbf{Y})$, obtained by a fully Bayesian procedure.

It is worth noting that, even in the simplest setting in which $\{(\bm{f}_i, \bm{u}_i)\}_{1 \le i \le n}$ are i.i.d. and $\bm{f}_i \sim P_f, \bm{u}_i \sim P_u$ are jointly independent, the exact posterior distribution given by a fully Bayesian procedure
\begin{align*}
&~~~\pi(\sigma^2, \bm{\alpha}, \bm{\beta}|\mbf{X}, \mbf{Y})\\
&\propto \pi(\sigma^2, \bm{\alpha}, \bm{\beta}) \int \mc{N}(\mbf{Y}|\mbf{F}\bm{\alpha} + (\mbf{X} - \mbf{F}\mbf{B}^\T)\bm{\beta}, \sigma^2 \mbf{I}) \prod_{i=1}^n P_f(\bm{f}_i) P_u(\bm{x}_i - \mbf{B}\bm{f}_i)d\bm{f}_i,
\end{align*}
is computationally intractable due to the involvement of latent variables in the complicated integral. Thus a fully Bayesian procedure does not solve model \eqref{model4} easily.

\section{Theory} \label{sec:3}
In this section, we show under commonly-seen assumptions for Bayesian sparse regression methods that the pseudo-posterior distribution \eqref{posterior} achieves the convergence rate $\epsilon_n = \sqrt{s\log p/n}$ of the $\ell_2$ estimation error for the coefficient vectors $(\bm{\alpha}^\star, \bm{\beta}^\star)$. This rate is so far the best rate Bayesian methods can achieve with observed $[\mbf{F}, \mbf{U}]$ \citep{song2017nearly}. We see that the factor adjustment added by our approach to the Bayesian sparse regression method incurs no loss in terms of $\ell_2$ estimation error rate. Byproducts of our analysis are the adaptivities of the pseudo-posterior distribution to the unknown sparsity $s$ and unknown standard deviation $\sigma^\star$. Finally, when the beta-min condition holds, we establish the model selection consistency of the pseudo-posterior distribution \eqref{posterior}.

\subsection{Assumptions}

In the high-dimensional regime $p \succ n$, a common assumption is that $\bm{\beta}^\star$ is sparse of size $s$. Following the sparse regression literature, we assume that $s \prec n / \log p$ such that the desired error rate $\epsilon_n = \sqrt{s\log p/n} \to 0$ as $n \to \infty$. To recover the sparse coefficient vector $\bm{\beta}^\star$ at rate $\epsilon_n$, we need the following assumptions.

\begin{assumption}\label{asm:1}
There exists a large integer $\bar{p}(n,p) \succ s$ and a constant $\kappa_0 > 0$ such that
$$\min_{\xi: |\xi| \le \bar{p}} \lambda_{\min}(\mbf{U}_\xi^\T\mbf{U}_\xi/n) \ge \kappa_0$$
holds with probability approaching $1$.
\end{assumption}

This assumption is commonly referred to as the sparse eigenvalue condition in the frequentist literature \citep{fan2018lamm}. In a recent study of Bayesian sparse regression with shrinkage priors, \citet{song2017nearly} imposed the same assumption on original covariates $\mbf{X}$. Here our assumption is imposed on their idiosyncratic components $\mbf{U}$.

Our next assumption upper bounds the maximum eigenvalue of $\mbf{U}_{\xi^\star}^\T\mbf{U}_{\xi^\star}/n$, which is the Gram matrix corresponding to the true model $\xi^\star = \{j:\beta^\star_j \ne 0\}$. Assumptions \ref{asm:1}-\ref{asm:2} together ensure that $\mbf{U}_{\xi^\star}^\T\mbf{U}_{\xi^\star}/n$ is well conditioned.

\begin{assumption}
\label{asm:2}
There exists a constant $\kappa_1 > 0$ such that
$$\lambda_{\max}(\mbf{U}_{\xi^\star}^\T\mbf{U}_{\xi^\star}/n) \le \kappa_1$$
holds with probability approaching $1$.
\end{assumption}

\citet{raskutti2010restricted, dobriban2016regularity} gave sufficient conditions for correlated covariates to satisfy Assumptions \ref{asm:1}-\ref{asm:2}. These theories typically allow $\bar{p}(n,p) \asymp n/\log p$ in Assumption \ref{asm:1}. If $\mbf{U}_{\xi^\star}$ consists of i.i.d. entries with zero mean, unit variance and only finite fourth moment, Assumption \ref{asm:2} holds by Bai-Yin theorem in the random matrix theory \citep{bai1988necessary, yin1988limit}.


Since we feed a Bayesian sparse regression method with the estimated variables $[\widehat{\mbf{F}}, \widehat{\mbf{U}}]$ rather than the latent variables $[\mbf{F},\mbf{U}]$, it is necessary to control the error of $(\widehat{\mbf{F}}\bm{\alpha} + \widehat{\mbf{U}}\bm{\beta}) - (\mbf{F}\bm{\alpha}^\star + \mbf{U}\bm{\beta}^\star)$. This goal is achieved by assumptions on the estimation errors of latent variables and the magnitudes of the true coefficient vectors. For the estimation error of the factor model, we impose a generic high-level condition as follows.
\begin{assumption}\label{asm:3}
The latent common factors and idiosyncratic components can be estimated by $\widehat{\mbf{F}}$ and $\widehat{\mbf{U}}$ as follows.
\begin{align*}
\max_{1 \leq j \leq k} \Vert (\widehat{\mbf{F}}\mbf{H})_j - \mbf{F}_j\Vert &= \bigOp(\sqrt{\log p}),\\
\max_{1 \leq j \leq p} \Vert \widehat{\mbf{U}}_j - \mbf{U}_j\Vert &= \bigOp(\sqrt{\log p}),
\end{align*}
for some nearly orthogonal matrix $\mbf{H}_{k \times k}$ such that $\Vert \mbf{H}^\T\mbf{H} - \mbf{I} \Vert = \bigOp(\sqrt{\log p/n})$ and $\Vert \mbf{H}\mbf{H}^\T - \mbf{I} \Vert = \bigOp(\sqrt{\log p/n})$.\end{assumption}

Since $\widehat{\mbf{F}}$ represents the eigenspace of the top $k$ eigenvalues of $\mbf{X}\mbf{X}^\T$ and mimics the column space of $\mbf{F}$, there is a nearly-orthogonal transformation, represented by $\mbf{H}$, between $\mbf{F}$ and $\widehat{\mbf{F}}$. Next section will verify this error rate in factor models under standard assumptions.

Our last assumption requires constant orders of the true parameters $(\sigma^\star, \bm{\alpha}^\star, \bm{\beta}^\star)$.
\begin{assumption} \label{asm:4}
$\sigma^\star > 0$ is fixed, $\Vert \bm{\alpha}^\star \Vert = \bigO(1)$, and $\Vert \bm{\beta}^\star \Vert = \bigO(1)$.
\end{assumption}

This condition is not restrictive. It holds if and only if the response variable has finite variance, under Assumptions \ref{asm:1}-\ref{asm:2}. To see this point, note that the variance of a single response variable $y = \bm{f}^\T\bm{\alpha}^\star  + \bm{u}^\T \bm{\beta}^\star + \sigma^\star \varepsilon$ is
$$\Var(y) = \Vert \bm{\alpha}^\star \Vert^2 + (\bm{\beta}^\star_{\xi^\star})^\T \Cov(\bm{u}_{\xi^\star})\bm{\beta}^\star_{\xi^\star} + \sigma^{\star 2},$$
where $\bm{u}_{\xi^\star}$ is the sub-vector of $\bm{u}$ corresponding to the true model $\xi^\star$, and $\Cov(\bm{u}_{\xi^\star})$ have all eigenvalues bounded away from $0$ and $\infty$, due to Assumptions \ref{asm:1}-\ref{asm:2}. Although our theoretical analyses need bounded magnitude of regression coefficients to avoid the amplification of estimation errors of latent variables, we remark here that, when the underlying true factors, $\mbf{F}_j$'s and $\mbf{U}_j$'s, and/or more accurate estimates are available, we can allow larger magnitudes of regression coefficients.

\subsection{Definition of Posterior Contraction Rate}
The definition of convergence rate in the Bayesian setting differs from that in the frequentist setting. We formally define it by following the classical Bayesian literature \citep{ghosal2000convergence, shen2001rates}.

\begin{definition}[Posterior contraction]
Consider a parametric model indexed by $\bm{\theta}$. Let $\{\mc{D}_n\}_{n \ge 1}$ be a sequence of data generations according to some true parameter $\bm{\theta}^\star$. Let $\bm{\gamma}(\bm{\theta})$ be a function of $\bm{\theta}$. Let $\ell(\bm{\gamma}(\bm{\theta}), \bm{\gamma}^\star)$ be a loss function between the estimate $\bm{\gamma}(\bm{\theta})$ and the parameter $\bm{\gamma}^\star$. A sequence of posterior distributions (random measures) $\{\pi(\bm{\theta} | \mc{D}_n)\}_{n \ge 1}$ is said to achieve convergence rate $\epsilon_n$ of estimation error $\ell(\bm{\gamma}(\bm{\theta}), \bm{\gamma}^\star)$ if
$$\pi(\ell(\bm{\gamma}(\bm{\theta}), \bm{\gamma}^\star) \ge M\epsilon_n|\mc{D}_n) \to 0$$
in $\mbb{P}_{\bm{\theta}^\star}$-probability as $n \to \infty$ for some constant $M > 0$.
\end{definition}

Specifically in the factor-adjusted regression model \eqref{model4} with covariates hidden in \eqref{model3}, we consider
$$\mc{D}_n = (\mbf{X}, \mbf{Y}), ~~~\bm{\theta} = (\mbf{B}, \sigma, \bm{\alpha}, \bm{\beta}),~~~\bm{\gamma}(\bm{\theta}) = {\bm{\alpha} \choose \bm{\beta}},~~~ \bm{\gamma}^\star = {\mbf{H}\bm{\alpha}^\star \choose \bm{\beta}^\star},$$
where $\mbf{H}$ is introduced by Assumption \ref{asm:3}, and want to show that $\widehat{\pi}(\sigma^2, \bm{\alpha}, \bm{\beta}|\mbf{X}, \mbf{Y})$ achieves the contraction rate $\epsilon_n = \sqrt{s\log p/n}$ of $\ell_2$ estimation error
$$\ell(\bm{\gamma}(\bm{\theta}), \bm{\gamma}^\star) = \Vert \bm{\gamma}(\bm{\theta}) - \bm{\gamma}^\star \Vert = \left\Vert{\bm{\alpha} \choose \bm{\beta}} - {\mbf{H}\bm{\alpha}^\star \choose \bm{\beta}^\star}\right\Vert.$$
As noted on Assumption \ref{asm:3}, $\widehat{\mbf{F}}$ approximates $\mbf{F}$ in the sense that they have almost the same column space and $\widehat{\mbf{F}}\mbf{H} \approx \mbf{F}$ element-wisely for some nearly orthogonal transformation matrix $\mbf{H}$. Thus the pseudo-posterior distribution would concentrate around $\bm{\alpha} \approx \mbf{H}\bm{\alpha}^\star$ such that $\widehat{\mbf{F}}\bm{\alpha} \approx \widehat{\mbf{F}}\mbf{H}\bm{\alpha}^\star \approx \mbf{F}\bm{\alpha}^\star$.

\subsection{Results}
This subsection presents the main results of the paper. Recall that $\epsilon_n = \sqrt{s\log p/n}$. Let
$$A(\sigma', \bm{\alpha}', \bm{\beta}', M_0, M_1, M_2, \epsilon_n) = \left\{(\sigma, \bm{\alpha}, \bm{\beta}):
\begin{split}
& |\xi \setminus \xi'|\le M_0 s,\\
& \frac{\sigma^2}{\sigma'^2} \in \left(\frac{1-M_1\epsilon_n}{1+M_1\epsilon_n}, \frac{1+M_1\epsilon_n}{1-M_1\epsilon_n}\right),\\ & \left\Vert{\bm{\alpha} \choose \bm{\beta}} - {\bm{\alpha}' \choose \bm{\beta}'}\right\Vert \le \sigma' M_2 \epsilon_n.
\end{split}\right\},$$
where $M_0, M_1, M_2$ are constants, $\xi$ and $\xi'$ are supports of $\bm{\beta}$ and $\bm{\beta}'$, respectively, and $|\xi \setminus\xi'|$ is the cardinality of the set difference of $\xi'$ and $\xi$.

\begin{theorem} \label{thm:1}
Let $\mbb{P}^\star = \mbb{P}_{(\mbf{B}, \sigma^\star, \bm{\alpha}^\star, \bm{\beta}^\star)}$ denote the probability measure under the true parameters. Under Assumptions \ref{asm:1}-\ref{asm:4}, the following statements hold.
\begin{enumerate}[label=(\alph*)]
\item (estimation error rate) There exist constants $M_0, M_1, M_2$ and $C_1$ such that
$$\mbb{P}^\star \left( \widehat{\pi}\left( A^c(\sigma^\star, \mbf{H}\bm{\alpha}^\star, \bm{\beta}^\star, M_0, M_1, M_2, \epsilon_n) | \mbf{X}, \mbf{Y}\right) \ge e^{-C_1s\log p}\right) \to 0$$
as $n \to \infty$.
\item (prediction error rate) There exist constants $M_3$ and $C_2$ such that
$$\mbb{P}^\star \left( \widehat{\pi}\left( \Vert (\widehat{\mbf{F}}\bm{\alpha} \!+\! \widehat{\mbf{U}}\bm{\beta}) \!-\! (\mbf{F}\bm{\alpha}^\star \!+\! \mbf{U}\bm{\beta}^\star) \Vert \!\ge\! \sigma^\star M_3 \sqrt{n} \epsilon_n| \mbf{X}, \mbf{Y}\right) \!\ge\! e^{-C_2s\log p}\right) \to 0$$
as $n \to \infty$.
\item (model selection consistency) If $\min_{j \in \xi^\star} |\beta^\star_j| \succ \epsilon_n$ then there exist constants $M_0, M_1, M_2$ and $C_3$ such that
$$\mbb{P}^\star \left( \widehat{\pi}\left( A^c(\sigma^\star, \mbf{H}\bm{\alpha}^\star, \bm{\beta}^\star, M_0, M_1, M_2, \epsilon_n) \cup \{\xi \not \supseteq \xi^\star\} | \mbf{X}, \mbf{Y}\right) \ge e^{-C_3s\log p}\right) \to 0$$
as $n \to \infty$. It follows that
\begin{align*}
\mbb{P}^\star \left(\widehat{\pi}\left( \{ |\xi \setminus \xi^\star| \le M_0s, \xi \supseteq \xi^\star\}^c| \mbf{X}, \mbf{Y}\right) \ge e^{-C_3s\log p}\right) &\to 0\\
\mbb{P}^\star \left(\widehat{\pi}\left( \left. \left\{j: |\beta_j| \ge \sigma\sqrt{|\xi|\log p/n}\right\} \ne \xi^\star \right| \mbf{X}, \mbf{Y}\right) \ge e^{-C_3s\log p}\right) &\to 0
\end{align*}
as $n \to \infty$.
\end{enumerate}
\end{theorem}

Part (a) establishes the convergence rate $\epsilon_n$ of the $\ell_2$-estimation error of $\bm{\alpha}^\star$ (up to a nearly orthogonal transformation $\mbf{H}$) and $\bm{\beta}^\star$, the adaptivity to the unknown sparsity $s$, and the adaptivity to the unknown standard deviation $\sigma^\star$.

Part (b) shows that $\widehat{\mbf{Y}} = \widehat{\mbf{F}}\bm{\alpha} + \widehat{\mbf{U}}\bm{\beta}$ predicts the conditional mean $\mbb{E}[\mbf{Y}|\mbf{F},\mbf{U}] = \mbf{F}\bm{\alpha}^\star \!+\! \mbf{U}\bm{\beta}^\star$ with mean squared error $ \bigOp(\epsilon_n)$ for each single datum instance on average.

The first implication in Part (c) asserts that the pseudo-posterior distribution will select all variables in $\xi^\star$ and at most $M_0s$ other variables, with high probability. In simulation experiments, we observe that the pseudo-posterior distribution overestimates the true support size $s = |\xi^\star|$ by less than $5\%$. The second implication asserts that
$$\widehat{\pi}\left( \left. \left\{j: |\beta_j| \ge \sigma\sqrt{|\xi|\log p/n}\right\} = \xi^\star \right| \mbf{X}, \mbf{Y}\right) \to 1$$
in probability as $n \to \infty$, and therefore provides a variable selection rule. Simply speaking, we can consistently select the true model $\xi^\star$ by thresholding the running coefficients $\beta_j$ at $\sigma\sqrt{|\xi|\log p/n}$. In simulation experiments, the majority of pseudo-posterior samples of parameters hit the true model correctly even if the thresholding rule is not used.

The additional condition that $\min_{j \in \xi^\star} |\beta^\star_j| \succ \epsilon_n$ in part (c) is called ``beta-min condition'' in the literature on Bayesian sparse regression \citep{castillo2015bayesian,song2017nearly}. \citet{narisetty2014bayesian} use another identifiability condition to achieve the model selection consistency. Their condition can be shown slightly stronger than the beta-min condition in presence of the minimum sparse eigenvalue condition. To see this point, one can compare their Condition 4.4 to our equation (10) in the proof of Lemma \ref{lem:test}, part(d).

\section{Factor Model Estimation} \label{sec:4}

This section verifies Assumption \ref{asm:3}, which concerns the estimation errors of factor models under standard assumptions. Following \cite{bickel2008covariance}, we define a uniformity class of positive semi-definite matrices as follows
\begin{align*}
\mathcal{S}_q^+ &= \left\{\mbf{\Sigma} \ge 0: \max_{1 \leq j \leq p} \sum_i |\mbf{\Sigma}_{ij}|^q < m_q(p), \Vert \mbf{\Sigma} \Vert_{\max} < C_0\right\}, ~~~\text{for}~0 \le q < 1,\\
\mathcal{S}_1^+ &= \left\{\mbf{\Sigma} \ge 0: \max_{1 \leq j \leq p} \sum_i |\mbf{\Sigma}_{ij}| < m_1(p) \right\}.
\end{align*}

\begin{assumption}\label{asm:5}
$\{(\bm{f}_i, \bm{u}_i)\}_{1 \le i \le n}$ are identically (not necessarily independently) distributed as $(\bm{f}, \bm{u})$. $\mbb{E}\bm{f} = \mbf{0}$, $\mbb{E}\bm{u} = \mbf{0}$; $\Cov(\bm{f}) = \mbf{I}$, $\Cov(\bm{u}) = \mbf{\Sigma} \in \mc{S}_q^+$ with $m_q(p) = \smallo(\log p)$ for some $0 \le q \le 1$, and $\Cov(\bm{f}, \bm{u}) = \mbf{0}$.
\end{assumption}

\begin{assumption} \label{asm:6}
All entries in the loading matrix $\mbf{B}$ are uniformly bounded, i.e., $\Vert \mbf{B} \Vert_{\max} = \bigO(1)$, and all the eigenvalues of $\mbf{B}^\T\mbf{B}/p$ is strictly bounded away from $0$ and $\infty$.
\end{assumption}

\begin{assumption} \label{asm:7}
The sample covariance matrices of $\mbf{F}$ and $\mbf{U}$ converge to the true covariance matrices at rate $\sqrt{\log p/n}$ in the element-wise maximum norm.
\begin{equation*}
\begin{split}
\Vert \mbf{F}^\T\mbf{F}/n - \mbf{I}\Vert_{\max} &= \bigOp(\sqrt{\log p / n}),\\
\Vert \mbf{U}^\T\mbf{U}/n - \mbf{\Sigma}\Vert_{\max} &= \bigOp(\sqrt{\log p / n}),\\
\Vert \mbf{F}^\T\mbf{U}/n \Vert_{\max} &= \bigOp(\sqrt{\log p / n}).
\end{split}
\end{equation*}
\end{assumption}

In Assumption \ref{asm:5}, $\Cov(\bm{f}) = \mbf{I}$ is made to avoid the non-identifiability issue of $\mbf{B}$ and $\bm{f}$.
If rows $\bm{b}_j, j=1,\dots,p$ of $\mbf{B}$ are $p$ i.i.d. copies of some $k$-dimensional distribution then $\mbf{B}^\T\mbf{B}/p$ converges almost surely to $\Cov(\bm{b}_j)$ as $p \to \infty$ and Assumption \ref{asm:6} holds when $\Cov(\bm{b}_j)$ has eigenvalues bounded away from 0 and $\infty$. Assumptions \ref{asm:5}-\ref{asm:6} together characterize the ``low-rank plus sparse'' structure of the covariance matrix of $\bm{x} = \mbf{B}\bm{f}+\bm{u}$. That is,
$$\Cov(\bm{x}) = \mbf{B}\mbf{B}^\T + \mbf{\Sigma},$$
where the first part $ \mbf{B}\mbf{B}^\T$ is of low rank $k$, and the second part is sparse in the sense that the quantity $\max_{1 \leq j \leq p} \sum_i |\mbf{\Sigma}_{ij}|^q$ for some $q \in [0,1]$ is $\smallo(\log p)$. This decomposition has a ``spike plus non-spike'' structural interpretation as well:  the smallest non-zero eigenvalue of $\mbf{B}\mbf{B}^\T$ is of order $p$, while the largest eigenvalue of $\mbf{\Sigma}$ is of order $\smallo(\log p)$. This eigen-gap plays the key role in estimating $\mbf{F}$ and $\mbf{U}$.

Assumption \ref{asm:7} requires that the sample covariance $\mbf{F}^\T\mbf{F}/n$, $\mbf{U}^\T\mbf{U}/n$ and $\mbf{F}^\T\mbf{U}/n$ converge to their ideal counterparts at an appropriate rate. \citet{kneip2011factor} provided sufficient conditions for it to hold in case that $\{(\bm{f}_i,\bm{u}_i)\}_{1 \le i \le n}$ are i.i.d.. \citet{fan2013large} established the same rate for stationary and weakly-correlated time-series. Our recent work on the concentration inequalities for general Markov chains \citep{jiang2018bernstein} can verify this assumption in case that $\{(\bm{f}_i, \bm{u}_i)\}_{1 \le i \le n}$ are functionals of ergodic Markov chains.

Next theorem summarizes the theoretical results on factor model estimation under Assumptions \ref{asm:5}-\ref{asm:7}. Part (b) of this theorem bounds the difference between column spaces of $\widehat{\mbf{F}}$ and $\mbf{F}$ in terms of principal angles, which is novel from the previous theory in the literature \citep{fan2013large} and may be of independent interest. Parts (c) and (d), which are immediate corollaries of part (b), derive Assumption \ref{asm:3}.

\begin{definition}
The principal angles between two linear spaces spanned by orthonormal column vectors of $\widehat{\mbf{\Psi}}_{n \times k}$ and $\widetilde{\mbf{\Psi}}_{n \times k}$ are defined as
$$\angle (\widehat{\mbf{\Psi}}, \widetilde{\mbf{\Psi}}) = (\arccos(d_1), \dots, \arccos(d_k))^\T,$$
where $d_1,\dots,d_k \in [0,1]$ are the singular values of $\widehat{\mbf{\Psi}}^\T\widetilde{\mbf{\Psi}}$ or $\widetilde{\mbf{\Psi}}^\T\widehat{\mbf{\Psi}}$.
\end{definition}

\begin{theorem} \label{theorem: factor model}
Let $\widetilde{\mbf{F}}$ consist of $\sqrt{n}$-scaled left singular vectors of $\mbf{F}$, which are orthonormal vectors spanning the column space of $\mbf{F}$. Under Assumptions \ref{asm:5}-\ref{asm:7}, the following statements hold.
\begin{enumerate}[label=(\alph*)]
\item Eigenvalue recovery:
$$\Vert \widehat{\mbf{\Lambda}} - \mbf{\Lambda} \Vert_{\max}/p = \bigOp(\sqrt{\log p/n}),~~~\max_{k+1\leq k \leq n} |\widehat{\lambda}_j|/p = \bigOp(\sqrt{\log p/n}).$$
\item Eigenspace recovery:
$$\Vert \sin \angle (\widehat{\mbf{F}}/\sqrt{n}, \widetilde{\mbf{F}}/\sqrt{n}) \Vert = \bigOp(\sqrt{\log p/n}).$$
\item Common factor recovery:
$$\Vert \widehat{\mbf{F}}\mbf{H} - \mbf{F}\Vert_\text{F} = \bigOp(\sqrt{\log p}),$$
for some nearly orthogonal matrix $\mbf{H}_{k \times k}$ with $\Vert \mbf{H}^\T\mbf{H} - \mbf{I} \Vert = \bigOp(\sqrt{\log p/n})$ and $\Vert \mbf{H}\mbf{H}^\T - \mbf{I} \Vert = \bigOp(\sqrt{\log p/n})$.
\item Idiosyncratic component recovery:
$$\max_{1 \leq j \leq p} \Vert \widehat{\mbf{U}}_j - \mbf{U}_j\Vert = \bigOp(\sqrt{\log p}).$$
\end{enumerate}
\end{theorem}

\section{Simulation Experiments}\label{sec:5}
This section reports simulation results. As a basic case, we set $(n,p,s,k) = (200,500,5,3)$, and generate $\bm{f}_i \overset{i.i.d.}{\sim} \mc{N}(\mbf{0},\mbf{I}_{k \times k})$, $\bm{u}_i \overset{i.i.d.}{\sim} \mc{N}(\mbf{0},\mbf{I}_{p \times p})$ and $\bm{b}_j \overset{i.i.d.}{\sim} \text{Uniform}[-1,+1]^k$. We set true parameters $\bm{\alpha}^\star = (0.8,1.0,1.2)$, $\xi^\star = \{1,2,3,4,5\}$, $\bm{\beta}^\star_{\xi^\star} = (0.3,0.3,0.3,0.3,0.3)^\T$, and $\sigma^\star = 0.5$.

For prior \eqref{prior}, we choose the inverse-gamma density $g$ with shape $1$ and scale $1$, the Gaussian density $h(z) = e^{-z^2/2}/\sqrt{2\pi}$ and set hyperparameters $s_0 = 1$ and $\tau_j = \Vert \widehat{\mbf{U}}_j \Vert/\sqrt{n}$. Starting from $(\sigma, \bm{\alpha}, \bm{\beta}) = (1.0, \mbf{0}, \mbf{0})$, we iterate a Gibbs sampler $T=20$ times and drop the first $T/2 = 10$ iterations as the burn-in period. The implementation details of the Gibbs sampler is given in the appendix.

The pseudo-posterior distribution are evaluated in terms of five measures. The posterior mean of $\bm{\beta}$ is compared to $\bm{\beta}^\star$ in terms of $\ell_2$ estimation error. The model selection rate and the sure screening rate are also computed. The former is the portion of the posterior samples that select the true model, i.e., $\xi = \xi^\star$, and the latter is the portion of the posterior samples that select all sparse coefficients, i.e., $\xi \supseteq \xi^\star$. To evaluate the adaptivity to unknown sparsity $s$, we report the average model size $|\xi|$. To evaluate the adaptivity to unknown standard deviation $\sigma^\star$, the posterior mean of $\sigma^2$ is compared to $\sigma^{\star 2}$ in terms of relative estimation error. These measures are evaluated over 100 replicates of the datasets, and their averages are reported.

For the comparison purpose, the factor-adjusted lasso method is implemented by using R package \textit{glmnet} \citep{friedman2010regularization}. The $\ell_1$-penalty hyperparameters of the lasso methods are tuned by 10-fold cross-validation. Since the generic Bayesian / lasso with $\mbf{X}$ as covariates can be seen as the factor-adjusted Bayesian / lasso with the underestimate $\widehat{k} = 0$ of $k = 3$, we also include them in the comparison.

\subsection{Comparison of four methods, and insensitivity to misestimates of $k$}
Table \ref{tab:1} summarizes the five measures of four methods in the basic case. Results show that the factor-adjusted Bayesian method outperforms the factor-adjusted lasso method in the tasks of $\bm{\beta}$ estimation and model selection. The poor performance of the factor-adjusted lasso method may partly result from the less satisfactory hyperparameter tuning procedure implemented in the R package \textit{glmnet}.

We feed the factor-adjusted methods with the various estimates $\widehat{k} = 3,6,9,12$, and observe that their performances are insensitive to the overestimate of $k$ (Table \ref{tab:1}). In case that there is no correlation among $\mbf{X}$, i.e., $k=0$, the factor-adjusted Bayesian method performs slightly worse than the generic Bayesian method (Table \ref{tab:2}).

\begin{table}[h!]
\center
\begin{tabular}{c | c |c c c| c}
\hline
Method & \thead{$\bm{\beta}$ estimation\\($\ell_2$ error)} & \thead{model \\ selection \\rate} & \thead{sure \\ screening \\ rate} & \thead{average \\ model \\ size} & \thead{$\sigma^2$ estimation\\ (relative error)}\\
\hline
generic Bayes, $\widehat{k} = 0$ & 1.536 & 0.0\% & 100.0\% & 15.92 & 5.940\\
Factor-adjusted Bayes, $\widehat{k} = 3$ & 0.124 & 88.2\% & 100.0\% & 5.13 & 2.057\\
Factor-adjusted Bayes, $\widehat{k} = 6$ & 0.124 & 87.0\% & 100.0\% & 5.14 & 2.058\\
Factor-adjusted Bayes, $\widehat{k} = 9$ & 0.130 & 86.4\% & 100.0\% & 5.14 & 2.057\\
Factor-adjusted Bayes, $\widehat{k} = 12$ & 0.133 & 86.1\% & 100.0\% & 5.14 & 2.069\\
\hline
generic lasso, $\widehat{k} = 0$ & 1.189 &  0\% & 100\% & 92.57 & 3.187\\
Factor-adjusted lasso, $\widehat{k} = 3$ & 0.460  & 27\%  & 100\% & 21.20 & 1.899\\
Factor-adjusted lasso, $\widehat{k} = 6$ & 0.463  & 25\% &  100\% & 21.56 & 1.865\\
Factor-adjusted lasso, $\widehat{k} = 9$ & 0.467  & 27\% & 100\% & 26.07 & 1.787\\
Factor-adjusted lasso, $\widehat{k} = 12$ & 0.466 & 24\% & 100\% & 29.21 & 1.657\\
\hline
\end{tabular}
\caption{Simulation results in the basic case with $k=3$.}
\label{tab:1}
\end{table}

\begin{table}[h!]
\center
\begin{tabular}{c | c |c c c| c}
\hline
Method & \thead{$\bm{\beta}$ estimation\\($\ell_2$ error)} & \thead{model \\ selection \\rate} & \thead{sure \\ screening \\ rate} & \thead{average \\ model \\ size} & \thead{$\sigma^2$ estimation\\ (relative error)}\\
\hline
generic Bayes, $\widehat{k} = 0$ & 0.092  & 90.5\%  & 100.0\%       & 5.10 & 0.881\\
Factor-adjusted Bayes, $\widehat{k} = 3$ & 0.095  & 87.9\% & 100.0\% & 5.13 & 0.906\\
Factor-adjusted Bayes, $\widehat{k} = 6$ & 0.097  & 88.4\% & 100.0\% & 5.12 & 0.914\\
Factor-adjusted Bayes, $\widehat{k} = 9$ & 0.098 & 88.5\% & 100.0\% & 5.12 & 0.932\\
Factor-adjusted Bayes, $\widehat{k} = 12$ & 0.102 & 88.9\% & 100.0\% & 5.12 & 0.968\\
\hline
generic lasso, $\widehat{k} = 0$ & 0.495 &  53\% & 100\% & 11.72 & 1.302\\
Factor-adjusted lasso, $\widehat{k} = 3$ & 0.498  & 61\%  & 100\% & 11.98 & 1.248\\
Factor-adjusted lasso, $\widehat{k} = 6$ & 0.500  & 56\% &  100\% & 13.28 & 1.279\\
Factor-adjusted lasso, $\widehat{k} = 9$ & 0.481  & 56\% & 100\% & 12.48 & 1.141\\
Factor-adjusted lasso, $\widehat{k} = 12$ & 0.487 & 58\% & 100\% & 13.62 & 1.114\\
\hline
\end{tabular}
\caption{Simulation results in no correlation case with $k=0$.}
\label{tab:2}
\end{table}

We emphasize that the meaning of the model selection rate for the Bayesian methods are slightly different from that for the frequentist methods. For example, 50\% model selection rate given by a frequentist method means that it select the true sparse model in 50 out of 100 replicates of the dataset. In contrast, 90\% model selection rate given by a Bayesian method means that every 9 of 10 posterior samples of parameters hit the true sparse model in a single replicate of the dataset on average. In the simulation experiments reported by Tables \ref{tab:1}-\ref{tab:2}, at least every 7 of 10 pseudo-posterior samples obtained by our method hit the true sparse model in each of 100 replicates of the dataset.

\subsection{Scalability as $n,p,s$ increase}
We vary the sample size $n$, the dimensionality $p$ and the sparsity $s$ in the basic case, and test the scalability of the proposed methodology.

In Figure \ref{fig:1}(a), we fix all parameters in the basic case but vary $n = 100$, $150$, $200$, $250$, $300$, $350$. In Figure \ref{fig:1}(b), we fix all parameters in the basic case but vary $p = 200$, $300$, $400$, $500$, $600$, $700$. In Figure \ref{fig:1}(c), we fix all parameters in the basic case but vary $s = 1$, $3$, $5$, $7$, $9$, $11$, $13$, $15$. For factor-adjusted methods, $\widehat{k} = k = 3$ are used.

We observe that our method outperforms the other three methods in terms of $\beta$ estimation error and model selection rate under each combination of $(n,p,s)$, and achieves comparable relative error of $\sigma^2$ to the factor-adjusted lasso method.

\begin{figure}[h!]
    \centering
    \begin{subfigure}[b]{\textwidth}
    	\includegraphics[width=\textwidth]{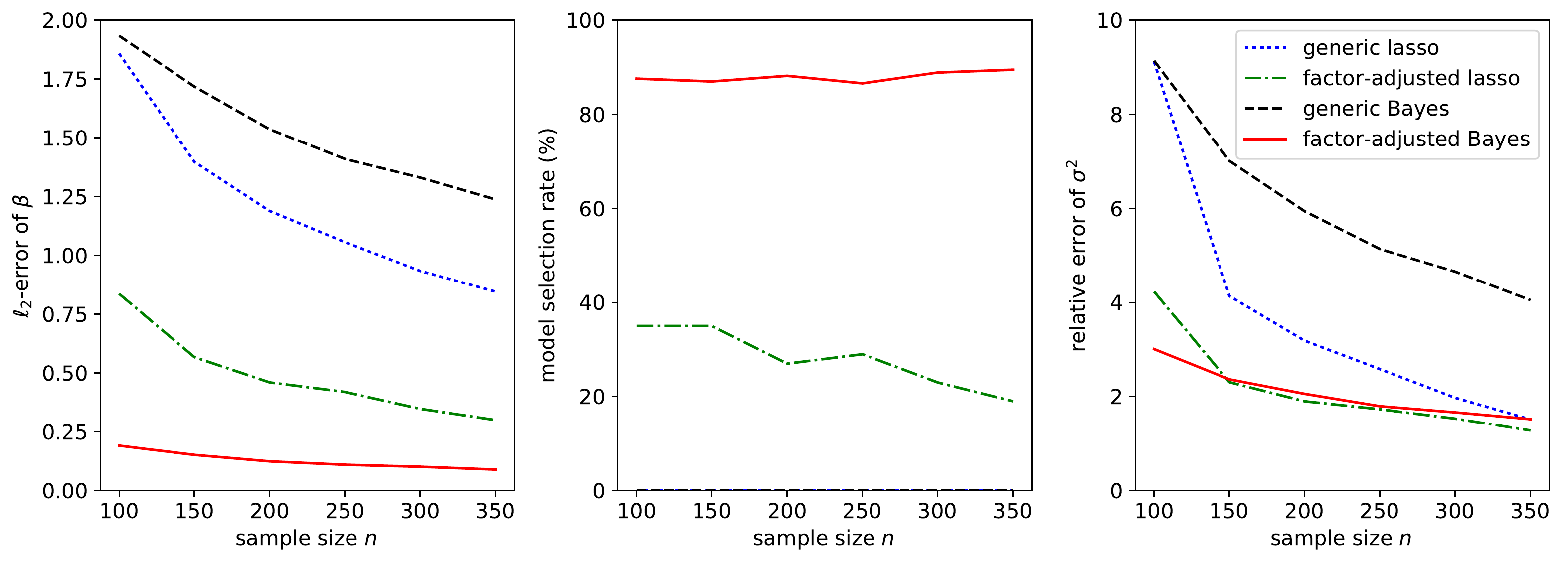}
    	\caption{}
    \end{subfigure}
    ~
    \begin{subfigure}{\textwidth}
    	\includegraphics[width=\textwidth]{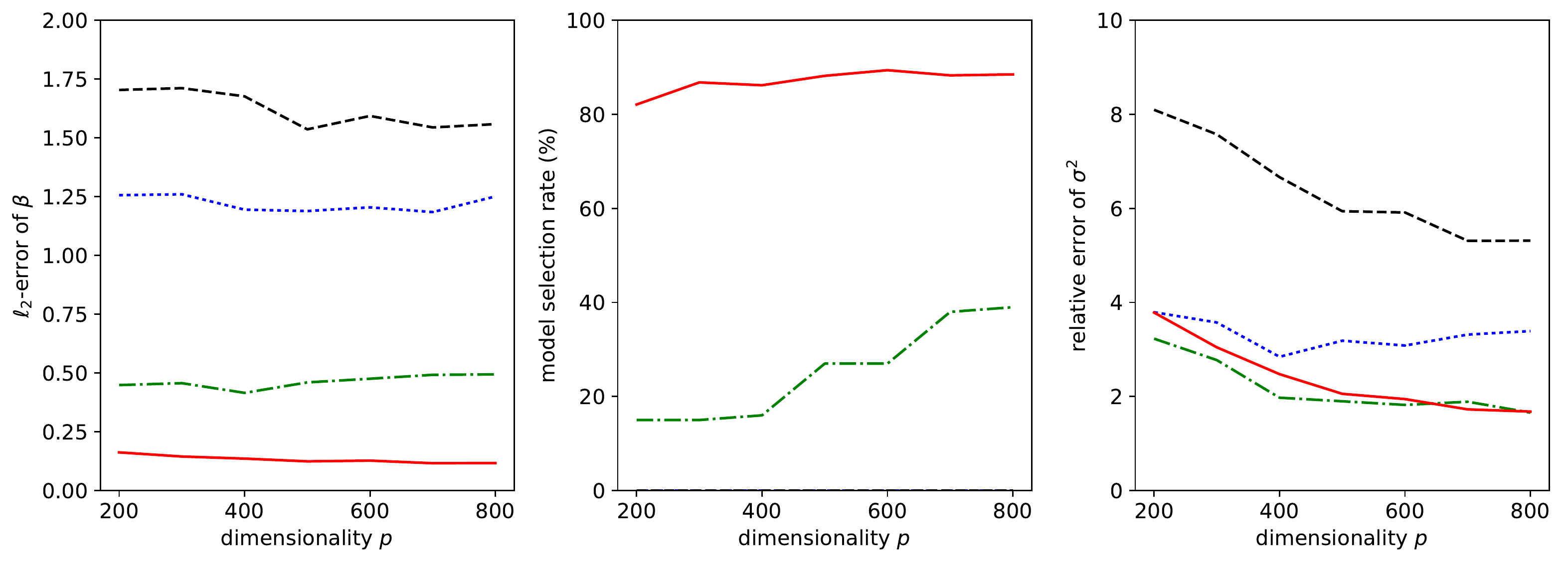}
    	\caption{}
    \end{subfigure}
    ~
    \begin{subfigure}{\textwidth}
    	\includegraphics[width=\textwidth]{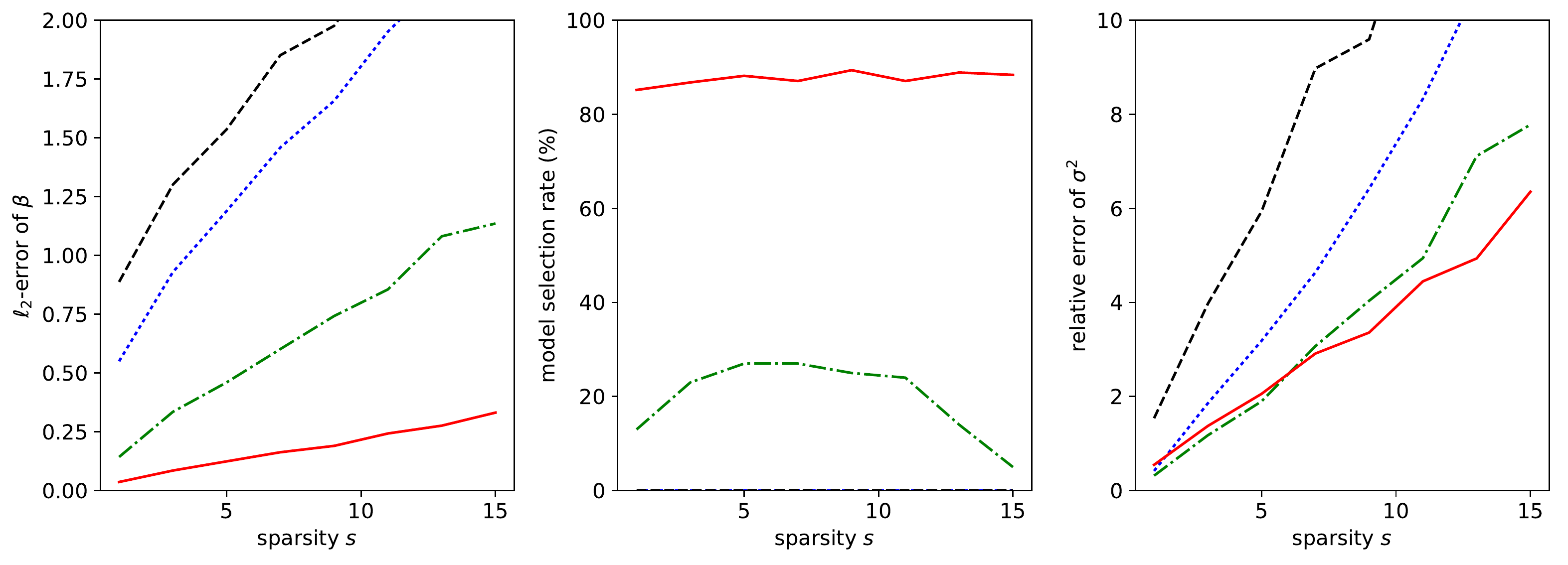}
    	\caption{}
    \end{subfigure}
    \caption{$\bm{\beta}$ estimation error (left), model selection rate (middle) and $\sigma^2$ estimation error (right) versus (a) sample size $n$, (b) dimensionality $p$ and (c) sparsity $s$. Factor-adjusted methods use $\widehat{k} = k = 3$.}
    \label{fig:1}
\end{figure}

\subsection{Estimating the standard regression model}

Recall that, when $\mbf{X}$ admits a factor structure \eqref{model3}, the standard regression model \eqref{model1} is a special case of the factor-adjusted regression model \eqref{model4} with the parameter constraint $\bm{\alpha} = \mbf{B}^\T\bm{\beta}$. We expect that the factor-adjusted Bayesian method solves model \eqref{model1} as well. To verify this expectation, we set $\bm{\alpha}^\star = \mbf{B}^\T\bm{\beta}^\star$ (or equivalently $\mbf{Y} = \mbf{X}\bm{\beta}^\star + \sigma^\star\bm{\varepsilon}$), and test four methods on simulation datasets.

\begin{table}[h!]
\center
\begin{tabular}{c | c |c c c| c}
\hline
Method & \thead{$\bm{\beta}$ estimation\\($\ell_2$ error)} & \thead{model \\ selection \\rate} & \thead{sure \\ screening \\ rate} & \thead{average \\ model \\ size} & \thead{$\sigma^2$ estimation\\ (relative error)}\\
\hline
generic Bayes, $\widehat{k} = 0$ & 0.070 & 91.5\%  & 100.0\%       & 5.09 & 0.913\\
Factor-adjusted Bayes, $\widehat{k} = 3$ & 0.090  & 91.1\% & 100.0\% & 5.09 & 1.690\\
Factor-adjusted Bayes, $\widehat{k} = 6$ & 0.091  & 90.7\% & 100.0\% & 5.10 & 1.715\\
Factor-adjusted Bayes, $\widehat{k} = 9$ & 0.093 & 90.4\% & 100.0\% & 5.10 & 1.733\\
Factor-adjusted Bayes, $\widehat{k} = 12$ & 0.095 & 89.8\% & 100.0\% & 5.11 & 1.763\\
\hline
generic lasso, $\widehat{k} = 0$ & 0.734 &  13\% & 100\% & 10.18 & 3.266\\
Factor-adjusted lasso, $\widehat{k} = 3$ & 0.454  & 53\%  & 100\% & 15.17 & 1.125\\
Factor-adjusted lasso, $\widehat{k} = 6$ & 0.471  & 57\% &  100\% & 14.96 & 1.193\\
Factor-adjusted lasso, $\widehat{k} = 9$ & 0.465  & 48\% & 100\% & 16.65 & 1.139\\
Factor-adjusted lasso, $\widehat{k} = 12$ & 0.492 & 55\% & 100\% & 18.06 & 1.213\\
\hline
\end{tabular}
\caption{experimental results on model \eqref{model1}.}
\label{tab:3}
\end{table}

We see that the factor-adjusted Bayesian method does solve model \eqref{model1}. Interestingly, while the factor adjustment added to the lasso method significantly increases the model selection rate from 13\% to roughly 50\%, the generic Bayesian method works comparably well or even better than the factor-adjusted Bayesian method. We will discuss this phenomenon in the discussion section.

\section{Predicting U.S. Bond Risk Premia} \label{sec:6}
This section applies our method to predict U.S. bond risk premia with a large panel of macroeconomic variables. The response variables are monthly U.S. bond risk premia with maturity of $m = 2, 3, 4, 5$ years spanning the period from January, 1964 to December, 2003 \citep{ludvigson2009macro}. The $m$-year bond risk premium at period $i+1$ is defined as the (log) holding return from buying an $m$-year bond at period $i$ and selling it as an $(m-1)$-year boud at period $i+1$, excessing the (log) return on one-year bond bought at period $i$. The covariates are $p = 131$ macroeconomic variables collected in the FRED-MD database \citep{mccracken2016fred} during the same period.

The $p = 131$ covariates over 480 months are strongly correlated. The scree plot of PCA of these covariates (Figure \ref{fig:2}) shows that the first principal component accounts for 55.9\% of the total variance, and that the first 5 principal components account for 89.7\% of the total variance.

\begin{figure}[ht!]
    \centering
    \includegraphics[width=0.8\textwidth]{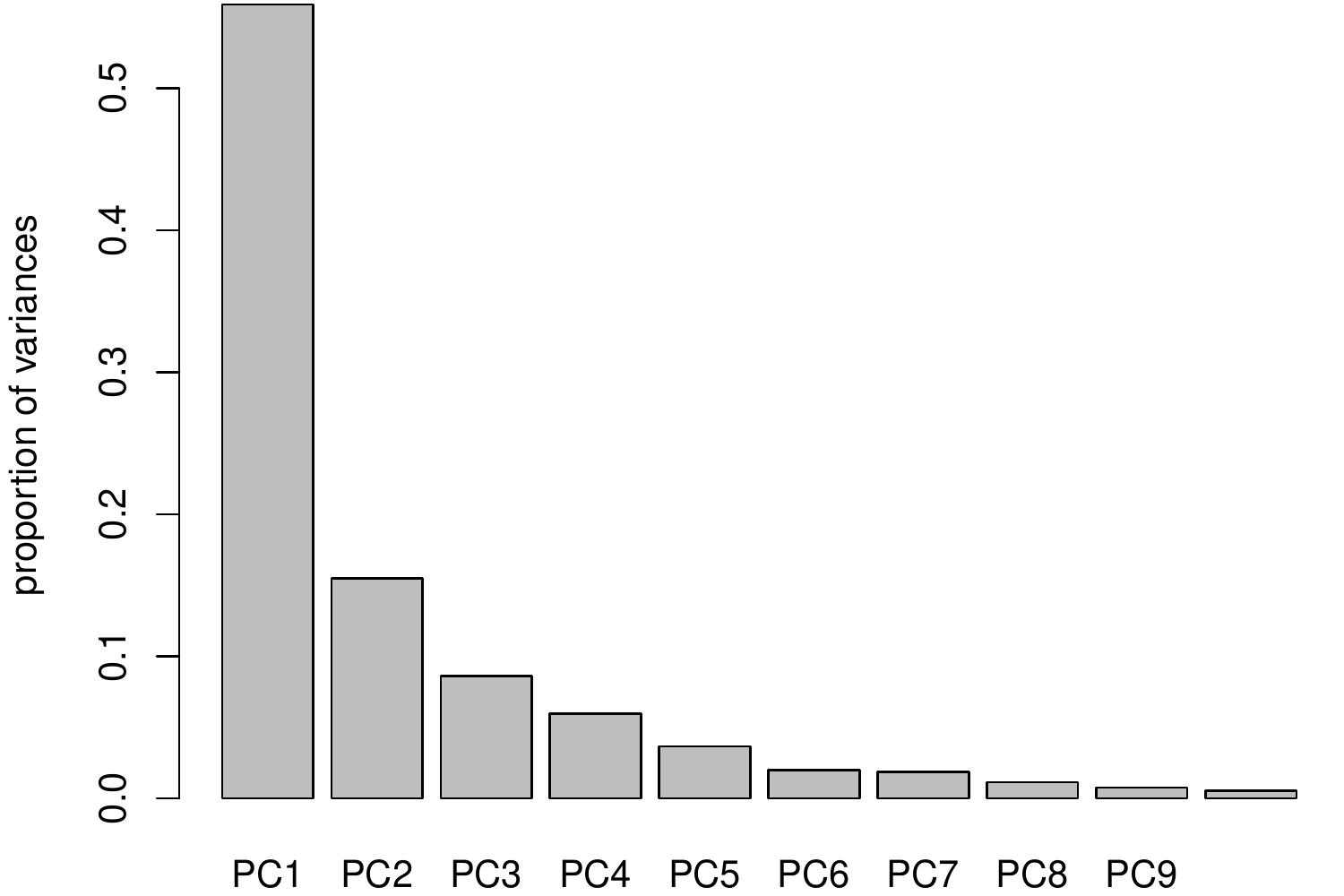}
    \caption{Proportion of variances explained by the first 10 principal components.}
    \label{fig:2}
\end{figure}

We consider the rolling window regression and next value prediction. Specifically, we regress a U.S. bond risk premium on the macroeconomic variables in the last month. For each time window of size $n=100$ ahead of month $t = n+2,\dots,480$, we fit
$$y_i = f(\bm{x}_{i-1}) + \sigma \varepsilon_i,~~~i = t-n, \dots, t-1,$$
and do out-of-sample prediction
$$\widehat{y}_t = \widehat{f}(\bm{x}_{t-1}).$$
The regression function $f$ is fitted as $\widehat{f}$ by one of the generic lasso method, the factor-adjusted lasso method, the generic Bayesian method, the factor-adjusted Bayesian method and the principal component regression (PCR) method \citep{ludvigson2009macro}. For the factor-adjusted methods, the number of common factors $k$ is estimated by \eqref{khat}. For the Bayesian methods, we set $s_0 = 10$ in prior \eqref{prior}. For PCR, the top eight principal components are included in the regression model in a similar vein to \citep{ludvigson2009macro}. The R package \textit{pls} \citep{wehrens2007pls} is used for implementation of PCR.

The prediction performance is evaluated by the out-of-sample $R^2$, which is computed as follows.
$$R^2 = 1 - \frac{\sum_{t=n+2}^{480} (\widehat{y}_t - y_t)^2}{\sum_{t=n+2}^{480} (\bar{y}_t - y_t)^2},$$
where $y_t$ is one of two-year, three-year, four-year and five-year U.S. bond risk premia, $\widehat{y}_t$ is the prediction of $y_t$ by one of five methods in comparison, and $\bar{y}_t$ is the average of $\{y_{t-n},\dots,y_{t-1}\}$.

Table \ref{tab:4} summarizes the out-of-sample $R^2$ five methods achieve on this task. Table \ref{tab:5} reports the average size of the sparse models they select. We observe that the factor-adjusted Bayesian method together with the factor-adjusted Bayesian method achieve higher out-of-sample $R^2$ than other methods. But the factor-adjusted Bayesian method select much sparser models than the factor-adjusted lasso method.

\begin{table}[h!]
\center
\begin{tabular}{c | c c c c }
\hline
Method & 2-yr bond & 3-yr bond & 4-yr bond & 5-yr bond\\
\hline
PCR & 0.646 & 0.603 & 0.568 & 0.540\\
generic Bayes & 0.765 & 0.734 & 0.722 & 0.696\\
factor-adjusted Bayes & 0.775 & 0.753 & 0.747 & 0.726\\
generic lasso & 0.719 & 0.717 & 0.701 & 0.688\\ 
factor-adjusted lasso & 0.766 & 0.764 & 0.746 & 0.719\\
\hline
\end{tabular}
\caption{Out-of-sample $R^2$ of five methods predicting U.S. bond risk premia.}
\label{tab:4}
\end{table}

\begin{table}[h!]
\center
\begin{tabular}{c | c c c c }
\hline
Method & 2-yr bond & 3-yr bond & 4-yr bond & 5-yr bond\\
\hline
generic Bayes & 12.97 & 12.97 & 13.13 & 13.05\\
factor-adjusted Bayes & 11.04 & 11.39 & 11.63 & 11.41\\
generic lasso & 24.06 & 24.25 & 25.62 & 25.71\\
factor-adjusted lasso & 34.46 & 35.12 & 36.91 & 36.57\\
\hline
\end{tabular}
\caption{The average size of sparse models four methods select.}
\label{tab:5}
\end{table}

\section{Discussion} \label{sec:7}
We propose a factor-adjusted regression model to handle the linear relationship between the response variable and possibly highly correlated covariates. We decompose the predictors into common factors and idiosyncratic components, where the common factors explain most of the variations, and assume all common factors but a small number of idiosyncratic components contribute to the response. The corresponding Bayesian methodology is then developed for estimating such a model. Theoretical results suggest that the proposed methodology can consistently estimate the factor-adjusted model and thus obtain consistent predictions, under an easily-to-hold sparse eigenvalue condition on the idiosyncratic components instead of the original covariates.

Our factor-adjusted model covers the standard linear model as a sub-model with the side constraint. Thus, our proposed methodology can easily handle the case when the standard linear regression model is assumed to be the underlying model. In simulation studies on the sub-model, we find that the factor adjustment greatly improves the performance of lasso, while the generic Bayesian sparse regression is comparable to the factor-adjusted Bayesian sparse regression in terms of estimation error and model selection rate (Table \ref{tab:3}). This suggests a fundamental difference between the frequentist sparse regression method and the Bayesian sparse regression method. Indeed, one can prove under Assumptions \ref{asm:1}-\ref{asm:2}, \ref{asm:5}-\ref{asm:7} that
\begin{align*}
\min_{\xi:|\xi \setminus \xi^\star| \le M_0s} \lambda_{\min}\left(\mbf{X}^\T_\xi \mbf{X}_\xi/n\right) &\ge \kappa_0 - \bigOp(s^2\log p/n)\\
\lambda_{\max}\left(\mbf{X}^\T_{\xi^\star} \mbf{X}_{\xi^\star}/n\right) &= \bigOp(s).
\end{align*}
If $s = \bigO(1)$, these two terms are of constant order, and then a similar argument to the proof of Theorem \ref{thm:1} would establish the convergence and model selection consistency of the generic Bayesian regression on standard regression model \eqref{model1}.

Nonetheless, we recommend the factor-adjusted Bayesian regression on model \eqref{model4} over the generic Bayesian regression on model \eqref{model1} for three reasons. First, the theoretical analyses of the former allow $s$ to grow with $n$, in contrast the latter requires fixed $s$. Second, model \eqref{model4} provides more flexibility than its sub-model \eqref{model1} in the regression analyses and would potentially explore more explanatory power from the data. On the real dataset of U.S. bond risk premia, the factor adjusted Bayesian regression achieves 1.0\%-3.0\% more out-of-sample $R^2$ with one or two less variables (Tables \ref{tab:4}-\ref{tab:5}). Third, in the no correlation case (although it is unlike the case in practice), the factor-adjusted Bayesian regression pays a negligible price for model misspecification (Table \ref{tab:2}).

\section*{Acknowledgement}
We would like to thank Yun Yang for helpful discussions. 
\bibliographystyle{ims}
\bibliography{ref}

\appendix
\appendixpage
\renewcommand{\thelemma}{\Alph{section}\arabic{lemma}}

\noindent \textbf{Notation}. We summarize all notation used in the appendices. Some of them may have been defined in the main body of the paper.

Let $\mbf{I}_{m \times m}$ be the identity matrix of dimension $m \times m$, $\mbf{0}_m$ the column vector of $m$ zeros, $\mbf{0}_{m_1 \times m_2}$ the matrix of $m_1 \times m_2$ zeros. If the dimension of an identity or zero matrix is clear in the context, we omit the subscripts. We write $\diag(a_1,a_2,\dots,a_m)$ for a diagonal matrix of elements $a_1,a_2,\dots,a_m$. For a symmetric matrix $\mbf{A}$, we write its trace as $\trace(\mbf{A})$. For a positive semi-definite matrix $\mbf{A}$, we write its largest eigenvalue as $\lambda_{\max}(\mbf{A})$ and its smallest eigenvalue as $\lambda_{\min}(\mbf{A})$. For two squared matrices $\mbf{A},\mbf{B}$ of the same dimension, we write $\mbf{A} \ge \mbf{B}$ (or $\mbf{B} \le \mbf{A}$) if $\mbf{A} - \mbf{B}$ is positive semidefinite. For a matrix $\mbf{A}_{m_1 \times m_2} = [a_{ij}]_{1 \le i \le m_1, 1 \le j \le m_2}$, we write $\mbf{A}_j$ to denote its $j$-th column. For a index set $\xi \subseteq \{1,\dots,m_2\}$, $\mbf{A}_\xi = [\mbf{A}_j: j \in \xi]$ is the sub-matrix of $\mbf{A}$ assembling the columns indexed by $\xi$. For a matrix $\mbf{A}$ of full column rank, write $\mbf{A}^\dagger = (\mbf{A}^\T\mbf{A})^{-1}\mbf{A}^\T$ as its left pseudo-inverse.

For a vector $\bm{v}$, let $\Vert \bm{v} \Vert_q$ denote its $\ell_q$ norm. If $q = 2$, we omit the subscript and write $\Vert \bm{v} \Vert$ for simplicity. For a matrix $\mbf{A}_{m_1 \times m_2} = [a_{ij}]_{1 \le i \le m_1, 1 \le j \le m_2}$, let
$$\Vert \mbf{A} \Vert_q = \sup\{ \Vert \mbf{A} \bm{v} \Vert_q: \bm{v} \in \mbb{R}^{m_2}, \Vert \bm{v} \Vert_q = 1\}$$
be the $\ell_q$ operator norm induced by $\ell_q$ vector norm. If $q = 2$, we omit the subscript and write $\Vert \mbf{A} \Vert$ for simplicity. $\Vert \mbf{A} \Vert$ is called the operator norm of $\mbf{A}$ in short. Evidently, it is equal to the largest singular value of $\mbf{A}$. Let $\Vert \mbf{A}\Vert_{\max} = \max_{i,j} |a_{ij}|$ be the element-wise maximum norm of $\mbf{A}$, and let $\Vert \mbf{A}\Vert_\text{F} = \sqrt{\sum_{i,j} a_{ij}^2} = \sqrt{\trace(\mbf{A}^\T\mbf{A})}$ be the Frobenius norm of $\mbf{A}$.

For some constant $c$, we write $c + \smallo(1)$ to denote a constant arbitrarily large than $c$, and $c - \smallo(1)$ to represent a constant arbitrarily less than $c$. For two positive sequences $a_n, b_n$, $a_n \asymp b_n$ means $\lim_{n \to \infty} a_n/b_n = c$ for some constant $c > 0$; $a_n \succcurlyeq b_n$ (or $b_n \preccurlyeq a_n$) means $b_n = \bigO(a_n)$; $a_n \succ b_n$ (or $b_n \prec a_n$) means $b_n = \smallo(a_n)$, and $a_n \gtrsim b_n$ (or $b_n \lesssim a_n$) means that $a_n \ge b_n$ for sufficiently large $n$. For a sequence of non-negative random variables $Z_n$ and a sequence of positive numbers $a_n$, $Z_n = \bigOp(a_n)$ means $Z_n/a_n \le c$ for some constant $c$ with probability approaching $1$ as $n \to \infty$, $Z_n = \smallop(a_n)$ means $|Z_n/a_n| \le c$ for any constant $c$ with probability approaching $1$ (i.e. $Z_n/a_n \to 0$ in probability) as $n \to \infty$.

\section{Technical Proofs for Bayesian Sparse Regression}
\setcounter{lemma}{0}
This appendix collects technical proofs for Theorem \ref{thm:1}.

\subsection{Proof of Theorem \ref{thm:1}}
The proofs of three parts use the same techniques and have a similar structure. First, we observe under Assumptions \ref{asm:1}-\ref{asm:3} that, for some constant $C_4$ and any constant $M_0$,
\begin{equation} \label{conditioning}
\begin{split}
\max_{j=1}^k \Vert (\widehat{\mbf{F}}\mbf{H})_j - \mbf{F}_j \Vert &\le C_4\sqrt{\log p},\\
\max_{j=1}^p \Vert \widehat{\mbf{U}}_j - \mbf{U}_j \Vert &\le C_4\sqrt{\log p},\\
\Vert \mbf{H}^\T \mbf{H} - \mbf{I}\Vert &\le C_4\sqrt{\log p/n}, ~~~ \Vert \mbf{H}\mbf{H}^\T - \mbf{I}\Vert \le C_4\sqrt{\log p/n}\\
\min_{\xi:~|\xi| \le (M_0+1)s} \lambda_{\min}(\widehat{\mbf{U}}_\xi^\T\widehat{\mbf{U}}_\xi/n) &\ge \widehat{\kappa}_0 := \kappa_0/2,\\
\lambda_{\max}(\widehat{\mbf{U}}_{\xi^\star}^\T\widehat{\mbf{U}}_{\xi^\star}/n) &\le \widehat{\kappa}_1 := 2\kappa_1
\end{split}
\end{equation}
hold with probability approaching $1$. The first three claimed bounds are directly taken from Assumption \ref{asm:3}. The last two claimed bound follow from Weyl's inequality. For any model $\xi$ of size at most $(M_0+1)s$, the singular values of $\widehat{\mbf{U}}_\xi$ differ from those of $\mbf{U}_\xi$ by at most
\begin{align*}
&~~~\max_{\xi: ~|\xi| \le (M_0+1)s} \Vert \widehat{\mbf{U}}_\xi - \mbf{U}_\xi\Vert \le \max_{\xi: ~|\xi| \le (M_0+1)s} \Vert \widehat{\mbf{U}}_\xi - \mbf{U}_\xi\Vert_\text{F}\\
&\le \sqrt{(M_0+1)s \max_{j=1}^p \Vert \widehat{\mbf{U}}_j - \mbf{U}_j\Vert^2} \le C_4\sqrt{(M_0+1)s\log p}.
\end{align*}
This implies that
\begin{align*}
\min_{\xi:~|\xi| \le (M_0+1)s} & \lambda_{\min}(\widehat{\mbf{U}}_\xi^\T\widehat{\mbf{U}}_\xi/n) \ge \left(\sqrt{\kappa_0} - C_4\sqrt{(M_0+1)s\log p/n}\right)^2 \gtrsim \kappa_0/2,\\
& \lambda_{\max}(\widehat{\mbf{U}}_{\xi^\star}^\T\widehat{\mbf{U}}_{\xi^\star}/n) \le \left(\sqrt{\kappa_1} + C_4\sqrt{(M_0+1)s\log p/n}\right)^2 \lesssim 2\kappa_1.
\end{align*}

We thereafter need to show that the conditional probabilities of
\begin{align*}
\text{Event (a):}~& \widehat{\pi}\left( A^c(\sigma^\star, \mbf{H}\bm{\alpha}^\star, \bm{\beta}^\star, M_0, M_1, M_2,\epsilon_n) | \mbf{X}, \mbf{Y}\right) \ge e^{-C_1s\log p}\\
\text{Event (b):}~& \widehat{\pi}\left( \Vert (\widehat{\mbf{F}}\bm{\alpha} +\widehat{\mbf{U}}\bm{\beta}) - (\mbf{F}\bm{\alpha}^\star + \mbf{U}\bm{\beta}^\star) \Vert \!\ge\! \sigma^\star M_3 \sqrt{n} \epsilon_n| \mbf{X}, \mbf{Y}\right) \ge e^{-C_2s\log p}\\
\text{Event (c):}~& \widehat{\pi}\left( A^c(\sigma^\star, \mbf{H}\bm{\alpha}^\star, \bm{\beta}^\star, M_0, M_1, M_2, \epsilon_n) \cup \{\xi \not \supseteq \xi^\star\} | \mbf{X}, \mbf{Y}\right) \ge e^{-C_3s\log p},
\end{align*}
given any realization of $(\mbf{F}, \mbf{U},  \mbf{X}, \widehat{\mbf{F}}, \widehat{\mbf{U}}, \mbf{H})$ satisfying \eqref{conditioning}, vanish $n \to \infty$. Recall that
$$A(\sigma', \bm{\alpha}', \bm{\beta}', M_0, M_1, M_2, \epsilon_n) = \left\{ (\sigma^2, \bm{\alpha}, \bm{\beta}):
\begin{split}
& |\xi \setminus \xi'| \le M_0s,\\
& \frac{\sigma^2}{\sigma'^2} \in \left(\frac{1-M_1\epsilon_n}{1+M_1\epsilon_n}, \frac{1+M_1\epsilon_n}{1-M_1\epsilon_n}\right),\\
& \left\Vert {\bm{\alpha} \choose \bm{\beta}} - {\bm{\alpha}' \choose \bm{\beta}'} \right\Vert \le \sigma M_2 \epsilon_n
\end{split} \right\},$$
where $\xi$ and $\xi'$ are supports of $\bm{\beta}$ and $\bm{\beta}'$, respectively, and $\epsilon_n = \sqrt{s\log p/n}$.

Consider the conditional probability of event (a). Since $\widehat{\pi}(\sigma, \bm{\alpha}, \bm{\beta}|\mbf{X}, \mbf{Y})$ depends on $\mbf{X}$ through $\widehat{\mbf{F}}$ and $\widehat{\mbf{U}}$, we have
\begin{align*}
&~~~\mbb{P}_{(\mbf{B}, \sigma^\star, \bm{\alpha}^\star, \bm{\beta}^\star)} \left( \left. \widehat{\pi}\left( A^c(\sigma^\star, \mbf{H}\bm{\alpha}^\star, \bm{\beta}^\star, M_0, M_1, M_2, \epsilon_n) | \mbf{X}, \mbf{Y}\right) \ge e^{-C_1s\log p} \right|\mbf{F}, \mbf{U},  \mbf{X}, \widehat{\mbf{F}}, \widehat{\mbf{U}}, \mbf{H}\right)\\
&= \mbb{P}_{(\mbf{B}, \sigma^\star, \bm{\alpha}^\star, \bm{\beta}^\star)} \left( \left. \widehat{\pi}( A^c(\sigma^\star, \mbf{H}\bm{\alpha}^\star, \bm{\beta}^\star, M_0, M_1, M_2, \epsilon_n) | \widehat{\mbf{F}}, \widehat{\mbf{U}}, \mbf{Y}) \ge e^{-C_1s\log p} \right|\mbf{F}, \mbf{U},  \mbf{X}, \widehat{\mbf{F}}, \widehat{\mbf{U}}, \mbf{H}\right).
\end{align*}
Next, by a change of measure trick and Cauchy-Schwarz inequality,
\begin{align*}
&= \int 1\left\{ \widehat{\pi}(A^c(\sigma^\star, \mbf{H}\bm{\alpha}^\star, \bm{\beta}^\star, M_0, M_1, M_2, \epsilon_n) |\widehat{\mbf{F}}, \widehat{\mbf{U}}, \mbf{y}) \ge e^{-C_1s\log p} \right\} \mc{N}(\mbf{y}| \mbf{F}\bm{\alpha}^\star + \mbf{U}\bm{\beta}^\star, \sigma^{\star 2}\mbf{I})d\mbf{y}\\
&= \int 1\left\{ \widehat{\pi}(A^c(\sigma^\star, \mbf{H}\bm{\alpha}^\star, \bm{\beta}^\star, M_0, M_1, M_2, \epsilon_n)|\widehat{\mbf{F}}, \widehat{\mbf{U}}, \mbf{y}) \ge e^{-C_1s\log p} \right\} \frac{\mc{N}(\mbf{y}| \mbf{F}\bm{\alpha}^\star + \mbf{U}\bm{\beta}^\star, \sigma^{\star 2}\mbf{I})}{\mc{N}(\mbf{y}| \widehat{\mbf{F}}\mbf{H}\bm{\alpha}^\star + \widehat{\mbf{U}}\bm{\beta}^\star, \sigma^{\star 2}\mbf{I})}\\
&~~~\times \mc{N}(\mbf{y}| \widehat{\mbf{F}}\mbf{H}\bm{\alpha}^\star + \widehat{\mbf{U}}\bm{\beta}^\star, \sigma^{\star 2}\mbf{I})d\mbf{y}\\
&\le \left[\int 1^2\left\{ \widehat{\pi}(A^c(\sigma^\star, \mbf{H}\bm{\alpha}^\star, \bm{\beta}^\star, M_0, M_1, M_2, \epsilon_n)|\widehat{\mbf{F}}, \widehat{\mbf{U}}, \mbf{y}) \ge e^{-C_1s\log p} \right\} \mc{N}(\mbf{y}| \widehat{\mbf{F}}\mbf{H}\bm{\alpha}^\star + \widehat{\mbf{U}}\bm{\beta}^\star, \sigma^{\star 2}\mbf{I}) d\mbf{y}\right]^{1/2}\\
&~~~\times \left[\int \left(\frac{\mc{N}(\mbf{y}| \mbf{F}\bm{\alpha}^\star + \mbf{U}\bm{\beta}^\star, \sigma^{\star 2}\mbf{I})}{\mc{N}(\mbf{y}| \widehat{\mbf{F}}\mbf{H}\bm{\alpha}^\star + \widehat{\mbf{U}}\bm{\beta}^\star, \sigma^{\star 2}\mbf{I})}\right)^2 \mc{N}(\mbf{y}| \widehat{\mbf{F}}\mbf{H}\bm{\alpha}^\star + \widehat{\mbf{U}}\bm{\beta}^\star, \sigma^{\star 2}\mbf{I})d\mbf{y}\right]^{1/2}.
\end{align*}

Proceed to bound two integrals separately. The logarithm of the second integral is the R\'{e}nyi Divergence of order 2 from $\mc{N}(\widehat{\mbf{F}}\mbf{H}\bm{\alpha}^\star + \widehat{\mbf{U}}\bm{\beta}^\star, \sigma^{\star 2}\mbf{I})$ to $\mc{N}(\mbf{F}\bm{\alpha}^\star + \mbf{U}\bm{\beta}^\star, \sigma^{\star 2}\mbf{I})$. It follows that
\begin{align*}
&~~~\log(\text{Second Integral}) = \Vert \widehat{\mbf{F}}\mbf{H}\bm{\alpha}^\star + \widehat{\mbf{U}}\bm{\beta}^\star - \mbf{F} \bm{\alpha}^\star - \mbf{U} \bm{\beta}^\star\Vert^2/\sigma^{\star 2}\\
&\le \left(\max_{j=1}^k \Vert (\widehat{\mbf{F}}\mbf{H})_j - \mbf{F}_j \Vert |\alpha_j| + \Vert \widehat{\mbf{U}}_{\xi^\star} - \mbf{U}_{\xi^\star}\Vert \Vert \bm{\beta}^\star_{\xi^\star} \Vert\right)^2/\sigma^{\star 2} \le C_4's\log p
\end{align*}
for some constant $C_4'$, where \eqref{conditioning} derives $\Vert (\widehat{\mbf{F}}\mbf{H})_j - \mbf{F}_j \Vert \le C_4\sqrt{\log p}$ and $\Vert \widehat{\mbf{U}}_{\xi^\star} - \mbf{U}_{\xi^\star}\Vert \le C_4\sqrt{s\log p}$, and Assumption \ref{asm:4} controls $\Vert \bm{\alpha}^\star\Vert = \bigO(1)$ and $\Vert \bm{\beta}^\star\Vert = \bigO(1)$.

On the other hand, let $\widehat{\mbb{P}}_{(\sigma, \bm{\alpha}, \bm{\beta})}$ denote the probability measure under which $\mbf{Y} \sim \mc{N}(\widehat{\mbf{F}}\bm{\alpha} + \widehat{\mbf{U}}\bm{\beta}, \sigma^2\mbf{I})$, then
$$\text{First Integral} = \widehat{\mbb{P}}_{(\sigma^\star, \mbf{H} \bm{\alpha}^\star, \bm{\beta}^\star)}\left( \widehat{\pi}(A^c(\sigma^\star, \mbf{H}\bm{\alpha}^\star, \bm{\beta}^\star, M_0, M_1, M_2, \epsilon_n)|\widehat{\mbf{F}}, \widehat{\mbf{U}}, \mbf{Y}) \ge e^{-C_1s\log p}\right),$$
which is concerning the posterior convergence rate of Bayesian sparse regression in model $\mbf{Y} \sim \mc{N}(\widehat{\mbf{F}}\bm{\alpha} + \widehat{\mbf{U}}\bm{\beta}, \sigma^2\mbf{I})$ with fixed design matrix $[\widehat{\mbf{F}}, \widehat{\mbf{U}}]$ to identify the true parameter $(\sigma^\star, \mbf{H}\bm{\alpha}^\star, \bm{\beta}^\star)$.

This fixe-design regression is analyzed by Theorem \ref{thm:3}. By part (a) of Theorem \ref{thm:3}, we can find $M_0, M_1, M_2, C_1, C_1'$ such that $C_1' > C_4'$ and the first Integral $\le e^{-C_1's\log p}$. Combining bounds of two integrals completes the proof of part (a) of Theorem \ref{thm:1}. Using similar arguments, parts (b) and (c) of Theorem \ref{thm:3} derive parts (b) and (c) of Theorem \ref{thm:1}, respectively.

\subsection{Bayesian Sparse Regression with Fixed Design}
\begin{theorem} \label{thm:3}
Recall that $\widehat{\mbb{P}}_{(\sigma, \bm{\alpha}, \bm{\beta})}$ denote the probability measure under the model $\mbf{Y} = \widehat{\mbf{F}} \bm{\alpha} + \widehat{\mbf{U}}\bm{\beta} + \sigma \bm{\varepsilon}$ with fixed design $[\widehat{\mbf{F}}, \widehat{\mbf{U}}]$. Suppose $[\widehat{\mbf{F}}, \widehat{\mbf{U}}]$ and true parameters satisfy
\begin{equation} \label{fixed design}
\begin{split}
& \widehat{\mbf{F}}^\T \widehat{\mbf{F}}/n = \mbf{I},~~~\widehat{\mbf{F}}^\T \widehat{\mbf{U}}/n = \mbf{0}\\
\min_{\xi:~|\xi| \le (M_0+1)s} & \lambda_{\min}(\widehat{\mbf{U}}_{\xi}^\T\widehat{\mbf{U}}_\xi/n) \ge \widehat{\kappa}_0\\
& \lambda_{\max}(\widehat{\mbf{U}}_{\xi^\star}^\T \widehat{\mbf{U}}_{\xi^\star}/n) \le \widehat{\kappa}_1\\
& \Vert \bm{\alpha}^\star \Vert = \bigO(1),~~~\Vert \bm{\beta}^\star \Vert = \bigO(1),
\end{split}
\end{equation}
then the following statements hold.
\begin{enumerate}[label=(\alph*)]
\item (estimation error rate) For any constants $C_1, C_1'$, there exist sufficiently large $M_0, M_1, M_2$ such that
\begin{align*}
\widehat{\mbb{P}}_{(\sigma^\star, \bm{\alpha}^\star, \bm{\beta}^\star)} & \left(\widehat{\pi}(A^c(\sigma^\star, \bm{\alpha}^\star, \bm{\beta}^\star, M_0, M_1, M_2, \epsilon_n)|\widehat{\mbf{F}}, \widehat{\mbf{U}}, \mbf{Y})\right.\\
&~~~~~~~~~~\left. \ge e^{-C_1s\log p}\right) \lesssim e^{-C_1's\log p}.
\end{align*}
\item (prediction error rate) For any constants $C_2, C_2'$, there exist sufficiently large $M_3$ such that
\begin{align*}
\widehat{\mbb{P}}_{(\sigma^\star, \bm{\alpha}^\star, \bm{\beta}^\star)} & \left( \widehat{\pi}( \Vert (\widehat{\mbf{F}}\bm{\alpha} + \widehat{\mbf{U}}\bm{\beta}) - (\widehat{\mbf{F}}\bm{\alpha}^\star + \widehat{\mbf{U}}\bm{\beta}^\star) \Vert  \ge \sigma^\star M_3 \sqrt{n} \epsilon_n| \widehat{\mbf{F}}, \widehat{\mbf{U}}, \mbf{Y}) \right.\\
&~~~~~~~~~~ \left. \ge e^{-C_2s\log p}\right)\lesssim e^{-C_2's\log p}.
\end{align*}
\item (model selection consistency) Suppose $\min_{j \in \xi^\star} |\beta^\star_j| \succ \epsilon_n$ in addition. For any constants $C_3, C_3'$, there exist sufficiently large $M_0, M_1, M_2$ such that
\begin{align*}
\widehat{\mbb{P}}_{(\sigma^\star, \bm{\alpha}^\star, \bm{\beta}^\star)} & \left(\widehat{\pi}( A^c(\sigma^\star, \bm{\alpha}^\star, \bm{\beta}^\star, M_0, M_1, M_2, \epsilon_n) \cup \{\xi \not \supseteq \xi^\star\} | \widehat{\mbf{F}}, \widehat{\mbf{U}}, \mbf{Y}) \right.\\
&~~~~~~~~~~ \left. \ge e^{-C_3s\log p}\right) \lesssim e^{-C_3's\log p}.
\end{align*}
\end{enumerate}
\end{theorem}
\noindent \textbf{Remark.} To apply Theorem \ref{thm:3} in the proof of Theorem \ref{thm:1}, we replace $(\sigma^\star, \bm{\alpha}^\star, \bm{\beta}^\star)$ with $(\sigma^\star, \mbf{H}\bm{\alpha}^\star, \bm{\beta}^\star)$, and check that $\Vert \mbf{H}\bm{\alpha}^\star \Vert = \bigO(1)$ in Theorem \ref{thm:1}.

The following proposition, which is also \citep[Lemma 6]{barron1998information} and \cite[Lemma A4]{song2017nearly}, is the central technique to prove Theorem \ref{thm:3}. 

\begin{proposition} \label{prop: barron}
Consider a parametric model $\{P_{\bm{\theta}}\}_{\bm{\theta} \in \Theta}$. Let $\Theta_{0n}$ and $\Theta_{n}$ be two subsets of the parameter space $\Theta$. Let $\{\mc{D}_n\}_{n \ge 1}$ be a sequence of data generations according to true parameter $\bm{\theta}^\star$. Let $\pi(\bm{\theta})$ be a prior distribution over $\Theta$. If
\begin{enumerate}[label=(\arabic*)]
\item $\pi(\Theta_{0n}) \le \delta_{0n}$,
\item there exists a test function $\phi_{n}(\mc{D}_n)$ such that
$$\sup_{\bm{\theta} \in \Theta_{n}} \mbb{E}_{\bm{\theta}} (1 - \phi_{n}) \le \delta_{1n}, \quad \mbb{E}_{\bm{\theta}^\star} \phi_{n} \le \delta_{1n}',$$
\item and
$$\mbb{P}_{\bm{\theta}^\star}\left(\frac{\int_{\Theta} \pi(\bm{\theta})P_{\bm{\theta}}(\mc{D}_n)d\bm{\theta}}{P_{\bm{\theta}^\star}(\mc{D}_n)} \le \delta_{2n}\right) \le \delta_{2n}',$$
\end{enumerate}
then for any $\delta_{3n}$,
$$\mbb{P}_{\bm{\theta}^\star}\left(\pi(\Theta_{0n} \cup \Theta_{n} |\mc{D}_n) \ge \frac{\delta_{0n}+\delta_{1n}}{\delta_{2n}\delta_{3n}}\right) \le \delta_{1n}' + \delta_{2n}' + \delta_{3n}.$$
\end{proposition}

The intuition of this proposition is that any less preferred $\bm{\theta} \in \Theta_{0n} \cup \Theta_{n}$ should either excluded by the prior (for $\bm{\theta} \in \Theta_{0n}$) or distinguished from $\bm{\theta}^\star$ by a uniformly powerful test $\phi_{n}$ (for $\bm{\theta} \in \Theta_{n}$).

Lemmas \ref{lem:Binomial}-\ref{lem:merging} are useful to verify the three conditions in Proposition \ref{prop: barron}, respectively. Their proofs are collected in the next subsection.

\begin{lemma} \label{lem:Binomial}
(Theorem 1.1 in \citep*{pelekis2016lower}) For a Binomial distributed random variable $\texttt{Binomial}(p, \mu)$, if $p \mu < m \le p -1$ then
$$\mbb{P}\left(\texttt{Binomial}(p, \mu) \ge m\right) \le \frac{\mu^{2(\tilde{m}+1)}}{2} \left. {p \choose \tilde{m}+1}  \right/ {m \choose \tilde{m}+1},$$
where $\tilde{m} = \lfloor (m - p\mu)/(1-\mu) \rfloor < m$.
\end{lemma}

\begin{lemma} \label{lem:test}
Under the same assumption of Theorem \ref{thm:3},
\begin{enumerate}[label=(\alph*)]
\item For
\begin{equation*}
\begin{split}
\Theta_{1n}
&= \left\{(\sigma^2, \bm{\alpha}, \bm{\beta}): \begin{split} 
& |\xi \setminus \xi^\star| \le M_0s,\\
& \frac{\sigma^2}{\sigma^{\star 2}} \not \in \left(\frac{1-M_1\epsilon_n}{1+M_1\epsilon_n}, \frac{1+M_1\epsilon_n}{1-M_1\epsilon_n}\right)
\end{split} \right\},\\
\phi_{1n} &=  1\left\{\max_{\xi:~|\xi\setminus\xi^\star| \le M_0s} \left|\mbf{Y}^\T\left[\mbf{I} - \widehat{\mbf{F}}\widehat{\mbf{F}}^\dagger - \widehat{\mbf{U}}_{\xi \cup \xi^\star}\widehat{\mbf{U}}_{\xi \cup \xi^\star}^\dagger\right]\mbf{Y}/ n \sigma^{\star 2}-1 \right| \ge M_1\epsilon_n\right\},
\end{split}
\end{equation*}
we have
\begin{align*}
\En \phi_{1n} &\le \exp(- (M_1^2/8 - M_0 - \smallo(1))s\log p),\\
\sup_{(\sigma, \alpha, \beta) \in \Theta_{1n}}\Ea (1-\phi_{1n}) &\le \exp(- (M_1^2/8 - \smallo(1))s\log p).
\end{align*}
\item For
\begin{equation*}
\begin{split}
\Theta_{2n}
&= \left\{(\sigma^2, \bm{\alpha}, \bm{\beta}): \begin{split} 
& |\xi \setminus \xi^\star| \le M_0s,\\
& \frac{\sigma^2}{\sigma^{\star 2}} \in \left(\frac{1-M_1\epsilon_n}{1+M_1\epsilon_n}, \frac{1+M_1\epsilon_n}{1-M_1\epsilon_n}\right)\\
& \left\Vert {\bm{\alpha} \choose \bm{\beta}} - {\bm{\alpha}^\star \choose \bm{\beta}^\star} \right\Vert > \sigma^\star M_2 \epsilon_n
\end{split} \right\},\\
\phi_{2n} &=  1\left\{\max_{\xi:~|\xi\setminus\xi^\star| \le M_0s} \left\Vert {\widehat{\mbf{F}}^\dagger\mbf{Y} \choose \widehat{\mbf{U}}_{\xi \cup \xi^\star}^\dagger \mbf{Y}} - {\bm{\alpha}^\star \choose \bm{\beta}^\star_{\xi \cup \xi^\star}} \right\Vert  \ge \sigma^\star M_2 \epsilon_n/2\right\}.
\end{split}
\end{equation*}
we have
\begin{align*}
\En \phi_{2n} &\le \exp(-(\min\{\widehat{\kappa}_0,1\}M_2^2/8 - M_0- \smallo(1))s\log p),\\
\sup_{(\sigma, \alpha, \beta) \in \Theta_{2n}}\Ea (1-\phi_{2n}) &\le \exp(-(\min\{\widehat{\kappa}_0,1\}M_2^2/8 - \smallo(1))s\log p).
\end{align*}
\item For
\begin{equation*}
\begin{split}
\Theta_{3n}
&= \left\{(\sigma^2, \bm{\alpha}, \bm{\beta}): \begin{split} 
& |\xi \setminus \xi^\star| \le M_0s,\\
& \frac{\sigma^2}{\sigma^{\star 2}} \in \left(\frac{1-M_1\epsilon_n}{1+M_1\epsilon_n}, \frac{1+M_1\epsilon_n}{1-M_1\epsilon_n}\right)\\
& \Vert (\widehat{\mbf{F}} \bm{\alpha} + \widehat{\mbf{U}}\bm{\beta}) - (\widehat{\mbf{F}} \bm{\alpha}^\star + \widehat{\mbf{U}}\bm{\beta}^\star)  \Vert > \sigma^\star M_3 \sqrt{n}\epsilon_n
\end{split} \right\},\\
\phi_{3n} &=  1\left\{\max_{\xi:~|\xi\setminus\xi^\star| \le M_0s} \left\Vert \left[ \widehat{\mbf{F}}\widehat{\mbf{F}}^\dagger + \widehat{\mbf{U}}_{\xi \cup \xi^\star}\widehat{\mbf{U}}_{\xi \cup \xi^\star}^\dagger\right]\mbf{Y} - \left(\widehat{\mbf{F}}\bm{\alpha}^\star+\widehat{\mbf{U}}\bm{\beta}^\star\right) \right\Vert  \ge \sigma^\star M_3 \sqrt{n}\epsilon_n/2\right\},
\end{split}
\end{equation*}
we have
\begin{align*}
\En \phi_{3n} &\le \exp(-(M_3^2/8 - M_0- \smallo(1))s\log p),\\
\sup_{(\sigma, \alpha, \beta) \in \Theta_{3n}}\Ea (1-\phi_{3n}) &\le \exp(-(M_3^2/8 - \smallo(1))s\log p).
\end{align*}
\item Suppose $\min_{j \in \xi^\star} |\beta_j^\star| \ge M_4 \sigma^\star \epsilon_n$ in addition. For
\begin{equation*}
\begin{split}
\Theta_{4n}
&= \left\{(\sigma^2, \bm{\alpha}, \bm{\beta}): \begin{split} 
& |\xi \setminus \xi^\star| \le M_0s,\\
& \frac{\sigma^2}{\sigma^{\star 2}} \in \left(\frac{1-M_1\epsilon_n}{1+M_1\epsilon_n}, \frac{1+M_1\epsilon_n}{1-M_1\epsilon_n}\right)\\
& \xi \not \supseteq \xi^\star
\end{split} \right\},\\
\phi_{4n} &=  1\left\{\min_{\xi \not \supseteq \xi^\star:~|\xi\setminus\xi^\star| \le M_0s} \left\Vert \left(\widehat{\mbf{U}}_{\xi \cup \xi^\star} \widehat{\mbf{U}}_{\xi \cup \xi^\star}^\dagger - \widehat{\mbf{U}}_{\xi} \widehat{\mbf{U}}_{\xi}^\dagger\right) \mbf{Y} \right\Vert  \le \sigma^\star \sqrt{\widehat{\kappa}_0}M_4 \sqrt{n}\epsilon_n/2\right\},
\end{split}
\end{equation*}
we have
\begin{align*}
\En \phi_{4n} &\le \exp(-(\widehat{\kappa}_0 M_4^2/8 - \smallo(1))s\log p),\\
\sup_{(\sigma, \alpha, \beta) \in \Theta_{4n}}\Ea (1-\phi_{4n}) &\le \exp(-(\widehat{\kappa}_0 M_4^2/8 - \smallo(1))s\log p).
\end{align*}
\item For
\begin{equation*}
\begin{split}
\Theta_{5n}
&= \left\{(\sigma^2, \bm{\alpha}, \bm{\beta}): \begin{split} 
& |\xi \setminus \xi^\star| \le M_0s,\\
& \frac{\sigma^2}{\sigma^{\star 2}} \in \left(\frac{1-M_1\epsilon_n}{1+M_1\epsilon_n}, \frac{1+M_1\epsilon_n}{1-M_1\epsilon_n}\right)\\
& \xi \supseteq \xi^\star,\\
& \left\Vert {\bm{\alpha} \choose \bm{\beta}} - {\bm{\alpha}^\star \choose \bm{\beta}^\star} \right\Vert > \sigma^\star M_2 \epsilon_n
\end{split} \right\},\\
\phi_{5n} &=  1\left\{\max_{\xi \supseteq \xi^\star:~|\xi\setminus\xi^\star| \le M_0s} \left\Vert { \widehat{\mbf{F}}^\dagger\mbf{Y} \choose \widehat{\mbf{U}}_{\xi}^\dagger \mbf{Y}} - {\bm{\alpha}^\star \choose \bm{\beta}^\star_{\xi}} \right\Vert  \ge \sigma^\star M_2 \epsilon_n/2\right\},
\end{split}
\end{equation*}
we have
\begin{align*}
\En \phi_{5n} &\le \exp(-(\min\{\widehat{\kappa}_0,1\}M_2^2/8 - M_0- \smallo(1))s\log p),\\
\sup_{(\sigma, \alpha, \beta) \in \Theta_{5n}}\Ea (1-\phi_{5n}) &\le \exp(-(\min\{\widehat{\kappa}_0,1\}M_2^2/8 - \smallo(1))s\log p).
\end{align*}
\end{enumerate}
\end{lemma}

\begin{lemma} \label{lem:merging}
Under the same assumption of Theorem \ref{thm:3},
$$\Pn \left(\int \frac{\mc{N}(\mbf{Y}|\widehat{\mbf{F}}\bm{\alpha} + \widehat{\mbf{U}}\bm{\beta}, \sigma^2\mbf{I})}{\mc{N}(\mbf{Y}|\widehat{\mbf{F}}\bm{\alpha}^\star + \widehat{\mbf{U}}\bm{\beta}^\star, \sigma^{\star 2}\mbf{I})} d\pi(\sigma, \bm{\alpha}, \bm{\beta}) \le e^{-C_5s\log p} \right) \lesssim e^{-C_5's\log p}$$
for sufficiently large $C_5$ and $C_5'$.
\end{lemma}

\begin{proof}[Proof of Theorem \ref{thm:1}, part (a)]
We verify the three conditions in Proposition \ref{prop: barron} one by one. Let
$$\Theta_{0n} = \{(\sigma^2, \bm{\alpha}, \bm{\beta}): |\xi \setminus \xi^\star| > M_0s\},~~~\Theta_{n} = \Theta_{1n} \cup \Theta_{2n},~~~\phi_{n} = \max\{\phi_{1n}, \phi_{2n}\},$$
where $\Theta_{1n}, \Theta_{2n}, \phi_{1n}, \phi_{2n}$ are defined in Lemma \ref{lem:test}. Then $\Theta_{0n} \cup \Theta_{n} = \Theta_{0n} \cup \Theta_{1n} \cup \Theta_{2n}  = A^c(\sigma^\star, \bm{\alpha}^\star, \bm{\beta}^\star, M_0, M_1, M_2, \epsilon_n)$.
Applying Lemma \ref{lem:Binomial} yields that
\begin{align*}
\pi(\Theta_{0n}) &\le \pi(|\xi| > M_0s) = \mbb{P}(\texttt{Binomial}(p, s_0/p) > M_0s)\\
&\lesssim \frac{1}{2}\left(\frac{s_0}{p}\right)^{2(M_0s-s_0+1)} {p \choose M_0s - s_0 + 1} \le \frac{1}{2}\left(\frac{s_0}{p}\right)^{2(M_0s - s_0 + 1)} p^{M_0s}\\
&\lesssim \delta_{0n} := e^{-M_0s\log p/2},
\end{align*}
for sufficiently large $M_0$. From parts (a),(b) of Lemma \ref{lem:test} and Lemma \ref{lem:split}, it follows that
\begin{align*}
\sup_{\Theta_n} \Ea(1-\phi_n) &\le \max\left\{\sup_{\Theta_{1n}} \Ea(1-\phi_{1n}),~\sup_{\Theta_{2n}} \Ea(1-\phi_{2n})\right\}\\
&\le \delta_{1n} := \exp(- (\min\{M_1^2/8, M_2^2/8, \widehat{\kappa}_0M_2^2/8\} - \smallo(1))s\log p)\\
\En \phi_n &\le \En \phi_{1n} +\En \phi_{2n}\\
&\le \delta_{1n}' := \exp(- (\min\{M_1^2/8, M_2^2/8, \widehat{\kappa}_0M_2^2/8\} -M_0 - \smallo(1))s\log p)
\end{align*}
By Lemma \ref{lem:merging}, the third condition in Proposition \ref{prop: barron} hold asymptotically with
$$\delta_{2n} = \exp( -C_5s\log p),~~~\delta_{2n}' = \exp( -C_5's\log p)$$
for any sufficiently large $C_5, C_5'$. For any $C_1$, $C_1'$, We can find sufficiently large $M_0, M_1, M_2, C_5, C_5'$ and suitable $\delta_{3n}$ such that
$$\frac{\delta_{0n} + \delta_{1n}}{\delta_{2n}\delta_{3n}} \le \exp(-C_1s\log p), ~~~\delta_{1n}' + \delta_{2n}' + \delta_{3n} \le \exp(-C_1's\log p)$$
to complete the proof.
\end{proof}

\begin{proof}[Proof of Theorem \ref{thm:3}, part(b)]
We use a similar argument to the proof of Theorem \ref{thm:3}, part (a) but different $\Theta_n$ and $\phi_n$. Let
$$\Theta_{n} = \Theta_{1n} \cup \Theta_{3n},~~~\phi_{n} = \max\{\phi_{1n}, \phi_{3n}\},$$
where $\Theta_{1n}, \Theta_{3n}, \phi_{1n}, \phi_{3n}$ are defined in Lemma \ref{lem:test}. Then
\begin{align*}
\Theta_{0n} \cup \Theta_{n} &= \Theta_{0n} \cup \Theta_{1n} \cup \Theta_{3n}\\
&\supseteq \{(\sigma^2, \bm{\alpha}, \bm{\beta}): \Vert (\widehat{\mbf{F}} \bm{\alpha} + \widehat{\mbf{U}}\bm{\beta}) - (\widehat{\mbf{F}} \bm{\alpha}^\star + \widehat{\mbf{U}}\bm{\beta}^\star) \Vert > \sigma^\star M_3 \sqrt{n}\epsilon_n\}.
\end{align*}
The second condition in Proposition \ref{prop: barron} follows from parts (a),(c) of Lemma \ref{lem:test} and Lemma \ref{lem:split}.
\begin{align*}
\sup_{\Theta_n} \Ea(1-\phi_n) &\le \max\left\{\sup_{\Theta_{1n}} \Ea(1-\phi_{1n}),~\sup_{\Theta_{3n}} \Ea(1-\phi_{3n})\right\}\\
&\le \delta_{1n} := \exp(- (\min\{M_1^2/8, M_3^2/8\} - \smallo(1))s\log p)\\
\En \phi_n &\le \En \phi_{1n} +\En \phi_{3n}\\
&\le \delta_{1n}' := \exp(- (\min\{M_1^2/8, M_3^2/8\} -M_0 - \smallo(1))s\log p)
\end{align*}
\end{proof}

\begin{proof}[Proof of Theorem \ref{thm:3}, part(c)]
We use a similar argument to the proof of Theorem \ref{thm:3}, part (a) but different $\Theta_n$ and $\phi_n$. $$\Theta_{n} = \Theta_{1n} \cup \Theta_{4n} \cup \Theta_{5n},~~~\phi_{n} = \max\{\phi_{1n}, \phi_{4n}, \phi_{5n}\},$$
where $\Theta_{1n}, \Theta_{4n}, \Theta_{5n}, \phi_{1n}, \phi_{4n}, \phi_{5n}$ are defined in Lemma \ref{lem:test}. Then $\Theta_{0n} \cup \Theta_{n} = \Theta_{0n} \cup \Theta_{1n} \cup \Theta_{4n} \cup \Theta_{5n} = A^c(\sigma^\star, \bm{\alpha}^\star, \bm{\beta}^\star, M_0, M_1, M_2, \epsilon_n) \cup \{\xi \not \supseteq \xi^\star \}$. The second condition in Proposition \ref{prop: barron} follows from parts (a),(d),(e) of Lemma \ref{lem:test} and Lemma \ref{lem:split}.
\begin{align*}
\sup_{\Theta_n} \Ea(1-\phi_n) &\le \max\left\{\sup_{\Theta_{1n}} \Ea(1-\phi_{1n}),~\sup_{\Theta_{4n}} \Ea(1-\phi_{4n}), ~\sup_{\Theta_{5n}} \Ea(1-\phi_{5n})\right\}\\
&\le \delta_{1n} := \exp(- (\min\{M_1^2/8, \widehat{\kappa}_0 M_4^2/8, M_2^2/8, \widehat{\kappa}_0M_2^2/8\} - \smallo(1))s\log p)\\
\En \phi_n &\le \En \phi_{1n} + \En \phi_{4n} + \En \phi_{5n}\\
&\le \delta_{1n}' := \exp(- (\min\{M_1^2/8, \widehat{\kappa}_0 M_4^2/8 + M_0, M_2^2/8, \widehat{\kappa}_0M_2^2/8\} - M_0 - \smallo(1))s\log p).
\end{align*}
\end{proof}

\subsection{Technical Proofs of Lemmas}

The proofs of Lemmas \ref{lem:test}-\ref{lem:merging} invoke a few preliminary results. We list them as follows.

\begin{lemma}[Probability bounds of chi-squared random variables] \label{lem:chi2}
Let $\chi^2_d$ be a chi-squared random variable of degree $d$.
\begin{enumerate}[label=(\alph*)]
\item For any $\epsilon_n$ such that $n\epsilon_n > d_n$,
\begin{align*}
\mbb{P}(\chi^2_{n - d_n} / n  \ge 1 + \epsilon_n) &\le e^{- \min \left\{\frac{(n\epsilon_n+d_n)^2}{8(n - d_n)}, \frac{n\epsilon_n+d_n}{8}\right\}},\\
\mbb{P}(\chi^2_{n - d_n} / n  \le 1 - \epsilon_n) &\le e^{- \min \left\{\frac{(n\epsilon_n-d_n)^2}{8(n - d_n)}, \frac{n\epsilon_n-d_n}{8}\right\}},
\end{align*}
In addition, if $\epsilon_n \to 0$ but $n\epsilon_n \succ d_n$,
\begin{align*}
\mbb{P}(\chi^2_{n - d_n} / n  \ge 1 + \epsilon_n) &\lesssim e^{-(1/8 - \smallo(1))n\epsilon_n^2}\\
\mbb{P}(\chi^2_{n - d_n} / n  \ge 1 + \epsilon_n) &\lesssim e^{-(1/8 - \smallo(1))n\epsilon_n^2}.
\end{align*}
\item
$$\mbb{P}(\chi^2_{d_n} \ge t_n) \le e^{-\left(\sqrt{2t_n - d_n} - \sqrt{d_n}\right)^2/4}.$$
In addition, if $t_n \succ d_n$ then for any $\tilde{t}_n$ such that $\tilde{t}_n/t_n \to 1$
$$\mbb{P}(\chi^2_{d_n} \ge t_n) \lesssim e^{-(1/2 - \smallo(1))\tilde{t}_n}.$$
\end{enumerate}
\end{lemma}
\begin{proof}
For part (a), the first assertion follows from the sub-exponential tail of chi-squared distribution, and the second assertion is due to
\begin{align*}
(1/8 - \smallo(1))n\epsilon_n^2 &\lesssim \frac{(n\epsilon_n+d_n)^2}{8(n - d_n)} \lesssim \frac{n\epsilon_n+d_n}{8}\\
(1/8 - \smallo(1))n\epsilon_n^2 &\lesssim \frac{(n\epsilon_n-d_n)^2}{8(n - d_n)} \lesssim \frac{n\epsilon_n-d_n}{8}
\end{align*}
For part (b), the first assertion is a corollary of \cite[Lemma 1]{laurent2000adaptive}, and the second assertion follows from
$$(1/2 - \smallo(1))\tilde{t}_n \lesssim \left(\sqrt{2t_n - d_n} - \sqrt{d_n}\right)^2/4.$$
\end{proof}

\begin{lemma}\label{lem:split}
For a collection of subspace $\{\Theta_j\}_{j=1}^m$ and a collection of test functions $\{\varphi_j\}_{j=1}^m$
$$\sup_{\theta \in \cup_{j=1}^m \Theta_j} \mbb{E}_\theta \left(1 - \max_{j=1}^m \varphi_j\right) \le \max_{j=1}^m \left\{\sup_{\theta \in \Theta_j} \mbb{E}_\theta(1 - \varphi_j) \right\}.$$
\end{lemma}
\begin{proof}
\begin{align*}
\sup_{\theta \in \cup_{j=1}^m \Theta_j} \mbb{E}_\theta \left(1 - \max_{j=1}^m \varphi_j\right)
&= \max_{j=1}^m \left\{\sup_{\theta \in \Theta_j} \mbb{E}_\theta \left(1 - \max_{k=1}^m \varphi_k\right) \right\}\\
&= \max_{j=1}^m \left\{\sup_{\theta \in \Theta_j} \mbb{E}_\theta \left( \min_{k=1}^m (1-\varphi_k)\right) \right\}\\
&\le \max_{j=1}^m \left\{\sup_{\theta \in \Theta_j} \mbb{E}_\theta \left( 1-\varphi_j \right) \right\}.
\end{align*}
\end{proof}

\begin{lemma}[Part of Corollary 2.4 in \citep{liu2005eigenvalue}] \label{lem:schur}
Let $$\mbf{G} = \left[\begin{array}{c c} \mbf{G}_{11} & \mbf{G}_{12}\\ 
\mbf{G}_{21} & \mbf{G}_{22}\end{array}\right]$$ be a $p \times p$ positive semi-definite matrix with $q \times q$ non-singular principal sub-matrix $\mbf{G}_{11}$ then
$$\lambda_{\min}(\mbf{G}_{22} - \mbf{G}_{21}\mbf{G}_{11}^{-1}\mbf{G}_{12}) \ge \lambda_{\min}(\mbf{G}).$$
\end{lemma}

\begin{proof}[Proof of Lemma \ref{lem:test}, part (a)]
Under the null hypothesis, write
\begin{align*}
&~~~\En \phi_{1n}\\
&\overset{(1)}{=} \Pn \left(\max_{\xi:~|\xi\setminus\xi^\star| \le M_0s} \left|\bm{\varepsilon}^\T\left[\mbf{I} - \widehat{\mbf{F}}\widehat{\mbf{F}}^\dagger - \widehat{\mbf{U}}_{\xi \cup \xi^\star}\widehat{\mbf{U}}_{\xi \cup \xi^\star}^\dagger\right]\bm{\varepsilon}/ n -1 \right| \ge M_1\epsilon_n \right)\\
&\overset{(2)}{\le} \Pn \left( \bm{\varepsilon}^\T\left[\mbf{I} - \widehat{\mbf{F}}\widehat{\mbf{F}}^\dagger - \widehat{\mbf{U}}_{\xi^\star}\widehat{\mbf{U}}_{\xi^\star}^\dagger\right]\bm{\varepsilon}/ n \ge 1 + M_1\epsilon_n\right)\\
&~~~+ \sum_{\xi: ~|\xi| = M_0s, ~\xi \setminus\xi^\star = \emptyset} \Pn \left( \bm{\varepsilon}^\T\left[\mbf{I} - \widehat{\mbf{F}}\widehat{\mbf{F}}^\dagger - \widehat{\mbf{U}}_{\xi \cup \xi^\star}\widehat{\mbf{U}}_{\xi \cup \xi^\star}^\dagger\right]\bm{\varepsilon}/ n \le 1 - M_1\epsilon_n\right)\\
&\overset{(3)}{=} \mbb{P}\left( \chi^2_{n - k - s}/n \ge 1+M_1\epsilon_n \right) +  {p-s \choose M_0 s} \mbb{P}\left(\chi^2_{n - k - (M_0+1)s}/n \le 1 - M_1 \epsilon_n\right).
\end{align*}
(1) follows from the facts that $\mbf{Y} = \widehat{\mbf{F}} \bm{\alpha}^\star + \widehat{\mbf{U}}\bm{\beta}^\star + \sigma^\star \bm{\varepsilon}$ with $\bm{\beta}^\star_{\xi^{\star c}} = \mbf{0}$ under $\Pn$ and that $\widehat{\mbf{F}}^\T \widehat{\mbf{U}} = \mbf{0}$. For (2), we observe that projection matrices
$$\widehat{\mbf{U}}_{\xi' \cup \xi^\star}\widehat{\mbf{U}}_{\xi' \cup \xi^\star}^\dagger \le \widehat{\mbf{U}}_{\xi'' \cup \xi^\star}\widehat{\mbf{U}}_{\xi'' \cup \xi^\star}^\dagger$$
for nested models $\xi' \subseteq \xi''$, and thus the term $\bm{\varepsilon}^\T\left[\mbf{I} - \widehat{\mbf{F}}\widehat{\mbf{F}}^\dagger - \widehat{\mbf{U}}_{\xi \cup \xi^\star}\widehat{\mbf{U}}_{\xi \cup \xi^\star}^\dagger\right]\bm{\varepsilon}$ achieves its maximum value at any $\xi \subseteq \xi^\star$ and its minimum value at some $\xi$ s.t. $|\xi| = M_0 s$ and $\xi \setminus \xi^\star = \emptyset$. (3) uses the fact that 
$$\bm{\varepsilon}^\T\left[\mbf{I} - \widehat{\mbf{F}}\widehat{\mbf{F}}^\dagger - \widehat{\mbf{U}}_{\xi \cup \xi^\star}\widehat{\mbf{U}}_{\xi \cup \xi^\star}^\dagger\right]\bm{\varepsilon} \sim \chi^2_{n - k - |\xi \cup \xi^\star|}.$$
Applying Lemma \ref{lem:chi2}, part (a) yields
$$\En \phi_{1n} \lesssim \left(1 + p^{M_0s}\right) e^{-(M_1^2/8 - \smallo(1)) s\log p} \lesssim e^{-(M_1^2/8 - M_0 - \smallo(1)) s\log p}.$$

Under the alternative hypothesis, observe that $\phi_{1n} = \max_{\xi':~|\xi' \setminus \xi^\star| \le M_0s} \phi_{1n}^{\xi'}$, where
$$\phi_{1n}^{\xi'} = 1\left\{ \left|\mbf{Y}^\T\left[\mbf{I} - \widehat{\mbf{F}}\widehat{\mbf{F}}^\dagger - \widehat{\mbf{U}}_{\xi \cup \xi^\star}\widehat{\mbf{U}}_{\xi \cup \xi^\star}^\dagger\right]\mbf{Y}/ n \sigma^{\star 2}-1 \right| \ge M_1\epsilon_n \right\}.$$
Using Lemma \ref{lem:split},
$$\sup_{\Theta_{1n}}\Ea(1 - \phi_{1n}) \le \max_{\xi':~|\xi' \setminus \xi^\star| \le M_0s} ~\sup_{\Theta_{1n} \cap \{\xi = \xi'\}} \Ea (1-\phi_{1n}^{\xi'}).$$
For any $\xi'$ such that $|\xi' \setminus \xi^\star| \le M_0s$ and any $(\sigma, \bm{\alpha}, \bm{\beta}) \in \Theta_{1n} \cap \{\xi = \xi'\}$, write
\begin{align*}
&~~~\Ea (1-\phi_{1n}^{\xi'})\\
&\overset{(1)}{=} \Pa \left(\left|\bm{\varepsilon}^\T\left[\mbf{I} - \widehat{\mbf{F}}\widehat{\mbf{F}}^\dagger - \widehat{\mbf{U}}_{\xi \cup \xi^\star}\widehat{\mbf{U}}_{\xi \cup \xi^\star}^\dagger\right]\bm{\varepsilon}/ n \times (\sigma^2 /\sigma^{\star 2})-1 \right| < M_1\epsilon_n\right)\\
&\overset{(2)}{\le} \Pa \left(\bm{\varepsilon}^\T\left[\mbf{I} - \widehat{\mbf{F}}\widehat{\mbf{F}}^\dagger - \widehat{\mbf{U}}_{\xi \cup \xi^\star}\widehat{\mbf{U}}_{\xi \cup \xi^\star}^\dagger\right]\bm{\varepsilon}/ n \not \in (1-M_1\epsilon_n, 1+M_1\epsilon_n)\right)\\
&\overset{(3)}{=} \mbb{P} \left(\chi^2_{n - k - |\xi \cup \xi^\star|}/ n \not \in (1-M_1\epsilon_n, 1+M_1\epsilon_n)\right)\\
&\le \mbb{P}\left( \chi^2_{n - k - (M_0+1)s}/n \le 1 - M_1 \epsilon_n\right) + \mbb{P}\left( \chi^2_{n - k - s}/n \ge 1 + M_1 \epsilon_n\right)
\end{align*}
(1) follows from the facts that $\mbf{Y} = \widehat{\mbf{F}} \bm{\alpha}+ \widehat{\mbf{U}}\bm{\beta} + \sigma \bm{\varepsilon}$ with $\bm{\beta}_{\xi^c} = \mbf{0}$ under $\Pa$ and that $\widehat{\mbf{F}}^\T \widehat{\mbf{U}} = \mbf{0}$. (2) plugs in the restriction
$$\frac{\sigma^2}{\sigma^{\star 2}} \not \in \left(\frac{1-M_1\epsilon_n}{1+M_1\epsilon_n}, \frac{1+M_1\epsilon_n}{1-M_1\epsilon_n}\right)$$
from the definition of $\Theta_{1n}$. (3) uses the fact that $$\bm{\varepsilon}^\T\left[\mbf{I} - \widehat{\mbf{F}}\widehat{\mbf{F}}^\dagger - \widehat{\mbf{U}}_{\xi \cup \xi^\star}\widehat{\mbf{U}}_{\xi \cup \xi^\star}^\dagger\right]\bm{\varepsilon} \sim \chi^2_{n - k - |\xi \cup \xi^\star|}$$
again. Since the final bound in the last display is uniform for any $\xi'$ such that $|\xi' \setminus \xi^\star| \le M_0s$ and any $(\sigma, \bm{\alpha}, \bm{\beta}) \in \Theta_{1n} \cap \{\xi = \xi'\}$, we apply Lemma \ref{lem:chi2}, part (a) and yield
$$\sup_{\Theta_{2n}} \Ea (1-\phi_{1n}) \lesssim e^{-(M_1^2/8 - \smallo(1))n\epsilon_n^2} = e^{-(M_1^2/8 - \smallo(1))s\log p}.$$
\end{proof}

\begin{proof}[Proof of Lemma \ref{lem:test}, part (b)]
Under the null hypothesis, write
\begin{align*}
&~~~\En \phi_{2n}\\
&\overset{(1)}{=} \Pn \left(\max_{\xi:~|\xi\setminus\xi^\star| \le M_0s} \left\Vert {\widehat{\mbf{F}}^\dagger\bm{\varepsilon} \choose \widehat{\mbf{U}}_{\xi \cup \xi^\star}^\dagger \bm{\varepsilon}} \right\Vert  \ge M_2 \epsilon_n/2\right)\\
&= \Pn \left(\max_{\xi:~|\xi\setminus\xi^\star| \le M_0s} \bm{\varepsilon}^\T\left[ \widehat{\mbf{F}}^{\dagger \T}\widehat{\mbf{F}}^\dagger + \widehat{\mbf{U}}_{\xi \cup \xi^\star}^{\dagger \T} \widehat{\mbf{U}}_{\xi \cup \xi^\star}^\dagger\right] \bm{\varepsilon} \ge M_2^2 \epsilon_n^2/4\right)\\
&\overset{(2)}{\le} \Pn \left(\max_{\xi:~|\xi\setminus\xi^\star| \le M_0s} \bm{\varepsilon}^\T \left[ \widehat{\mbf{F}}\widehat{\mbf{F}}^\dagger + \widehat{\mbf{U}}_{\xi \cup \xi^\star} \widehat{\mbf{U}}_{\xi \cup \xi^\star}^\dagger\right] \bm{\varepsilon} \ge \min\{\widehat{\kappa}_0, 1\} M_2^2 n\epsilon_n^2/4\right)\\
&\overset{(3)}{\le} \sum_{\xi: ~|\xi| = M_0s, ~\xi \setminus\xi^\star = \emptyset} \Pn \left( \bm{\varepsilon}^\T \left[ \widehat{\mbf{F}}\widehat{\mbf{F}}^\dagger + \widehat{\mbf{U}}_{\xi \cup \xi^\star} \widehat{\mbf{U}}_{\xi \cup \xi^\star}^\dagger\right] \bm{\varepsilon} \ge \min\{\widehat{\kappa}_0, 1\} M_2^2 n\epsilon_n^2/4\right)\\
&\overset{(4)}{=} {p - s \choose M_0 s} \mbb{P}\left( \chi^2_{k+(M_0+1)s} \ge \min\{\widehat{\kappa}_0, 1\} M_2^2 n\epsilon_n^2/4\right)
\end{align*}
(1) follows from the facts that $\mbf{Y} = \widehat{\mbf{F}} \bm{\alpha}^\star + \widehat{\mbf{U}}\bm{\beta}^\star + \sigma^\star \bm{\varepsilon}$ with $\bm{\beta}^\star_{\xi^{\star c}} = \mbf{0}$ under $\Pn$ and that $\widehat{\mbf{F}}^\T \widehat{\mbf{U}} = \mbf{0}$. (2) is due to
\begin{align*}
\widehat{\mbf{F}}^{\dagger \T}\widehat{\mbf{F}}^\dagger
&\le \lambda_{\min}\left( \widehat{\mbf{F}}^\T\widehat{\mbf{F}} \right)^{-1} \widehat{\mbf{F}} \widehat{\mbf{F}}^\dagger = n^{-1}\widehat{\mbf{F}} \widehat{\mbf{F}}^\dagger\\
\widehat{\mbf{U}}_{\xi \cup \xi^\star}^{\dagger \T}\widehat{\mbf{U}}_{\xi \cup \xi^\star}^\dagger
&\le \lambda_{\min}\left( \widehat{\mbf{U}}_{\xi \cup \xi^\star}^\T\widehat{\mbf{U}}_{\xi \cup \xi^\star} \right)^{-1} \widehat{\mbf{U}}_{\xi \cup \xi^\star} \widehat{\mbf{U}}_{\xi \cup \xi^\star}^\dagger \le \left(n{\widehat{\kappa}}_0\right)^{-1}\widehat{\mbf{U}}_{\xi \cup \xi^\star} \widehat{\mbf{U}}_{\xi \cup \xi^\star}^\dagger
\end{align*}
For (3), we observe that projection matrices
$$\widehat{\mbf{U}}_{\xi' \cup \xi^\star}\widehat{\mbf{U}}_{\xi' \cup \xi^\star}^\dagger \le \widehat{\mbf{U}}_{\xi'' \cup \xi^\star}\widehat{\mbf{U}}_{\xi'' \cup \xi^\star}^\dagger$$
for nested models $\xi' \subseteq \xi''$, and thus the term $\bm{\varepsilon}^\T \widehat{\mbf{U}}_{\xi \cup \xi^\star}\widehat{\mbf{U}}_{\xi \cup \xi^\star}^\dagger \bm{\varepsilon}$ achieves its maximum value at some $\xi$ s.t. $|\xi| = M_0 s$ and $\xi \setminus \xi^\star = \emptyset$. (4) uses the fact that $$\bm{\varepsilon}^\T\left[ \widehat{\mbf{F}}\widehat{\mbf{F}}^\dagger + \widehat{\mbf{U}}_{\xi \cup \xi^\star}\widehat{\mbf{U}}_{\xi \cup \xi^\star}^\dagger\right]\bm{\varepsilon} \sim \chi^2_{ k + |\xi \cup \xi^\star|}.$$
Applying Lemma \ref{lem:chi2}, part (b) yields
$$\En \phi_{2n} \lesssim p^{M_0s}e^{-(\min\{\widehat{\kappa}_0,1\} M_2^2/8 - \smallo(1)) n\epsilon_n^2} = e^{-(\min\{\widehat{\kappa}_0,1\} M_2^2/8 -M_0 - \smallo(1)) s \log p}.$$

Under the alternative hypothesis, observe that $\phi_{2n} = \max_{\xi':~|\xi' \setminus \xi^\star| \le M_0s} \phi_{2n}^{\xi'}$, where
$$\phi_{2n}^{\xi'} = 1\left\{ \left\Vert {\widehat{\mbf{F}}^\dagger\mbf{Y} \choose \widehat{\mbf{U}}_{\xi' \cup \xi^\star}^\dagger \mbf{Y}} - {\bm{\alpha}^\star \choose \bm{\beta}^\star_{\xi' \cup \xi^\star}} \right\Vert  \ge \sigma^\star M_2 \epsilon_n/2\right\}.$$
Using Lemma \ref{lem:split},
$$\sup_{\Theta_{2n}}\Ea(1 - \phi_{2n}) \le \max_{\xi':~|\xi' \setminus \xi^\star| \le M_0s} ~\sup_{\Theta_{2n} \cap \{\xi = \xi'\}} \Ea (1-\phi_{2n}^{\xi'}).$$
For any $\xi'$ such that $|\xi' \setminus \xi^\star| \le M_0s$ and any $(\sigma, \bm{\alpha}, \bm{\beta}) \in \Theta_{2n} \cap \{\xi = \xi'\}$, write
\begin{align*}
&~~~\Ea (1-\phi_{2n}^{\xi'})\\
&\overset{(1)}{=} \Pa \left(\left\Vert {\bm{\alpha} \choose \bm{\beta}_{\xi \cup \xi^\star}} - {\bm{\alpha}^\star \choose \bm{\beta}^\star_{\xi \cup \xi^\star}}  + \sigma {\widehat{\mbf{F}}^\dagger\bm{\varepsilon} \choose \widehat{\mbf{U}}_{\xi \cup \xi^\star}^\dagger \bm{\varepsilon}}  \right\Vert  < \sigma^\star M_2 \epsilon_n/2\right)\\
&\overset{(2)}{\le} \Pa \left( \left\Vert {\widehat{\mbf{F}}^\dagger\bm{\varepsilon} \choose \widehat{\mbf{U}}_{\xi \cup \xi^\star}^\dagger \bm{\varepsilon}} \right\Vert  \ge \sqrt{\frac{1-M_1\epsilon_n}{1+M_1\epsilon_n}} \times M_2 \epsilon_n/2\right)\\
&\overset{(3)}{\le} \mbb{P}\left( \chi^2_{k + (M_0+1)s} \ge \frac{1-M_1\epsilon_n}{1+M_1\epsilon_n} \times \min\{\widehat{\kappa}_0,1\} M_2^2 n\epsilon_n^2/4\right).
\end{align*}
(1) follows from the facts that $\mbf{Y} = \widehat{\mbf{F}} \bm{\alpha}+ \widehat{\mbf{U}}\bm{\beta} + \sigma \bm{\varepsilon}$ with $\bm{\beta}_{\xi^c} = \mbf{0}$ under $\Pa$ and that $\widehat{\mbf{F}}^\T \widehat{\mbf{U}} = \mbf{0}$. (2) plugs in the restrictions
$$\left\Vert {\bm{\alpha} \choose \bm{\beta}_{\xi \cup \xi^\star}} - {\bm{\alpha}^\star \choose \bm{\beta}^\star_{\xi \cup \xi^\star}} \right\Vert > \sigma^\star M_2\epsilon_n, ~~~\frac{\sigma^{\star 2}}{\sigma^2} > \frac{1-M_1\epsilon_n}{1+M_1\epsilon_n}$$
from the definition of $\Theta_{2n}$. (3) uses a similar argument to what we have used for the null hypothesis. Since the final bound in the last display is uniform for any $\xi'$ such that $|\xi' \setminus \xi^\star| \le M_0s$ and any $(\sigma, \bm{\alpha}, \bm{\beta}) \in \Theta_{2n} \cap \{\xi = \xi'\}$, we apply Lemma \ref{lem:chi2}, part (b) and yield
$$\sup_{\Theta_{2n}} \Ea (1-\phi_{2n}) \lesssim e^{-(\min\{\widehat{\kappa}_0,1\} M_2^2/8 - \smallo(1))n\epsilon_n^2} = e^{-(\min\{\widehat{\kappa}_0,1\}M_2^2/8 - \smallo(1))s\log p}.$$
\end{proof}

\begin{proof}[Proof of Lemma \ref{lem:test}, part (c)]
Under the null hypothesis, write
\begin{align*}
&~~~\En \phi_{3n}\\
&\overset{(1)}{=} \Pn \left(\max_{\xi:~|\xi\setminus\xi^\star| \le M_0s} \left\Vert \left( \widehat{\mbf{F}}\widehat{\mbf{F}}^\dagger + \widehat{\mbf{U}}_{\xi \cup \xi^\star}\widehat{\mbf{U}}_{\xi \cup \xi^\star}^\dagger\right) \bm{\varepsilon} \right\Vert  \ge M_3 \sqrt{n}\epsilon_n/2\right)\\
&= \Pn \left(\max_{\xi:~|\xi\setminus\xi^\star| \le M_0s} \bm{\varepsilon}^\T \left[ \widehat{\mbf{F}}\widehat{\mbf{F}}^\dagger + \widehat{\mbf{U}}_{\xi \cup \xi^\star}\widehat{\mbf{U}}_{\xi \cup \xi^\star}^\dagger\right] \bm{\varepsilon} \ge M_3^2 n\epsilon_n^2/4\right)\\
&\overset{(2)}{\le} \sum_{\xi: ~|\xi| = M_0s, ~\xi \setminus\xi^\star = \emptyset} \Pn \left( \bm{\varepsilon}^\T \left[ \widehat{\mbf{F}}\widehat{\mbf{F}}^\dagger + \widehat{\mbf{U}}_{\xi \cup \xi^\star}\widehat{\mbf{U}}_{\xi \cup \xi^\star}^\dagger\right] \bm{\varepsilon} \ge M_3^2 n\epsilon_n^2/4\right)\\
&\overset{(3)}{=} {p - s \choose M_0 s} \mbb{P}\left( \chi^2_{k+(M_0+1)s} \ge M_3^2 n\epsilon_n^2/4\right)
\end{align*}
(1) follows from the facts that $\mbf{Y} = \widehat{\mbf{F}} \bm{\alpha}^\star + \widehat{\mbf{U}}\bm{\beta}^\star + \sigma^\star \bm{\varepsilon}$ with $\bm{\beta}^\star_{\xi^{\star c}} = \mbf{0}$ under $\Pn$ and that $\widehat{\mbf{F}}^\T \widehat{\mbf{U}} = \mbf{0}$. For (2), we observe that projection matrices
$$\widehat{\mbf{U}}_{\xi' \cup \xi^\star}\widehat{\mbf{U}}_{\xi' \cup \xi^\star}^\dagger \le \widehat{\mbf{U}}_{\xi'' \cup \xi^\star}\widehat{\mbf{U}}_{\xi'' \cup \xi^\star}^\dagger$$
for nested models $\xi' \subseteq \xi''$, and thus the term $\bm{\varepsilon}^\T \widehat{\mbf{U}}_{\xi \cup \xi^\star}\widehat{\mbf{U}}_{\xi \cup \xi^\star}^\dagger \bm{\varepsilon}$ achieves its maximum value at some $\xi$ s.t. $|\xi| = M_0 s$ and $\xi \setminus \xi^\star = \emptyset$. (3) uses the fact that $$\bm{\varepsilon}^\T\left[ \widehat{\mbf{F}}\widehat{\mbf{F}}^\dagger + \widehat{\mbf{U}}_{\xi \cup \xi^\star}\widehat{\mbf{U}}_{\xi \cup \xi^\star}^\dagger\right]\bm{\varepsilon} \sim \chi^2_{ k + |\xi \cup \xi^\star|}.$$
Applying Lemma \ref{lem:chi2}, part (b) yields
$$\En \phi_{3n} \lesssim p^{M_0s}e^{-(M_3^2/8 - \smallo(1)) n\epsilon_n^2} = e^{-(M_3^2/8 -M_0 - \smallo(1)) s \log p}.$$

Under the alternative hypothesis, observe that $\phi_{3n} = \max_{\xi':~|\xi' \setminus \xi^\star| \le M_0s} \phi_{3n}^{\xi'}$, where
$$\phi_{3n}^{\xi'} = 1\left\{ \left\Vert \left[ \widehat{\mbf{F}}\widehat{\mbf{F}}^\dagger + \widehat{\mbf{U}}_{\xi \cup \xi^\star}\widehat{\mbf{U}}_{\xi \cup \xi^\star}^\dagger\right]\mbf{Y} - \left(\widehat{\mbf{F}}\bm{\alpha}^\star+\widehat{\mbf{U}}\bm{\beta}^\star\right) \right\Vert  \ge \sigma^\star M_3 \sqrt{n}\epsilon_n/2\right\}.$$
Using Lemma \ref{lem:split},
$$\sup_{\Theta_{3n}}\Ea(1 - \phi_{3n}) \le \max_{\xi':~|\xi' \setminus \xi^\star| \le M_0s} ~\sup_{\Theta_{3n} \cap \{\xi = \xi'\}} \Ea (1-\phi_{3n}^{\xi'}).$$
For any $\xi'$ such that $|\xi' \setminus \xi^\star| \le M_0s$ and any $(\sigma, \bm{\alpha}, \bm{\beta}) \in \Theta_{3n} \cap \{\xi = \xi'\}$, write
\begin{align*}
&~~~\Ea (1-\phi_{3n}^{\xi'})\\
&\overset{(1)}{=} \Pa \left(\left\Vert \left(\widehat{\mbf{F}} \bm{\alpha} + \widehat{\mbf{U}}\bm{\beta}\right) - \left(\widehat{\mbf{F}} \bm{\alpha}^\star + \widehat{\mbf{U}}\bm{\beta}^\star\right) + \sigma \left[ \widehat{\mbf{F}}\widehat{\mbf{F}}^\dagger + \widehat{\mbf{U}}_{\xi \cup \xi^\star}\widehat{\mbf{U}}_{\xi \cup \xi^\star}^\dagger\right] \bm{\varepsilon} \right\Vert  < \sigma^\star M_3 \sqrt{n}\epsilon_n/2\right)\\
&\overset{(2)}{\le} \Pa \left( \left\Vert \left[ \widehat{\mbf{F}}\widehat{\mbf{F}}^\dagger + \widehat{\mbf{U}}_{\xi \cup \xi^\star}\widehat{\mbf{U}}_{\xi \cup \xi^\star}^\dagger\right] \bm{\varepsilon} \right\Vert  \ge \sqrt{\frac{1-M_1\epsilon_n}{1+M_1\epsilon_n}} \times M_3 \sqrt{n}\epsilon_n/2\right)\\
&= \Pa \left( \bm{\varepsilon}^\T \left[ \widehat{\mbf{F}}\widehat{\mbf{F}}^\dagger + \widehat{\mbf{U}}_{\xi \cup \xi^\star}\widehat{\mbf{U}}_{\xi \cup \xi^\star}^\dagger\right] \bm{\varepsilon} \ge \frac{1-M_1\epsilon_n}{1+M_1\epsilon_n} \times M_3^2 n\epsilon_n^2/4\right)\\
&\overset{(3)}{=} \mbb{P}\left( \chi^2_{k + (M_0+1)s} \ge \frac{1-M_1\epsilon_n}{1+M_1\epsilon_n} \times M_3^2 n\epsilon_n^2/4\right).
\end{align*}
(1) follows from the facts that $\mbf{Y} = \widehat{\mbf{F}} \bm{\alpha}+ \widehat{\mbf{U}}\bm{\beta} + \sigma \bm{\varepsilon}$ with $\bm{\beta}_{\xi^c} = \mbf{0}$ under $\Pa$ and that $\widehat{\mbf{F}}^\T \widehat{\mbf{U}} = \mbf{0}$. (2) plugs in the restrictions
$$\left\Vert \left(\widehat{\mbf{F}} \bm{\alpha} + \widehat{\mbf{U}}\bm{\beta}\right) - \left(\widehat{\mbf{F}} \bm{\alpha}^\star + \widehat{\mbf{U}}\bm{\beta}^\star\right) \right\Vert > \sigma^\star M_3\sqrt{n}\epsilon_n, ~~~\frac{\sigma^{\star 2}}{\sigma^2} > \frac{1-M_1\epsilon_n}{1+M_1\epsilon_n}$$
from the definition of $\Theta_{3n}$. (3) uses the fact that $$\bm{\varepsilon}^\T\left[ \widehat{\mbf{F}}\widehat{\mbf{F}}^\dagger + \widehat{\mbf{U}}_{\xi \cup \xi^\star}\widehat{\mbf{U}}_{\xi \cup \xi^\star}^\dagger\right]\bm{\varepsilon} \sim \chi^2_{ k + |\xi \cup \xi^\star|}.$$
Since the final bound in the last display is uniform for any $\xi'$ such that $|\xi' \setminus \xi^\star| \le M_0s$ and any $(\sigma, \bm{\alpha}, \bm{\beta}) \in \Theta_{3n} \cap \{\xi = \xi'\}$, we apply Lemma \ref{lem:chi2}, part (b) and yield
$$\sup_{\Theta_{3n}} \Ea (1-\phi_{3n}) \lesssim e^{-(M_3^2/8 - \smallo(1))n\epsilon_n^2} = e^{-(M_3^2/8 - \smallo(1))s\log p}.$$
\end{proof}

\begin{proof}[Proof of Lemma \ref{lem:test}, part (d)]
We first show that
\begin{equation} \label{schur}
\min_{\xi \not \supseteq \xi^\star:~|\xi\setminus\xi^\star| \le M_0s} \left\Vert \left(\widehat{\mbf{U}}_{\xi \cup \xi^\star} \widehat{\mbf{U}}_{\xi \cup \xi^\star}^\dagger - \widehat{\mbf{U}}_{\xi} \widehat{\mbf{U}}_{\xi}^\dagger\right)\widehat{\mbf{U}}_{\xi^\star} \bm{\beta}^\star_{\xi^\star}\right\Vert \gtrsim \sigma^\star \sqrt{\widehat{\kappa}_0} M_4 \sqrt{n}\epsilon_n.   
\end{equation}
Indeed, for any $\xi \not \supseteq \xi^\star$ s.t. $|\xi\setminus\xi^\star| \le M_0s$,
\begin{align*}
\left\Vert \left(\widehat{\mbf{U}}_{\xi \cup \xi^\star} \widehat{\mbf{U}}_{\xi \cup \xi^\star}^\dagger - \widehat{\mbf{U}}_{\xi} \widehat{\mbf{U}}_{\xi}^\dagger\right)\widehat{\mbf{U}}_{\xi^\star} \bm{\beta}^\star_{\xi^\star}\right\Vert^2
&= \left\Vert \left(\mbf{I} - \widehat{\mbf{U}}_{\xi} \widehat{\mbf{U}}_{\xi}^\dagger\right)\widehat{\mbf{U}}_{\xi^\star} \bm{\beta}^\star_{\xi^\star}\right\Vert^2\\
&= \left\Vert \left(\mbf{I} - \widehat{\mbf{U}}_{\xi} \widehat{\mbf{U}}_{\xi}^\dagger\right)\widehat{\mbf{U}}_{\xi^\star \setminus \xi} \bm{\beta}^\star_{\xi^\star \setminus \xi}\right\Vert^2\\
&= \bm{\beta}_{\xi^\star \setminus \xi}^\T\widehat{\mbf{U}}_{\xi^\star \setminus \xi}^\T \left(\mbf{I} - \widehat{\mbf{U}}_{\xi} \widehat{\mbf{U}}_{\xi}^\dagger\right) \widehat{\mbf{U}}_{\xi^\star \setminus \xi} \bm{\beta}_{\xi^\star \setminus \xi}
\end{align*}
Note that $\widehat{\mbf{U}}_{\xi^\star \setminus \xi}^\T \left(\mbf{I} - \widehat{\mbf{U}}_{\xi} \widehat{\mbf{U}}_{\xi}^\dagger\right) \widehat{\mbf{U}}_{\xi^\star \setminus \xi}$ is the Schur complement of the principal submatrix $\widehat{\mbf{U}}_{\xi^\star \setminus \xi}^\T\widehat{\mbf{U}}_{\xi^\star \setminus \xi}$ in the matrix $\widehat{\mbf{U}}_{\xi \cup \xi^\star}^\T\widehat{\mbf{U}}_{\xi \cup \xi^\star}$. Thus, by Lemma \ref{lem:schur},
\begin{equation*}
\lambda_{\min}\left(\widehat{\mbf{U}}_{\xi^\star \setminus \xi}^\T \left(\mbf{I} - \widehat{\mbf{U}}_{\xi} \widehat{\mbf{U}}_{\xi}^\dagger\right) \widehat{\mbf{U}}_{\xi^\star \setminus \xi}\right) \ge \lambda_{\min}\left(\widehat{\mbf{U}}_{\xi \cup \xi^\star}^\T\widehat{\mbf{U}}_{\xi \cup \xi^\star}\right) \ge n\widehat{\kappa}_0.   
\end{equation*}
It further implies that
\begin{align*}
\left\Vert \left(\widehat{\mbf{U}}_{\xi \cup \xi^\star} \widehat{\mbf{U}}_{\xi \cup \xi^\star}^\dagger - \widehat{\mbf{U}}_{\xi} \widehat{\mbf{U}}_{\xi}^\dagger\right)\widehat{\mbf{U}}_{\xi^\star} \bm{\beta}^\star_{\xi^\star}\right\Vert^2
&\ge n\widehat{\kappa}_0 \Vert \bm{\beta}^\star_{\xi^\star \setminus \xi} \Vert^2\\
&\ge n\widehat{\kappa}_0 |\xi^\star \setminus \xi| \min_{j \in \xi^\star} |\beta^\star_j|^2\\
&\ge \sigma^{\star 2} \widehat{\kappa}_0 M_4^2 n \epsilon_n^2.
\end{align*}

Under the null hypothesis, write
\begin{align*}
&~~~\En \phi_{4n}\\
&\overset{(1)}{=} \Pn \left(\min_{\xi \not \supseteq \xi^\star:~|\xi\setminus\xi^\star| \le M_0s} \left\Vert \left(\widehat{\mbf{U}}_{\xi \cup \xi^\star} \widehat{\mbf{U}}_{\xi \cup \xi^\star}^\dagger - \widehat{\mbf{U}}_{\xi} \widehat{\mbf{U}}_{\xi}^\dagger\right)\left(\widehat{\mbf{U}}_{\xi^\star}\bm{\beta}^\star_{\xi^\star} + \sigma^\star \bm{\varepsilon}\right) \right\Vert  \le \sigma^\star \sqrt{\widehat{\kappa}_0}M_4 \sqrt{n}\epsilon_n/2\right)\\
&\le \Pn \left(\min_{\xi \not \supseteq \xi^\star:~|\xi\setminus\xi^\star| \le M_0s} \left\Vert \left(\widehat{\mbf{U}}_{\xi \cup \xi^\star} \widehat{\mbf{U}}_{\xi \cup \xi^\star}^\dagger - \widehat{\mbf{U}}_{\xi} \widehat{\mbf{U}}_{\xi}^\dagger\right) \widehat{\mbf{U}}_{\xi^\star} \bm{\beta}^\star_{\xi^\star} \right\Vert \right.\\
&~~~~~~~~~~~~~~~~~~~~~~~\left.- \max_{\xi \not \supseteq \xi^\star:~|\xi\setminus\xi^\star| \le M_0s} \left\Vert \left(\widehat{\mbf{U}}_{\xi \cup \xi^\star} \widehat{\mbf{U}}_{\xi \cup \xi^\star}^\dagger - \widehat{\mbf{U}}_{\xi} \widehat{\mbf{U}}_{\xi}^\dagger\right) \sigma^\star \bm{\varepsilon} \right\Vert \le \sigma^\star \sqrt{\widehat{\kappa}_0}M_4 \sqrt{n}\epsilon_n/2\right)\\
&\overset{(2)}{\le} \Pn \left(\max_{\xi \not \supseteq \xi^\star:~|\xi\setminus\xi^\star| \le M_0s} \left\Vert \left(\widehat{\mbf{U}}_{\xi \cup \xi^\star} \widehat{\mbf{U}}_{\xi \cup \xi^\star}^\dagger - \widehat{\mbf{U}}_{\xi} \widehat{\mbf{U}}_{\xi}^\dagger\right) \bm{\varepsilon} \right\Vert  \ge \sqrt{\widehat{\kappa}_0}M_4 \sqrt{n}\epsilon_n/2\right)\\
&\overset{(3)}{\le} \Pn \left( \left\Vert \widehat{\mbf{U}}_{\xi^\star}\widehat{\mbf{U}}_{\xi^\star}^\dagger \bm{\varepsilon} \right\Vert  \ge \sqrt{\widehat{\kappa}_0}M_4 \sqrt{n}\epsilon_n/2\right)\\
&= \Pn \left( \bm{\varepsilon}^\T \widehat{\mbf{U}}_{\xi^\star}\widehat{\mbf{U}}_{\xi^\star}^\dagger \bm{\varepsilon} \ge \widehat{\kappa}_0 M_4^2 n\epsilon_n^2/4\right)\\
&\overset{(4)}{=} \mbb{P}\left( \chi^2_s \ge \widehat{\kappa}_0 M_4^2 n\epsilon_n^2/4\right)
\end{align*}
(1) follows from the facts that $\mbf{Y} = \widehat{\mbf{F}} \bm{\alpha}^\star + \widehat{\mbf{U}}\bm{\beta}^\star + \sigma^\star \bm{\varepsilon}$ with $\bm{\beta}^\star_{\xi^{\star c}} = \mbf{0}$ under $\Pn$ and that $\widehat{\mbf{F}}^\T \widehat{\mbf{U}} = \mbf{0}$. (2) plugs in \eqref{schur}. (3) is due to the fact that
\begin{equation*}
\widehat{\mbf{U}}_{\xi \cup \xi^\star} \widehat{\mbf{U}}_{\xi \cup \xi^\star}^\dagger - \widehat{\mbf{U}}_{\xi} \widehat{\mbf{U}}_{\xi}^\dagger \le \widehat{\mbf{U}}_{\xi^\star} \widehat{\mbf{U}}_{\xi^\star}^\dagger.
\end{equation*}
(4) uses the fact that
$$\bm{\varepsilon}^\T \widehat{\mbf{U}}_{\xi^\star}\widehat{\mbf{U}}_{\xi^\star}^\dagger \bm{\varepsilon} \sim \chi^2_{|\xi^\star|}.$$
Applying Lemma \ref{lem:chi2}, part (b) yields
\begin{equation*}
\En \phi_{4n} \lesssim e^{-(\widehat{\kappa}_0 M_4^2/8 - \smallo(1)) n\epsilon_n^2} = e^{-(\widehat{\kappa}_0 M_4^2/8 -M_0 - \smallo(1)) s \log p}.
\end{equation*}

Under the alternative hypothesis, observe that $\phi_{4n} = \max_{\xi' \not \supseteq \xi^\star :~|\xi' \setminus \xi^\star| \le M_0s} \phi_{4n}^{\xi'}$, where
\begin{equation*}
\phi_{4n}^{\xi'} = 1\left\{ \left\Vert \left(\widehat{\mbf{U}}_{\xi' \cup \xi^\star} \widehat{\mbf{U}}_{\xi' \cup \xi^\star}^\dagger - \widehat{\mbf{U}}_{\xi'} \widehat{\mbf{U}}_{\xi'}^\dagger\right) \mbf{Y} \right\Vert  \le \sigma^\star \sqrt{\widehat{\kappa}_0}M_4 \sqrt{n}\epsilon_n/2\right\}.
\end{equation*}
Using Lemma \ref{lem:split},
$$\sup_{\Theta_{4n}}\Ea(1 - \phi_{4n}) \le \max_{\xi' \not \supseteq \xi^\star:~|\xi' \setminus \xi^\star| \le M_0s} ~\sup_{\Theta_{4n} \cap \{\xi = \xi'\}} \Ea (1-\phi_{4n}^{\xi'}).$$
For any $\xi' \not \supseteq \xi^\star$ such that $|\xi' \setminus \xi^\star| \le M_0s$ and any $(\sigma, \bm{\alpha}, \bm{\beta}) \in \Theta_{4n} \cap \{\xi = \xi'\}$, write
\begin{align*}
&~~~\Ea (1-\phi_{4n}^{\xi'})\\
&= \Pa \left(\left\Vert \left(\widehat{\mbf{U}}_{\xi \cup \xi^\star} \widehat{\mbf{U}}_{\xi \cup \xi^\star}^\dagger - \widehat{\mbf{U}}_{\xi} \widehat{\mbf{U}}_{\xi}^\dagger\right) \mbf{Y} \right\Vert  > \sigma^\star \sqrt{\widehat{\kappa}_0}M_4 \sqrt{n}\epsilon_n/2 \right)\\
&\overset{(1)}{=} \Pa \left(\left\Vert \left(\widehat{\mbf{U}}_{\xi \cup \xi^\star} \widehat{\mbf{U}}_{\xi \cup \xi^\star}^\dagger - \widehat{\mbf{U}}_{\xi} \widehat{\mbf{U}}_{\xi}^\dagger\right) \sigma \bm{\varepsilon} \right\Vert > \sigma^\star \sqrt{\widehat{\kappa}_0}M_4 \sqrt{n}\epsilon_n/2 \right)\\
&\overset{(2)}{\le} \Pa \left(\left\Vert \left(\widehat{\mbf{U}}_{\xi \cup \xi^\star} \widehat{\mbf{U}}_{\xi \cup \xi^\star}^\dagger - \widehat{\mbf{U}}_{\xi} \widehat{\mbf{U}}_{\xi}^\dagger\right) \bm{\varepsilon} \right\Vert  \ge \sqrt{\frac{1-M_1\epsilon_n}{1+M_1\epsilon_n}}\sqrt{\widehat{\kappa}_0}M_4 \sqrt{n}\epsilon_n/2 \right)\\
&\overset{(3)}{\le} \mbb{P}\left( \chi^2_s \ge \frac{1-M_1\epsilon_n}{1+M_1\epsilon_n} \times \widehat{\kappa}_0 M_4^2 n\epsilon_n^2/4\right).
\end{align*}
(1) follows from the facts that $\mbf{Y} = \widehat{\mbf{F}} \bm{\alpha}+ \widehat{\mbf{U}}\bm{\beta} + \sigma \bm{\varepsilon}$ with $\bm{\beta}_{\xi^c} = \mbf{0}$ under $\Pa$ and that $\widehat{\mbf{F}}^\T \widehat{\mbf{U}} = \mbf{0}$. (2) plugs in the restriction
$$\frac{\sigma^{\star 2}}{\sigma^2} > \frac{1-M_1\epsilon_n}{1+M_1\epsilon_n}$$
from the definition of $\Theta_{4n}$. (3) uses a similar argument to what we have used for the null hypothesis. Since the final bound in the last display is uniform for any $\xi' \not \supseteq \xi^\star$ such that $|\xi' \setminus \xi^\star| \le M_0s$ and any $(\sigma, \bm{\alpha}, \bm{\beta}) \in \Theta_{4n} \cap \{\xi = \xi'\}$, we apply Lemma \ref{lem:chi2}, part (b) and yield
$$\sup_{\Theta_{4n}} \Ea (1-\phi_{4n}) \lesssim e^{-(\widehat{\kappa}_0 M_4^2/8 - \smallo(1))n\epsilon_n^2} = e^{-(\widehat{\kappa}_0 M_4^2/8 - \smallo(1))s\log p}.$$
\end{proof}

\begin{proof}[Proof of Lemma \ref{lem:test}, part (e)]
Under the null hypothesis, write
\begin{align*}
&~~~\En \phi_{5n}\\
&\overset{(1)}{=} \Pn \left(\max_{\xi \supseteq \xi^\star:~|\xi\setminus\xi^\star| \le M_0s} \left\Vert {\widehat{\mbf{F}}^\dagger\bm{\varepsilon} \choose \widehat{\mbf{U}}_{\xi}^\dagger \bm{\varepsilon}} \right\Vert  \ge M_2 \epsilon_n/2\right)\\
&= \Pn \left(\max_{\xi \supseteq \xi^\star:~|\xi\setminus\xi^\star| \le M_0s} \bm{\varepsilon}^\T\left[ \widehat{\mbf{F}}^{\dagger \T}\widehat{\mbf{F}}^\dagger + \widehat{\mbf{U}}_{\xi}^{\dagger \T} \widehat{\mbf{U}}_{\xi}^\dagger\right] \bm{\varepsilon} \ge M_2^2 \epsilon_n^2/4\right)\\
&\overset{(2)}{\le} \Pn \left(\max_{\xi \supseteq \xi^\star:~|\xi\setminus\xi^\star| \le M_0s} \bm{\varepsilon}^\T \left[ \widehat{\mbf{F}}\widehat{\mbf{F}}^\dagger + \widehat{\mbf{U}}_{\xi} \widehat{\mbf{U}}_{\xi}^\dagger\right] \bm{\varepsilon} \ge \min\{\widehat{\kappa}_0, 1\} M_2^2 n\epsilon_n^2/4\right)\\
&\overset{(3)}{\le} \sum_{\xi \supseteq \xi^\star: ~|\xi \setminus \xi^\star| = M_0s} \Pn \left( \bm{\varepsilon}^\T \left[ \widehat{\mbf{F}}\widehat{\mbf{F}}^\dagger + \widehat{\mbf{U}}_{\xi} \widehat{\mbf{U}}_{\xi}^\dagger\right] \bm{\varepsilon} \ge \min\{\widehat{\kappa}_0, 1\} M_2^2 n\epsilon_n^2/4\right)\\
&\overset{(4)}{=} {p - s \choose M_0 s} \mbb{P}\left( \chi^2_{k+(M_0+1)s} \ge \min\{\widehat{\kappa}_0, 1\} M_2^2 n\epsilon_n^2/4\right)
\end{align*}
(1) follows from the facts that $\mbf{Y} = \widehat{\mbf{F}} \bm{\alpha}^\star + \widehat{\mbf{U}}\bm{\beta}^\star + \sigma^\star \bm{\varepsilon}$ with $\bm{\beta}^\star_{\xi^{\star c}} = \mbf{0}$ under $\Pn$ and that $\widehat{\mbf{F}}^\T \widehat{\mbf{U}} = \mbf{0}$. (2) is due to
\begin{align*}
\widehat{\mbf{F}}^{\dagger \T}\widehat{\mbf{F}}^\dagger
&\le \lambda_{\min}\left( \widehat{\mbf{F}}^\T\widehat{\mbf{F}} \right)^{-1} \widehat{\mbf{F}} \widehat{\mbf{F}}^\dagger = \widehat{\mbf{F}} \widehat{\mbf{F}}^\dagger/n\\
\widehat{\mbf{U}}_{\xi}^{\dagger \T}\widehat{\mbf{U}}_{\xi}^\dagger
&\le \lambda_{\min}\left( \widehat{\mbf{U}}_{\xi}^\T\widehat{\mbf{U}}_{\xi } \right)^{-1} \widehat{\mbf{U}}_{\xi} \widehat{\mbf{U}}_{\xi }^\dagger \le \widehat{\mbf{U}}_{\xi} \widehat{\mbf{U}}_{\xi}^\dagger/n{\widehat{\kappa}}_0.
\end{align*}
For (3), we observe that projection matrices
$$\widehat{\mbf{U}}_{\xi'}\widehat{\mbf{U}}_{\xi'}^\dagger \le \widehat{\mbf{U}}_{\xi''}\widehat{\mbf{U}}_{\xi''}^\dagger$$
for nested models $\xi' \subseteq \xi''$, and thus the term $\bm{\varepsilon}^\T \widehat{\mbf{U}}_{\xi}\widehat{\mbf{U}}_{\xi}^\dagger \bm{\varepsilon}$ achieves its maximum value at some $\xi \supseteq \xi^\star$ s.t. $|\xi \setminus \xi^\star| = M_0 s$. (4) uses the fact that $$\bm{\varepsilon}^\T\left[ \widehat{\mbf{F}}\widehat{\mbf{F}}^\dagger + \widehat{\mbf{U}}_{\xi}\widehat{\mbf{U}}_{\xi}^\dagger\right]\bm{\varepsilon} \sim \chi^2_{ k + |\xi|}.$$
Applying Lemma \ref{lem:chi2}, part (b) yields
$$\En \phi_{5n} \lesssim p^{M_0s}e^{-(\min\{\widehat{\kappa}_0,1\} M_2^2/8 - \smallo(1)) n\epsilon_n^2} = e^{-(\min\{\widehat{\kappa}_0,1\} M_2^2/8 -M_0 - \smallo(1)) s \log p}.$$

Under the alternative hypothesis, observe that $\phi_{5n} = \max_{\xi' \supseteq \xi^\star:~|\xi' \setminus \xi^\star| \le M_0s} \phi_{5n}^{\xi'}$, where
$$\phi_{5n}^{\xi'} = 1\left\{ \left\Vert {\widehat{\mbf{F}}^\dagger\mbf{Y} \choose \widehat{\mbf{U}}_{\xi'}^\dagger \mbf{Y}} - {\bm{\alpha}^\star \choose \bm{\beta}^\star_{\xi'}} \right\Vert  \ge \sigma^\star M_2 \epsilon_n/2\right\}.$$
Using Lemma \ref{lem:split},
$$\sup_{\Theta_{5n}}\Ea(1 - \phi_{5n}) \le \max_{\xi':~|\xi' \setminus \xi^\star| \le M_0s} ~\sup_{\Theta_{5n} \cap \{\xi = \xi'\}} \Ea (1-\phi_{5n}^{\xi'}).$$
For any $\xi' \supseteq \xi^\star$ such that $|\xi' \setminus \xi^\star| \le M_0s$ and any $(\sigma, \bm{\alpha}, \bm{\beta}) \in \Theta_{5n} \cap \{\xi = \xi'\}$, write
\begin{align*}
&~~~\Ea (1-\phi_{5n}^{\xi'})\\
&\overset{(1)}{=} \Pa \left(\left\Vert {\bm{\alpha} \choose \bm{\beta}_{\xi}} - {\bm{\alpha}^\star \choose \bm{\beta}^\star_{\xi}}  + \sigma {\widehat{\mbf{F}}^\dagger\bm{\varepsilon} \choose \widehat{\mbf{U}}_{\xi}^\dagger \bm{\varepsilon}}  \right\Vert  < \sigma^\star M_2 \epsilon_n/2\right)\\
&\overset{(2)}{\le} \Pa \left( \left\Vert {\widehat{\mbf{F}}^\dagger\bm{\varepsilon} \choose \widehat{\mbf{U}}_{\xi}^\dagger \bm{\varepsilon}} \right\Vert  \ge \sqrt{\frac{1-M_1\epsilon_n}{1+M_1\epsilon_n}} \times M_2 \epsilon_n/2\right)\\
&\overset{(3)}{\le} \mbb{P}\left( \chi^2_{k + (M_0+1)s} \ge \frac{1-M_1\epsilon_n}{1+M_1\epsilon_n} \times \min\{\widehat{\kappa}_0,1\} M_2^2 n\epsilon_n^2/4\right).
\end{align*}
(1) follows from the facts that $\mbf{Y} = \widehat{\mbf{F}} \bm{\alpha}+ \widehat{\mbf{U}}\bm{\beta} + \sigma \bm{\varepsilon}$ with $\bm{\beta}_{\xi^c} = \mbf{0}$ under $\Pa$, that $\widehat{\mbf{F}}^\T \widehat{\mbf{U}} = \mbf{0}$ and that $\xi \supseteq \xi^\star$. (2) plugs in the restrictions
$$\left\Vert {\bm{\alpha} \choose \bm{\beta}_{\xi}} - {\bm{\alpha}^\star \choose \bm{\beta}^\star_{\xi}} \right\Vert = \left\Vert {\bm{\alpha} \choose \bm{\beta}} - {\bm{\alpha}^\star \choose \bm{\beta}} \right\Vert > \sigma^\star M_2\epsilon_n, ~~~\frac{\sigma^{\star 2}}{\sigma^2} > \frac{1-M_1\epsilon_n}{1+M_1\epsilon_n}$$
from the definition of $\Theta_{5n}$. (3) uses a similar argument to what we have used for the null hypothesis. Since the final bound in the last display is uniform for any $\xi' \supseteq \xi^\star$ such that $|\xi' \setminus \xi^\star| \le M_0s$ and any $(\sigma, \bm{\alpha}, \bm{\beta}) \in \Theta_{5n} \cap \{\xi = \xi'\}$, we apply Lemma \ref{lem:chi2}, part (b) and yield
$$\sup_{\Theta_{5n}} \Ea (1-\phi_{5n}) \lesssim e^{-(\min\{\widehat{\kappa}_0,1\} M_2^2/8 - \smallo(1))n\epsilon_n^2} = e^{-(\min\{\widehat{\kappa}_0,1\}M_2^2/8 - \smallo(1))s\log p}.$$
\end{proof}

\begin{proof}[Proof of Lemma \ref{lem:merging}]
Define
\begin{equation*}
A^\star_n(\eta_1, \eta_2) = \left\{(\sigma, \bm{\alpha}, \bm{\beta}):
\begin{split}
& \sigma^2/\sigma^{\star 2} \in [1, 1+\eta_1\epsilon_n^2],\\
&\xi = \xi^\star,\\
& |\alpha_j - \alpha^\star_j | \le  \sigma \eta_2 \epsilon_n/\sqrt{k}, j = 1,\dots,k\\
& |\beta_j - \beta^\star_j | \le \tau_j\sigma \eta_2 \epsilon_n/\sqrt{s}, j \in \xi^\star,
\end{split}
\right\}
\end{equation*}
\noindent \textbf{Step 1.} We first choose sufficiently small $\eta_1, \eta_2$ such that
$$\Pn \left( \inf_{ (\sigma, \bm{\alpha}, \bm{\beta}) \in A^\star_n(\eta_1,\eta_2)} \frac{\mc{N}(\mbf{Y}|\widehat{\mbf{F}}\bm{\alpha} + \widehat{\mbf{U}}\bm{\beta}, \sigma^2\mbf{I})}{\mc{N}(\mbf{Y}|\widehat{\mbf{F}}\bm{\alpha}^\star + \widehat{\mbf{U}}\bm{\beta}^\star, \sigma^{\star 2}\mbf{I})} \ge e^{-C_5 s\log p/2}\right) \lesssim e^{-C_5's\log p}.$$
Observing that
\begin{align*}
&~~~\Pn\left(\bm{\varepsilon}^\T \left[\widehat{\mbf{F}}\widehat{\mbf{F}}^\dagger + \widehat{\mbf{U}}_{\xi^\star}\widehat{\mbf{U}}_{\xi^\star}^\dagger\right] \bm{\varepsilon} > 3C_5'n \epsilon_n^2\right)\\
&= \mbb{P}\left( \chi^2_{k + s} > 3C_5' n \epsilon_n^2\right) \lesssim e^{-(3C_5'/2-\smallo(1))n \epsilon_n^2} \le e^{-C_5'n \epsilon_n^2} = e^{-C_5's\log p}
\end{align*}
We proceed to find sufficiently small $\eta_1, \eta_2$ such that
$$\inf_{ (\sigma, \bm{\alpha}, \bm{\beta}) \in A^\star_n(\eta_1,\eta_2)} \log \frac{\mc{N}(\mbf{Y}|\widehat{\mbf{F}}\bm{\alpha} + \widehat{\mbf{U}}\bm{\beta}, \sigma^2\mbf{I})}{\mc{N}(\mbf{Y}|\widehat{\mbf{F}}\bm{\alpha}^\star + \widehat{\mbf{U}}\bm{\beta}^\star, \sigma^{\star 2}\mbf{I})} \ge - C_5 s\log p/2 = - C_5 n\epsilon_n^2/2$$
with conditional probability $1$ given $\bm{\varepsilon}^\T \left[\widehat{\mbf{F}}\widehat{\mbf{F}}^\dagger + \widehat{\mbf{U}}_{\xi^\star}\widehat{\mbf{U}}_{\xi^\star}^\dagger\right] \bm{\varepsilon} \le 3C_5'n\epsilon_n^2$. To this end, write
\begin{align*}
&~~~ - \log \frac{\mc{N}(\mbf{Y}|\widehat{\mbf{F}}\bm{\alpha} + \widehat{\mbf{U}}\bm{\beta}, \sigma^2\mbf{I})}{\mc{N}(\mbf{Y}|\widehat{\mbf{F}}\bm{\alpha}^\star + \widehat{\mbf{U}}\bm{\beta}^\star, \sigma^{\star 2}\mbf{I})}\\
&= \Vert \mbf{Y} - \widehat{\mbf{F}}\bm{\alpha} - \widehat{\mbf{U}}\bm{\beta}\Vert^2/2\sigma^2 - \Vert \mbf{Y} - \widehat{\mbf{F}}\bm{\alpha}^\star - \widehat{\mbf{U}}\bm{\beta}^\star\Vert^2/2\sigma^{\star 2} + n\log(\sigma^2/\sigma^{\star 2})\\
&= \Vert \sigma^\star \bm{\varepsilon} + \widehat{\mbf{F}}(\bm{\alpha}^\star - \bm{\alpha}) + \widehat{\mbf{U}}_{\xi^\star}(\bm{\beta}^\star_{\xi^\star} - \bm{\beta}_{\xi^\star})\Vert^2/2\sigma^2 - \Vert \bm{\varepsilon} \Vert^2/2 + n\log(\sigma^2/\sigma^{\star 2})\\
&\le \Vert \widehat{\mbf{F}}(\bm{\alpha}^\star - \bm{\alpha}) \Vert^2/2\sigma^2 + \Vert \widehat{\mbf{U}}_{\xi^\star}(\bm{\beta}^\star_{\xi^\star} - \bm{\beta}_{\xi^\star}) \Vert^2/2\sigma^2\\
&~~~+ \bm{\varepsilon}^\T \widehat{\mbf{F}}( \bm{\alpha}^\star - \bm{\alpha})/\sigma \times \sigma^\star/\sigma + \bm{\varepsilon}^\T \widehat{\mbf{U}}_{\xi^\star} (\bm{\beta}^\star_{\xi^\star} - \bm{\beta}_{\xi^\star})/\sigma \times \sigma^\star/\sigma + \eta_1 n\epsilon_n^2,
\end{align*}
where
\begin{align*}
\Vert \widehat{\mbf{F}}(\bm{\alpha}^\star - \bm{\alpha}) \Vert^2/2\sigma^2
&\le \lambda_{\max}\left(\widehat{\mbf{F}}^\T \widehat{\mbf{F}}\right) \Vert \bm{\alpha}^\star - \bm{\alpha} \Vert^2/2\sigma^2 \le \eta_2^2 n\epsilon_n^2/2\\
\Vert \widehat{\mbf{U}}_{\xi^\star}(\bm{\beta}^\star_{\xi^\star} - \bm{\beta}_{\xi^\star}) \Vert^2/2\sigma^2
&\le \lambda_{\max}\left(\widehat{\mbf{U}}_{\xi^\star}^\T \widehat{\mbf{U}}_{\xi^\star}\right) \Vert \bm{\beta}^\star_{\xi^\star} - \bm{\beta}_{\xi^\star} \Vert^2/2\sigma^2 \le \eta_2^2 \widehat{\kappa}_1 n\epsilon_n^2/2\\
\bm{\varepsilon}^\T \widehat{\mbf{F}}( \bm{\alpha}^\star - \bm{\alpha})/\sigma \times \sigma^\star/\sigma
&= \bm{\varepsilon}^\T \widehat{\mbf{F}}\widehat{\mbf{F}}^\dagger \widehat{\mbf{F}}( \bm{\alpha}^\star - \bm{\alpha})/\sigma \times \sigma^\star/\sigma\\
&\le \Vert \widehat{\mbf{F}}^\dagger\widehat{\mbf{F}}^\T \bm{\varepsilon} \Vert \times \Vert \widehat{\mbf{F}}(\bm{\alpha}^\star - \bm{\alpha})/\sigma \Vert \times 1\\
&\le \sqrt{3C_5'}\sqrt{n}\epsilon_n \times \eta_2\sqrt{n}\epsilon_n = \sqrt{3C_5'}\eta_2 n\epsilon_n^2\\
\bm{\varepsilon}^\T \widehat{\mbf{U}}_{\xi^\star}( \bm{\beta}_{\xi^\star}^\star - \bm{\beta}_{\xi^\star})/\sigma \times \sigma^\star/\sigma
&= \bm{\varepsilon}^\T \widehat{\mbf{U}}_{\xi^\star}\widehat{\mbf{U}}_{\xi^\star}^\dagger \widehat{\mbf{U}}_{\xi^\star}(\bm{\beta}^\star - \bm{\beta})/\sigma \times \sigma^\star/\sigma\\
&\le \Vert \widehat{\mbf{U}}_{\xi^\star}^\dagger\widehat{\mbf{U}}_{\xi^\star}^\T \bm{\varepsilon} \Vert \times \Vert \widehat{\mbf{U}}_{\xi^\star}(\bm{\beta}_{\xi^\star}^\star - \bm{\beta}_{\xi^\star})/\sigma \Vert \times 1\\
&\le \sqrt{3C_5'}\sqrt{n}\epsilon_n \times \eta_2\sqrt{\widehat{\kappa}_1}\sqrt{n}\epsilon_n = \sqrt{3C_5'\widehat{\kappa}_1}\eta_2 n\epsilon_n^2
\end{align*}
We choose sufficiently small $\eta_1, \eta_2$ such that $(1+\widehat{\kappa}_1)\eta_2^2/2 + \sqrt{3C_5'} (1+\sqrt{\widehat{\kappa}_1}) \eta_2 + \eta_1 \le C_5/2$.

\noindent \textbf{Step 2.}
Since
\begin{align*}
&~~~\int \frac{\mc{N}(\mbf{Y}|\widehat{\mbf{F}}\bm{\alpha} + \widehat{\mbf{U}}\bm{\beta}, \sigma^2\mbf{I})}{\mc{N}(\mbf{Y}|\widehat{\mbf{F}}\bm{\alpha}^\star + \widehat{\mbf{U}}\bm{\beta}^\star, \sigma^{\star 2}\mbf{I})} d\pi(\sigma, \bm{\alpha}, \bm{\beta})\\
&\ge \pi(A^\star_n(\eta_1,\eta_2)) \inf_{ (\sigma, \bm{\alpha}, \bm{\beta}) \in A^\star_n(\eta_1,\eta_2)} \frac{\mc{N}(\mbf{Y}|\widehat{\mbf{F}}\bm{\alpha} + \widehat{\mbf{U}}\bm{\beta}, \sigma^2\mbf{I})}{\mc{N}(\mbf{Y}|\widehat{\mbf{F}}\bm{\alpha}^\star + \widehat{\mbf{U}}\bm{\beta}^\star, \sigma^{\star 2}\mbf{I})},
\end{align*}
it is left to show
$$\pi(A^\star_n(\eta_1,\eta_2)) \gtrsim e^{-C_5s\log p/2}.$$
Note that $\Vert \bm{\alpha}^\star \Vert = \bigO(1)$, $\Vert \bm{\beta}^\star \Vert = \bigO(1)$, and for $j \in \xi^\star$, $\tau_j^{-1} = \Vert \widehat{\mbf{U}}_j\Vert/\sqrt{n} \in [\sqrt{\widehat{\kappa}_0}, \sqrt{\widehat{\kappa}_1}]$. For all $(\sigma, \bm{\alpha}, \bm{\beta}) \in A^\star_n(\eta_1,\eta_2)$, we can find constant $C > 0$ such that
\begin{align*}
& |\alpha_j /\sigma| \le |\alpha^\star_j / \sigma| + \eta_2 \epsilon_n/\sqrt{k} \le C, j = 1,\dots,k\\
& |\beta_j / \tau_j\sigma| \le |\beta^\star_j / \tau_j\sigma | + \eta_2 \epsilon_n/\sqrt{s} \le C, j \in \xi^\star
\end{align*}
hold for sufficiently large $n$. Thus
\begin{align*}
\pi(A^\star_n(\eta_1, \eta_2))
&= \left(\frac{1}{p}\right)^s \int_{\sigma^{\star 2}}^{\sigma^{\star 2}(1+\eta_1 \epsilon_n^2)} g(\sigma^2)d\sigma^2 \times \prod_{j=1}^k \int_{\alpha^\star_j/ \sigma - \eta_2 \epsilon_n/\sqrt{k}}^{\alpha^\star_j / \sigma + \eta_2 \epsilon_n/\sqrt{k}} h\left(\frac{\alpha_j}{\sigma}\right)d\left(\frac{\alpha_j}{\sigma}\right)\\
&~~~\times \prod_{j \in \xi^\star} \int_{\beta^\star_j/ \tau_j\sigma - \eta_2 \epsilon_n/\sqrt{s}}^{\beta^\star_j / \tau_j \sigma + \eta_2 \epsilon_n/\sqrt{s}} h\left(\frac{\beta_j}{\tau_j\sigma}\right)d\left(\frac{\beta_j}{\tau_j\sigma}\right)\\
&\gtrsim \left(\frac{1}{p}\right)^s \times \sigma^{\star 2}\eta_1\epsilon_n^2 g(\sigma^{\star 2})/2 \times \left( \frac{2\eta_2 \epsilon_n}{\sqrt{k}} \inf_{|z| \le C} h(z)\right)^k \times \left( \frac{2\eta_2 \epsilon_n}{\sqrt{s}} \inf_{|z| \le C} h(z)\right)^s\\
&\gtrsim C' \left(\frac{1}{p}\right)^s \times \epsilon_n^2 \times \left( \frac{1}{\sqrt{p}}  \right)^s\\
&\gtrsim e^{-C_5s\log p/2},
\end{align*}
if $C_5 > 3$.
\end{proof}

\section{Technical Proofs for Factor Model Estimation}
This section is devoted to the proofs of Theorem \ref{thm:3}. Parts (a) and (b) of Theorem \ref{thm:3} are restated as Lemmas \ref{lem:eigenvalue}-\ref{lem:eigenspace}. The proof of Lemma \ref{lem:eigenvalue} is straightforward. To prove Lemma \ref{lem:eigenspace}, we generalize the Davis-Kahan theorem \citep{davis1970rotation,yu2014useful} as Proposition \ref{prop:davis-kahan} and apply it to bound the principal angles from the perturbed eigenspace to the target eigenspace. Two preliminary lemmas, required by the proof of Lemma \ref{lem:eigenspace}, are stated as Lemmas \ref{lem:UB}-\ref{lem:D}. Parts (c) and (d) of Theorem \ref{thm:3}, restated as Lemmas \ref{lem:F}-\ref{lem:U}, are immediate corollaries of Lemma \ref{lem:eigenspace}.

\begin{lemma}[Theorem \ref{thm:3}, part (a)] \label{lem:eigenvalue}
Suppose Assumptions \ref{asm:5}-\ref{asm:7}. Recall that $\widehat{\lambda}_1 \ge \dots \ge \widehat{\lambda}_n$ are the $n$ eigenvalues of $\mbf{X}\mbf{X}^\T/n$, and that $\lambda_1 \ge \dots \ge \lambda_k$ are $k$ eigenvalues of $\mbf{B}^\T\mbf{B}$. Then
\begin{align*}
\max_{j=1}^k |\widehat{\lambda}_j - \lambda_j | &= \bigOp(p\sqrt{\log p/n}),\\
\max_{j=k+1}^n |\widehat{\lambda}_j - 0| &= \bigOp(p\sqrt{\log p/n}).
\end{align*}
\end{lemma}
\begin{proof}
It suffices to show $\Vert \mbf{X}^\T \mbf{X}/n - \mbf{B}\mbf{B}^\T\Vert = \bigOp(p\sqrt{\log p/n})$ so that Weyl's inequality applies. To this end, write
\begin{align*}
&~~~\mbf{X}^\T \mbf{X}/n - \mbf{B}\mbf{B}^\T = (\mbf{F}\mbf{B}^\T+\mbf{U})^\T(\mbf{F}\mbf{B}^\T+\mbf{U})/n - \mbf{B}\mbf{B}^\T\\
&= \mbf{B}(\mbf{F}^\T\mbf{F}/n - \mbf{I})\mbf{B}^\T
+ \mbf{U}^\T\mbf{F}\mbf{B}^\T/n + \mbf{B}\mbf{F}^\T\mbf{U}/n + \mbf{U}^\T\mbf{U}/n,
\end{align*}
where
\begin{align*}
\Vert \mbf{B}(\mbf{F}^\T\mbf{F}/n - \mbf{I})\mbf{B}^\T \Vert
&\le \Vert \mbf{B} \Vert \Vert \mbf{F}^\T\mbf{F}/n - \mbf{I} \Vert \Vert \mbf{B}^\T \Vert \le \lambda_1 k \Vert \mbf{F}^\T\mbf{F}/n - \mbf{I} \Vert_{\max}\\
\Vert \mbf{U}^\T\mbf{F}\mbf{B}^\T/n \Vert  = \Vert \mbf{B}\mbf{F}^\T\mbf{U}/n\Vert
&\le \Vert \mbf{B} \Vert \Vert \mbf{F}^\T\mbf{U}/n\Vert \le \sqrt{\lambda_1} \sqrt{pk} \Vert \mbf{F}^\T\mbf{U}/n\Vert_{\max}\\
\Vert \mbf{U}^\T\mbf{U}/n\Vert &\le \Vert \mbf{U}^\T\mbf{U}/n - \mbf{\Sigma}\Vert + \Vert \mbf{\Sigma} \Vert \le p \Vert \mbf{U}^\T\mbf{U}/n - \mbf{\Sigma}\Vert_{\max} + \Vert \mbf{\Sigma}\Vert.
\end{align*}
Plugging into it rates in Assumptions \ref{asm:5}-\ref{asm:7} and $\Vert \mbf{\Sigma}\Vert \le \Vert \mbf{\Sigma}\Vert_1 \le m_q(p)C_0^{1-q} = o(\log p)$ completes the proof.
\end{proof}

\begin{proposition} \label{prop:davis-kahan} 
Let $\widehat{\mbf{A}}$ be an $n \times n$ symmetric matrix with eigenvalues $\widehat{\lambda}_1 \ge \widehat{\lambda}_2 \ge \dots \ge \widehat{\lambda}_n$ and corresponding eigenvectors $\widehat{\psi}_1, \dots, \widehat{\psi}_n$. Fix $1 \le l \le r \le n$ and assume that $\min \{\widehat{\lambda}_{l-1} - \widehat{\lambda}_{l}, \widehat{\lambda}_{r} - \widehat{\lambda}_{r+1}\} > 0$, where $\widehat{\lambda}_0 := +\infty$ and $\widehat{\lambda}_{n+1} := -\infty$. Let $k = l - r + 1$. Let $\widehat{\mbf{\Lambda}} = \diag(\widehat{\lambda}_l, \dots, \widehat{\lambda}_r)$ and $\widehat{\mbf{\Lambda}}_c$ consists of the other $n - k$ eigenvalues of $\widehat{\mbf{A}}$. Let  $\widehat{\mbf{\Psi}} = (\widehat{\psi}_l, \dots, \widehat{\psi}_r)$ and $\widehat{\mbf{\Psi}}_c$ consists of the other $n - k$ eigenvectors of $\widehat{\mbf{A}}$. Let $\mbf A$ be an $n \times n$ (not necessarily symmetric) matrix with ``$\mbf{\Delta}$-approximate'' eigenvalues $\lambda_l \ge \dots \ge \lambda_r$ in the sense that 
$$\mbf{A} \mbf{\Psi} = \mbf{\Psi} \mbf{\Lambda} + \mbf \Delta,$$
where $\mbf{\Lambda} = \diag(\lambda_l, \dots, \lambda_r)$ and $\mbf{\Psi} = (\psi_l, \dots, \psi_r)$ consists of $k$ (not necessarily orthonormal) vectors. Then
$$\Vert \widehat{\mbf{\Psi}}_c^\T \mbf{\Psi}\Vert_\text{F} \le \frac{\Vert \mbf \Delta \Vert_\text{F} + \sqrt{k}\Vert \mbf{\Psi} \Vert \left(\Vert \widehat{\mbf A} - \mbf A \Vert + \Vert \widehat{\mbf{\Lambda}} - \mbf{\Lambda} \Vert_{\max}\right)}{\min \{\widehat{\lambda}_{l-1} - \widehat{\lambda}_{l}, \widehat{\lambda}_{r} - \widehat{\lambda}_{r+1}\}}$$
\end{proposition}

\begin{proof}
Write
$$\mbf \Delta = \mbf{A} \mbf{\Psi} - \mbf{\Psi}\mbf{\Lambda} = \widehat{\mbf A} \mbf{\Psi} - \mbf{\Psi} \widehat{\mbf{\Lambda}} + ( \mbf{A} - \widehat{\mbf{A}}) \mbf{\Psi} - \mbf{\Psi} ( \mbf{\Lambda} -\widehat{\mbf{\Lambda}}).$$
Using the facts that $\Vert \mbf{T}_1 \mbf{T}_2\Vert_\text{F} \le \Vert \mbf{T}_1 \Vert \Vert \mbf{T}_2\Vert_\text{F}$ and that $\Vert \mbf{T}_1 \Vert \le \Vert \mbf{T}_1 \Vert_\text{F} \le \sqrt{\rank (\mbf{T}_1)}\Vert \mbf{T}_1 \Vert$, we derive that
\begin{align*}
\Vert \widehat{\mbf A} \mbf{\Psi} - \mbf{\Psi} \widehat{\mbf{\Lambda}} \Vert_\text{F}
&\le \Vert \mbf \Delta \Vert_\text{F} + \Vert (\widehat{\mbf A} - \mbf A) \mbf{\Psi} \Vert_\text{F} + \Vert \mbf{\Psi} (\widehat{\mbf{\Lambda}} - \mbf{\Lambda}) \Vert_\text{F}\\
&\le \Vert \mbf \Delta \Vert_\text{F} + \Vert \widehat{\mbf A} - \mbf A \Vert \Vert \mbf{\Psi} \Vert_\text{F} + \Vert \mbf{\Psi} \Vert \Vert \widehat{\mbf{\Lambda}} - \mbf{\Lambda} \Vert_\text{F},\\
&\le \Vert \mbf \Delta \Vert_\text{F} + \Vert \widehat{\mbf A} - \mbf A \Vert \sqrt{k}\Vert \mbf{\Psi} \Vert  + \sqrt{k} \Vert \mbf{\Psi} \Vert  \Vert \widehat{\mbf{\Lambda}} - \mbf{\Lambda} \Vert_{\max},
\end{align*}
which is the numerator on the right hand side of the claimed inequality. On the other hand,
\begin{align*}
\Vert \widehat{\mbf A} \mbf{\Psi} - \mbf{\Psi} \widehat{\mbf{\Lambda}} \Vert_\text{F}
&= \Vert \widehat{\mbf{\Psi}} \widehat{\mbf{\Psi}}^\T \mbf{\Psi} \widehat{\mbf{\Lambda}} + \widehat{\mbf{\Psi}}_c \widehat{\mbf{\Psi}}_c^\T \mbf{\Psi} \widehat{\mbf{\Lambda}}- \widehat{\mbf{\Psi}} \widehat{\mbf{\Lambda}} \widehat{\mbf{\Psi}}^\T \mbf{\Psi} - \widehat{\mbf{\Psi}}_c \widehat{\mbf{\Lambda}}_c \widehat{\mbf{\Psi}}_c^\T \mbf{\Psi} \Vert_\text{F}\\
&\ge \Vert \widehat{\mbf{\Psi}}_c \widehat{\mbf{\Psi}}_c^\T \mbf{\Psi} \widehat{\mbf{\Lambda}} - \widehat{\mbf{\Psi}}_c \widehat{\mbf{\Lambda}}_c \widehat{\mbf{\Psi}}_c^\T \mbf{\Psi}  \Vert_\text{F}\\
&= \Vert \widehat{\mbf{\Psi}}_c^\T \mbf{\Psi} \widehat{\mbf{\Lambda}} - \widehat{\mbf{\Lambda}}_c \widehat{\mbf{\Psi}}_c^\T \mbf{\Psi}  \Vert_\text{F}
\end{align*}
where the first (in)equality follows from the identities that $\mbf{I} = \widehat{\mbf{\Psi}} \widehat{\mbf{\Psi}}^\T+ \widehat{\mbf{\Psi}}_c \widehat{\mbf{\Psi}}_c^\T$ and that $\widehat{\mbf{A}} = \widehat{\mbf{\Psi}} \widehat{\mbf{\Lambda}}\widehat{\mbf{\Psi}}^\T + \widehat{\mbf{\Psi}}_c \widehat{\mbf{\Lambda}}_c \widehat{\mbf{\Psi}}_c^\T$, the second (in)equality uses the orthogonality of $[\widehat{\mbf{\Psi}}, \widehat{\mbf{\Psi}}_c]$ to derive that
\begin{align*}
\Vert \widehat{\mbf{\Psi}} \mbf{T}_1 + \widehat{\mbf{\Psi}}_c \mbf{T}_2 \Vert_\text{F}
&= \trace^{1/2}\left[(\widehat{\mbf{\Psi}} \mbf{T}_1 + \widehat{\mbf{\Psi}}_c \mbf{T}_2 )^\T(\widehat{\mbf{\Psi}} \mbf{T}_1 +\widehat{\mbf{\Psi}}_c \mbf{T}_2) \right]\\
&= \trace^{1/2}\left[\mbf{T}_1^\T \widehat{\mbf{\Psi}}^\T \widehat{\mbf{\Psi}} \mbf{T}_1 + \mbf{T}_2^\T \widehat{\mbf{\Psi}}_c^\T \widehat{\mbf{\Psi}}_c \mbf{T}_2\right]\\
&\ge \trace^{1/2}\left[\mbf{T}_2^\T \widehat{\mbf{\Psi}}_c^\T \widehat{\mbf{\Psi}}_c \mbf{T}_2\right] = \Vert \widehat{\mbf{\Psi}}_c \mbf{T}_2 \Vert_\text{F},
\end{align*}
and the third (in)equality uses the column orthonormality of $\widehat{\mbf{\Psi}}_c$ again.

Proceed to consider the term $\Vert \widehat{\mbf{\Psi}}_c^\T \mbf{\Psi} \widehat{\mbf{\Lambda}} - \widehat{\mbf{\Lambda}}_c \widehat{\mbf{\Psi}}_c^\T \mbf{\Psi}  \Vert_\text{F}$. For real matrices $\mbf{T}_1, \mbf{T}_2, \mbf{T}_3$, we write $\vectorize(\mbf{T}_1)$ as the vectorization of $\mbf{T}_1$, which is the vector obtained by stacking columns of $\mbf{T}_1$, and denote by $\mbf{T}_1 \otimes \mbf{T}_2$ the kronecker product of matrices $\mbf{T}_1$ and $\mbf{T}_2$. Using the identity $\vectorize(\mbf{T}_1 \mbf{T}_2 \mbf{T}_3) = \mbf{T}_3^\T \otimes \mbf{T}_1 \vectorize(\mbf{T}_2)$ for any matrices $\mbf{T}_1, \mbf{T}_2, \mbf{T}_3$ with appropriate dimensions, we have
\begin{align*}
\Vert \widehat{\mbf{\Psi}}_c^\T \mbf{\Psi} \widehat{\mbf{\Lambda}} - \widehat{\mbf{\Lambda}}_c \widehat{\mbf{\Psi}}_c^\T \mbf{\Psi}  \Vert_\text{F}
&= \Vert \vectorize(\mbf{I}_{n-k}\widehat{\mbf{\Psi}}_c^\T \mbf{\Psi} \widehat{\mbf{\Lambda}})  - \vectorize(\widehat{\mbf{\Lambda}}_c \widehat{\mbf{\Psi}}_c^\T \mbf{\Psi} \mbf{I}) \Vert\\
&= \Vert \widehat{\mbf{\Lambda}} \otimes \mbf{I}_{n-k} \vectorize(\widehat{\mbf{\Psi}}_c^\T \mbf{\Psi})  - \mbf{I} \otimes \widehat{\mbf{\Lambda}}_c \vectorize( \widehat{\mbf{\Psi}}_c^\T\mbf {\Psi}) \Vert\\
&\ge \min \{\widehat{\lambda}_{l-1} - \widehat{\lambda}_{l}, \widehat{\lambda}_{r} - \widehat{\lambda}_{r+1}\} \Vert \vectorize( \widehat{\mbf{\Psi}}_c^\T\mbf {\Psi})\Vert\\
&= \min \{\widehat{\lambda}_{l-1} - \widehat{\lambda}_{l}, \widehat{\lambda}_{r} - \widehat{\lambda}_{r+1}\} \Vert \widehat{\mbf{\Psi}}_c^\T\mbf {\Psi}\Vert_\text{F},
\end{align*}
which is the left hand side times the denominator on the right hand side in the claimed inequality.
\end{proof}

\begin{lemma}[Theorem \ref{thm:3}, part (b)] \label{lem:eigenspace}
Suppose Assumptions \ref{asm:5}-\ref{asm:7} hold. Recall that $\widetilde{\mbf{F}}$ consists of $\sqrt{n}$-scaled left singular vectors of $\mbf{F}$, and that $\widehat{\mbf{F}}$ consists of $\sqrt{n}$-scalded top $k$ left singular vectors of $\mbf{X}$. Then $\widehat{\mbf{F}}$ recovers the column space of the latent common factor matrix $\mbf{F}$ in the sense that
$$\Vert \sin \angle (\widehat{\mbf{F}}/\sqrt{n}, \widetilde{\mbf{F}}/\sqrt{n})\Vert^2  = \bigOp(\log p/n).$$
\end{lemma}
\begin{proof}
Let $\widehat{\mbf{F}}_c$ consist of $\sqrt{n}$-scaled left singular vectors of $\mbf{X}$ except those in $\widetilde{\mbf{F}}$, then
$$\Vert \sin \angle (\widehat{\mbf{F}}/\sqrt{n}, \widetilde{\mbf{F}}/\sqrt{n})\Vert^2 = k - \Vert \widehat{\mbf{F}}^\T\widetilde{\mbf{F}}/n\Vert_\text{F}^2 = \Vert \widehat{\mbf{F}}_c^\T\widetilde{\mbf{F}}/n\Vert_\text{F}^2.$$
Thus it suffices to show $\Vert \widehat{\mbf{F}}_c^\T\widetilde{\mbf{F}}/n\Vert_\text{F} = \bigOp(\sqrt{\log p/n})$. For this goal, we first apply Proposition \ref{prop:davis-kahan} to show that
\begin{equation}\label{step1}
\Vert \widehat{\mbf{F}}_c^\T \mbf{F} / n \Vert_\text{F} = \bigOp(\sqrt{\log p/n}).
\end{equation}
Recall that $\mbf{R}$ consist of the right singular vectors of $\mbf{B}$, i.e. $\mbf{B}^\T\mbf{B} = \mbf{R} \mbf{\Lambda} \mbf{R}^\T$, and write
\begin{align*}
\frac{\mbf{X} \mbf{X}^\T}{n} \left(\frac{\widehat{\mbf{F}}}{\sqrt{n}}\right) &= \left(\frac{\widehat{\mbf{F}}}{\sqrt{n}}\right)\widehat{\mbf{\Lambda}},\\
\frac{\mbf{F} \mbf{B}^\T \mbf{B} \mbf{F}^\T}{n} \left(\frac{\mbf{F} \mbf{R}}{\sqrt{n}}\right) &= \left(\frac{\mbf{F} \mbf{R}}{\sqrt{n}}\right)\mbf{\Lambda} + \mbf{\Delta}
\end{align*}
where $\mbf{\Delta} = \mbf{F} \mbf{R} \mbf{\Lambda} \left(\mbf{R}^\T \mbf{F}^\T \mbf{F} \mbf{R}/ n - \mbf{I}\right) / \sqrt{n}$.  Applying Proposition \ref{prop:davis-kahan} yields
\begin{align*}
&~~~\Vert \widehat{\mbf{F}}_c^\T \mbf{F} / n \Vert_\text{F} = \Vert \widehat{\mbf{F}}_c^\T \mbf{F}\mbf{R}/ n \Vert_\text{F}\\
&\le \frac{\Vert \mbf{\Delta} \Vert_\text{F} + \sqrt{k} \Vert \mbf{F}\mbf{R}/\sqrt{n} \Vert \left( \Vert \mbf{X} \mbf{X}^\T/n -  \mbf{F} \mbf{B}^\T \mbf{B} \mbf{F}^\T/n\Vert + \Vert \mbf{\Lambda} - \widehat{\mbf{\Lambda}}\Vert_{\max}\right)}{\widehat{\lambda}_k - \widehat{\lambda}_{k+1}}.
\end{align*}
To prove \eqref{step1}, we are going to bound each term in the last display by Assumptions \ref{asm:5}-\ref{asm:7}, and Lemmas \ref{lem:eigenvalue},\ref{lem:UB}.
\begin{enumerate}[label=(\alph*)]
\item For the term $\Vert \mbf{F}\mbf{R}/\sqrt{n} \Vert = \Vert \mbf{F}/\sqrt{n} \Vert$, we have by Assumption \ref{asm:7} that
\begin{align*}
\Vert \mbf{F}^\T\mbf{F} / n - \mbf{I}\Vert &\le k \Vert \mbf{F}^\T\mbf{F} / n - \mbf{I}\Vert_{\max} = \bigOp(\sqrt{\log p / n}),\\
\left|\Vert \mbf{F}\mbf{R} / \sqrt{n} \Vert^2 - 1\right| &= \left| \Vert \mbf{F}^\T\mbf{F} / n\Vert -  \Vert \mbf{I} \Vert \right| \le \Vert \mbf{F}^\T\mbf{F} / n - \mbf{I}\Vert  = \bigOp(\sqrt{\log p / n}).
\end{align*}
\item Using the facts that $\Vert \mbf{T}_1 \mbf{T}_2\Vert_\text{F} = \Vert \mbf{T}_1 \Vert_\text{F} \Vert \mbf{T_2} \Vert$, that $\Vert \mbf{R} \Vert = 1$ and that $ \Vert \mbf{T}_1 \Vert_\text{F} \le \sqrt{\rank(\mbf{T}_1)} \Vert \mbf{T}_1 \Vert$,
\begin{align*}
\Vert \mbf{\Delta} \Vert_\text{F}
&\le \Vert \mbf{F} \mbf{R}/ \sqrt{n} \Vert_\text{F} \Vert \mbf{\Lambda} \Vert \Vert \mbf{R}^\T \mbf{F}^\T \mbf{F} \mbf{R}/ n - \mbf{I} \Vert\\
&= \sqrt{k}\Vert \mbf{F} / \sqrt{n} \Vert \Vert \mbf{\Lambda} \Vert \Vert \mbf{F}^\T \mbf{F} / n - \mbf{I}\Vert = \bigOp(p\sqrt{\log p / n}).
\end{align*}
\item For the term
$$\mbf{X} \mbf{X}^\T/n -  \mbf{F} \mbf{B}^\T \mbf{B} \mbf{F}^\T/n = \mbf{U} \mbf{B} \mbf{F}^\T/n + \mbf{F} \mbf{B}^\T \mbf{U}^\T/n + \mbf{U} \mbf{U}^\T/n,$$
we have, by Assumptions \ref{asm:5} and \ref{asm:7},
$$\Vert \mbf{U} \mbf{U}^\T/n \Vert \le p\Vert \mbf{U} \mbf{U}^\T/n - \mbf{\Sigma}\Vert_{\max} + \Vert \mbf{\Sigma}\Vert = \bigOp(p\sqrt{\log p /n})$$
and, by Lemma \ref{lem:UB},
$$\Vert \mbf{U} \mbf{B} \mbf{F}^\T/n \Vert = \Vert \mbf{F} \mbf{B}^\T \mbf{U}^\T/n \Vert
\le \Vert \mbf{F} /\sqrt{n} \Vert \Vert \mbf{U}\mbf{B}\Vert_\text{F}/\sqrt{n} = \bigOp(\sqrt{p\log p}).$$
\item From Lemma \ref{lem:eigenvalue}, it follows that
$$\Vert \widehat{\mbf{\Lambda}} - \mbf{\Lambda} \Vert_{\max} = \bigOp(p\sqrt{\log p/n}), ~~\widehat{\lambda}_{k+1} = \smallop(p),~~\widehat{\lambda}_k = \lambda_k + \smallop(p).$$
\end{enumerate}
Next, recall that $\mbf{F}/\sqrt{n} = (\widetilde{\mbf{F}}/\sqrt{n}) \mbf{D} \mbf{O}_0^\T$ is the singular value decomposition of $\mbf{F}/\sqrt{n}$. Write
\begin{align*}
\Vert \widehat{\mbf{F}}_c^\T\widetilde{\mbf{F}}/n\Vert_\text{F} &= \Vert \widehat{\mbf{F}}_c^\T \mbf{F} \mbf{O}_0 \mbf{D}^{-1}/n\Vert_\text{F}
\le \Vert \mbf{D}^{-1} \Vert_{\max} \Vert \widehat{\mbf{F}}_c^\T \mbf{F} \mbf{O}_0/n\Vert_\text{F}\\
&= \Vert \mbf{D}^{-1} \Vert_{\max} \Vert \widehat{\mbf{F}}_c^\T \mbf{F} /n\Vert_\text{F} = \bigOp(\sqrt{\log p/n}).
\end{align*}
where the second (in)equality follows from the fact that, if $\mbf{T_2}$ is diagonal,
$$\Vert \mbf{T}_1\mbf{T}_2\Vert_\text{F} = \Vert \vectorize(\mbf{I}\mbf{T}_1\mbf{T}_2)\Vert = \Vert \mbf{T}_2 \otimes \mbf{I} \vectorize(\mbf{T}_1) \Vert \le \Vert \mbf{T}_2\Vert_{\max} \Vert \mbf{T}_1 \Vert_\text{F},$$
the third (in)equality uses the orthogonality of $\mbf{O}_0$, and the final (in)equality combines rates given by Lemma \ref{lem:D} and \eqref{step1}.
\end{proof}

\begin{lemma} \label{lem:UB}
Suppose Assumptions \ref{asm:5} and \ref{asm:6} hold. Then
$$\Vert \mbf{U}\mbf{B} \Vert_\text{F} = \bigOp(\sqrt{np\log p}).$$
\end{lemma}
\begin{proof}
Write
\begin{align*}
\mbb{E}\left[\Vert \mbf{U}\mbf{B} \Vert_\text{F}^2 \right]
&= \mbb{E}\left[\trace(\mbf{B}^\T\mbf{U}^\T\mbf{U}\mbf{B})\right] = \trace(\mbf{B}^\T\mbb{E}\left[\mbf{U}^\T\mbf{U}\right]\mbf{B}) = n \times \trace(\mbf{B}^\T \mbf{\Sigma} \mbf{B})\\
&= n \sum_{j=1}^k \mbf{B}_j^\T \mbf{\Sigma} \mbf{B}_j \le n \Vert \mbf{\Sigma} \Vert \sum_{j=1}^k \Vert \mbf{B}_j \Vert^2 \le nkp\Vert \mbf{B}\Vert_{\max}^2 \Vert \mbf{\Sigma} \Vert\\
&\le nkp \Vert \mbf{B}\Vert_{\max}^2 m_q(p)C_0^{1-q} = \smallo(np\log p)
\end{align*}
Applying Markov's inequalities to $\Vert \mbf{U}\mbf{B} \Vert_\text{F}^2$ completes the proof.
\end{proof}

\begin{lemma} \label{lem:D}
Suppose Assumption \ref{asm:6} and \ref{asm:7} holds. Let $\mbf{F}/\sqrt{n} = \widetilde{\mbf{F}}/\sqrt{n} \mbf{D}\mbf{O}_0^\T$ be the singular value decomposition of $\mbf{F}/\sqrt{n}$. Then
$$\Vert \mbf{D}^2 - \mbf{I} _k\Vert_{\max} = \bigOp(\sqrt{\log p/n}).$$
\end{lemma}

\begin{proof}
Write
$$\Vert \mbf{D}^2 - \mbf{I}\Vert_{\max} = \Vert \mbf{D}^2 - \mbf{I}\Vert = \Vert \mbf{F}^\T\mbf{F}/n - \mbf{I}\Vert \le k\Vert \mbf{F}^\T\mbf{F}/n - \mbf{I}\Vert_{\max} = \bigOp( \sqrt{\log p / n})$$.
\end{proof}

\begin{lemma}[Theorem \ref{thm:3}, part (c)] \label{lem:F}
Suppose Assumptions \ref{asm:5}-\ref{asm:7} hold. For some non-singular matrix $\mbf{H}_{k \times k}$ with $\Vert \mbf{H}^\T\mbf{H} - \mbf{I} \Vert = \bigOp(\sqrt{\log p/n})$ and $\Vert \mbf{H}\mbf{H}^\T - \mbf{I} \Vert = \bigOp(\sqrt{\log p/n})$, 
$$\Vert \widehat{\mbf{F}}\mbf{H} - \mbf{F} \Vert_\text{F} = \bigOp(\sqrt{\log p}),$$
\end{lemma}
\begin{proof}
Recall that $\mbf{F}/\sqrt{n} = (\widetilde{\mbf{F}}/\sqrt{n})\mbf{D}\mbf{O}_0^\T$ is the singular value decomposition of $\mbf{F}/\sqrt{n}$. Note that all singular values of $\widehat{\mbf{F}}^\T\widetilde{\mbf{F}}/n$ is bounded by $\Vert \widehat{\mbf{F}}^\T\widetilde{\mbf{F}}/n \Vert \le 1$. Let $\mbf{O}_1$ and $\mbf{O}_2$ consist of the left and right singular vectors of $\widehat{\mbf{F}}^\T\widetilde{\mbf{F}}/n$ (the signs of vectors are properly set such that the singular values of $\widehat{\mbf{F}}^\T\widetilde{\mbf{F}}/n$ are non-negative). Thus
\begin{align*}
\Vert \widehat{\mbf{F}}\mbf{O}_1 - \widetilde{\mbf{F}} \mbf{O}_2\Vert_\text{F}^2/n &= \trace\left[ (\widehat{\mbf{F}}\mbf{O}_1 - \widetilde{\mbf{F}} \mbf{O}_2)^\T (\widehat{\mbf{F}}\mbf{O}_1 - \widetilde{\mbf{F}} \mbf{O}_2) \right]/n\\
&= 2k - 2 \trace\left[\mbf{O}_1^\T (\widehat{\mbf{F}}^\T \widetilde{\mbf{F}}/n) \mbf{O}_2\right] \le 2k - 2 \Vert \mbf{O}_1^\T (\widehat{\mbf{F}}^\T \widetilde{\mbf{F}}/n) \mbf{O}_2 \Vert_\text{F}^2\\
&= 2k - 2 \Vert \widehat{\mbf{F}}^\T \widetilde{\mbf{F}}/n \Vert_\text{F}^2 = \bigOp(\log p/n).
\end{align*}
where the last step uses Lemma \ref{lem:eigenspace}. Set $\mbf{H} = \mbf{O}_1\mbf{O}_2^\T\mbf{D}\mbf{O}_0^\T$ then
\begin{align*}
\Vert \widehat{\mbf{F}} \mbf{H} - \mbf{F} \Vert_\text{F}
&= \Vert (\widehat{\mbf{F}} \mbf{O}_1 - \widetilde{\mbf{F}}\mbf{O}_2) \mbf{O}_2^\T\mbf{D}\mbf{O}_0^\T \Vert_\text{F} \le \Vert \widehat{\mbf{F}} \mbf{O}_1 - \widetilde{\mbf{F}}\mbf{O}_2 \Vert_\text{F} \Vert \mbf{D} \Vert = \bigOp(\log p)\\
\mbf{H}^\T\mbf{H} - \mbf{I} &= \mbf{O}_0 (\mbf{D}^2 - \mbf{I})\mbf{O}_0^\T\\
\mbf{H}\mbf{H}^\T - \mbf{I} &= \mbf{O}_1 \mbf{O}_2^\T (\mbf{D}^2 - \mbf{I})\mbf{O}_2 \mbf{O}_1^\T.
\end{align*}
The eigenvalues of $\mbf{H}\mbf{H}^\T$ or $\mbf{H}^\T\mbf{H}$ are the diagonal elements in $\mbf{D}^2$, which are $\bigOp(\sqrt{\log p/n})$-close to $1$ as shown by Lemma \ref{lem:D}.
\end{proof}

\begin{lemma}[Theorem \ref{thm:3}, part (d)] \label{lem:U}
Suppose Assumptions \ref{asm:5}-\ref{asm:7} hold. $\widehat{\mbf{U}} = (\mbf{I} - \widehat{\mbf{F}} \widehat{\mbf{F}}^\T/n)\mbf{X}$ recovers the latent individual factor matrix $\mbf{U}$ in the sense that
$$\max_{j=1}^p \Vert \widehat{\mbf{U}}_j - \mbf{U}_j \Vert = \bigOp(\sqrt{\log p}).$$
\end{lemma}
\begin{proof}
Recall that $\bm{b}_j$ denote the $j$-th row of $\mbf{B}$, $j=1,\dots,p$. Recall the definition of $\mbf{H}$ in Lemma \ref{lem:F}. It is elementary that $\Vert \mbf{H} \Vert = \bigOp(1)$ and $\Vert \mbf{H}^{-1} \Vert = \bigOp(1)$. Note that $\widehat{\mbf{F}}^\T\widehat{\mbf{F}}/n = \mbf{I}$. Write
\begin{align*}
\widehat{\mbf{U}}_j - \mbf{U}_j 
&= (\mbf{I} - \widehat{\mbf{F}} \widehat{\mbf{F}}^\T/n)\mbf{F}\bm{b}_j - (\widehat{\mbf{F}} \widehat{\mbf{F}}^\T/n) \mbf{U}_j,\\
&= (\mbf{I} - \widehat{\mbf{F}} \widehat{\mbf{F}}^\T/n)(\mbf{F} - \widehat{\mbf{F}}\mbf{H})\bm{b}_j - \widehat{\mbf{F}}(\widehat{\mbf{F}} - \mbf{F}\mbf{H}^{-1})^\T \mbf{U}_j/n - \widehat{\mbf{F}}\mbf{H}^{-\T}\mbf{F}^\T \mbf{U}_j/n
\end{align*}
For the first term,
\begin{align*}
\Vert (\mbf{I} - \widehat{\mbf{F}} \widehat{\mbf{F}}^\T/n)(\mbf{F} - \widehat{\mbf{F}}\mbf{H})\bm{b}_j \Vert
&\le \Vert \mbf{F} - \widehat{\mbf{F}}\mbf{H} \Vert \Vert \bm{b}_j \Vert\\
&\le \Vert \mbf{F} - \widehat{\mbf{F}}\mbf{H} \Vert_\text{F} \sqrt{k} \Vert \mbf{B} \Vert_{\max} = \bigOp(\sqrt{\log p}).
\end{align*}
For the second term,
\begin{align*}
\Vert \widehat{\mbf{F}}(\widehat{\mbf{F}} - \mbf{F}\mbf{H}^{-1})^\T \mbf{U}_j/n \Vert &\le \Vert \widehat{\mbf{F}}/\sqrt{n} \Vert \Vert \widehat{\mbf{F}}\mbf{H} - \mbf{F} \Vert \Vert \mbf{H}^{-1} \Vert \Vert \mbf{U}_j /\sqrt{n}\Vert\\
&\le 1 \times \Vert \widehat{\mbf{F}}\mbf{H} - \mbf{F} \Vert_\text{F} \Vert \mbf{H}^{-1} \Vert \sqrt{\mbf{\Sigma}_{jj} + \Vert \mbf{U}^\T \mbf{U}/n - \bm{\Sigma}\Vert_{\max}} = \bigOp(\sqrt{\log p}).
\end{align*}
For the third term,
\begin{align*}
\Vert \widehat{\mbf{F}} \mbf{H}^{-\T} \mbf{F}^\T \mbf{U}_j/n \Vert &\le  \sqrt{n}\Vert \widehat{\mbf{F}}/\sqrt{n} \Vert \Vert \mbf{H}^{-1} \Vert \Vert \mbf{F}^\T \mbf{U}_j/n \Vert\\
&\le \sqrt{n} \times 1 \times \Vert \mbf{H}^{-1} \Vert \times \sqrt{k} \Vert \mbf{F}^\T\mbf{U}/n\Vert_{\max} = \bigOp(\sqrt{\log p}).
\end{align*}
\end{proof}

\section{Implementation of Gibbs Samplers}
For the prior \eqref{prior}, we set $h$ as the Gaussian density function $h(z) = e^{-z^2/2}/\sqrt{2\pi}$ and $g$ as the inverse-gamma density function with shape $a_0 = 1$ and scale $b_0 = 1$. A Gibbs sampler is implemented to explore the pseudo-posterior distribution \eqref{posterior}. This Gibbs sampler runs towards the pseudo-posterior joint distribution of $(\sigma^2, \bm{\alpha}, \bm{\beta})$ by iterating the following steps: (1) draw $\xi$ given $\bm{\alpha}$ and $\sigma^2$, (2) draw $\bm{\beta}$ given $\xi$, $\bm{\alpha}$ and $\sigma^2$, (3) draw $\bm{\alpha}$ given $\xi,\bm{\beta}$ and $\sigma^2$, (4) draw $\sigma^2$ given $\xi,\bm{\beta}$ and $\bm{\alpha}$.

For simplicity, we illustrate the implementation details with $s_0 = 1$, $\tau_j = 1$ for $j = 1,\dots,p$. For the first step, we have
\begin{align*}
\widehat{\pi}(\xi, \bm{\beta}_\xi | \sigma^2, \bm{\alpha}, \widehat{\mbf{F}}, \widehat{\mbf{U}}, \mbf{Y})
&\propto p^{-|\xi|} \exp{\left(-\frac{\Vert \mbf{Y} - \widehat{\mbf{F}}\bm{\alpha} - \widehat{\mbf{U}}_\xi\bm{\beta}_\xi\Vert^2}{2\sigma^2}\right)} \sigma^{-|\xi|}\exp{\left(-\frac{\Vert \bm{\beta}_\xi\Vert^2}{2\sigma^2}\right)}.
\end{align*}
This implies
\begin{align}
&~~~\widehat{\pi}(\xi | \sigma^2, \bm{\alpha}, \widehat{\mbf{F}}, \widehat{\mbf{U}}, \mbf{Y}) \nonumber\\
&= \int \widehat{\pi}(\xi, \bm{\beta}_\xi | \sigma^2, \bm{\alpha}, \widehat{\mbf{F}}, \widehat{\mbf{U}}, \mbf{Y}) d\bm{\beta}_\xi \nonumber\\
&\propto p^{-|\xi|} \det(\mbf{S}_\xi)^{-1/2} \exp{\left(-\frac{\left(\mbf{Y} - \widehat{\mbf{F}}\bm{\alpha}\right)^\T \mbf{S}_\xi^{-1}\left(\mbf{Y} - \widehat{\mbf{F}}\bm{\alpha}\right)}{2\sigma^2}\right)}, \label{cp xi}
\end{align}
where $\mbf{S}_\xi = \widehat{\mbf{U}}_\xi \widehat{\mbf{U}}_\xi^\T + \mbf{I}$. However, it is computationally prohibitive to directly sample from this conditional distribution, as $\xi$ takes $2^p$ possible values. As a remedy, we flip $Z_j = 1\{j \in \xi\}$ in Gibbs random scans. In our experiments, we found that just one random scan suffices for the proposed method to perform well. Details of flipping $Z_j$ will be given at the end of this section.

For the second step, we derive, by elementary calculus,
$$\widehat{\pi}(\bm{\beta}_\xi | \sigma^2, \bm{\alpha}, \xi, \widehat{\mbf{F}}, \widehat{\mbf{U}},\mbf{Y}) \sim \mc{N}\left(\left(\widehat{\mbf{U}}_\xi^\T\widehat{\mbf{U}}_\xi+\mbf{I}\right)^{-1}\widehat{\mbf{U}}_\xi^\T\left(\mbf{Y} - \widehat{\mbf{F}}\bm{\alpha}\right), \sigma^2 \left(\widehat{\mbf{U}}_\xi^\T\widehat{\mbf{U}}_\xi+\mbf{I}\right)^{-1}\right).$$
Recall that $\bm{\beta}_{\xi^c} \equiv 0$. Similarly, for the third step,
\begin{align*}
\widehat{\pi}(\bm{\alpha} | \sigma^2, \bm{\beta}, \xi, \widehat{\mbf{F}}, \widehat{\mbf{U}},\mbf{Y})
&\sim \mc{N}\left(\left(\widehat{\mbf{F}}^\T\widehat{\mbf{F}}+\mbf{I}\right)^{-1}\widehat{\mbf{F}}^\T\left(\mbf{Y} - \widehat{\mbf{U}}_\xi\bm{\beta}_\xi\right), \sigma^2 \left(\widehat{\mbf{F}}^\T\widehat{\mbf{F}}+\mbf{I}\right)^{-1}\right)\\
&\sim  \mc{N}\left(\widehat{\mbf{F}}^\T\left(\mbf{Y} - \widehat{\mbf{U}}_\xi\bm{\beta}_\xi\right) / (n+1), \sigma^2 \mbf{I}/(n+1)\right)
\end{align*}
The final step uses the conjugacy of normal distribution and inverse-gamma distribution
\begin{align*}
&~~~\widehat{\pi}(\sigma^2 |\bm{\alpha}, \bm{\beta}, \xi, \widehat{\mbf{F}}, \widehat{\mbf{U}},\mbf{Y})\\
&\propto g(\sigma^2|a_0,b_0) \mc{N}\left(\bm{\beta}_\xi|\mbf{0},\sigma^2\mbf{I}\right)\mc{N}\left(\bm{\alpha}|\mbf{0},\sigma^2\mbf{I}\right)\mc{N}\left(\mbf{Y}|\widehat{\mbf{F}}\bm{\alpha}+\widehat{\mbf{U}}_\xi\bm{\beta}_\xi, \sigma^2 \mbf{I} \right)\\
&\propto g\left(\sigma^2\left|a_0 + \frac{|\xi|+k+n}{2}, b_0 + \frac{\Vert \bm{\beta}_\xi \Vert^2 + \Vert \bm{\alpha} \Vert^2 + \Vert \mbf{Y} - \widehat{\mbf{F}}\bm{\alpha} - \widehat{\mbf{U}}_\xi\bm{\beta}_\xi\Vert^2}{2}\right.\right).
\end{align*}

In the first step, in order to sample from the conditional distribution \eqref{cp xi}, we flip $Z_j$ with probability
\begin{align*}
\widehat{\pi}(Z_j = 1 | \{Z_{j'}\}_{1 \le j' \ne j \le p}, \sigma^2, \bm{\alpha}, \widehat{\mbf{F}}, \widehat{\mbf{U}}, \mbf{Y}) = \left\{1 + \left[\frac{\widehat{\pi}(\xi = \omega \cup \{j\}| \sigma^2, \bm{\alpha}, \widehat{\mbf{F}}, \widehat{\mbf{U}}, \mbf{Y})}{\widehat{\pi}(\xi = \omega | \sigma^2, \bm{\alpha}, \widehat{\mbf{F}}, \widehat{\mbf{U}}, \mbf{Y})}\right]^{-1}\right\}^{-1},
\end{align*}
where $\omega = \{j' \ne j: Z_j' = 1\}$. The posterior probability ratio is computed as
\begin{align*}
p^{-1} \left[\frac{\det(\mbf{S}_{\omega \cup \{j\}})}{\det(\mbf{S}_\omega)}\right]^{-1/2} \exp{\left(-\frac{\left(\mbf{Y} - \widehat{\mbf{F}}\bm{\alpha}\right)^\T \left(\mbf{S}_{\omega \cup \{j\}}^{-1} - \mbf{S}_\omega^{-1}\right)\left(\mbf{Y} - \widehat{\mbf{F}}\bm{\alpha}\right)}{2\sigma^2}\right)},
\end{align*}
where we derive, by Sylvester's determinant theorem and properties of Schur complements, that
\begin{align*}
\frac{\det(\mbf{S}_{\omega \cup \{j\}})}{\det(\mbf{S}_\omega)}
&= \frac{\det(\widehat{\mbf{U}}_{\omega \cup \{j\}}^\T\widehat{\mbf{U}}_{\omega \cup \{j\}}+\mbf{I})}{\det(\widehat{\mbf{U}}_{\omega}^\T\widehat{\mbf{U}}_{\omega}+\mbf{I})} = (\widehat{\mbf{U}}_j^\T\widehat{\mbf{U}}_j +1) -  \widehat{\mbf{U}}_j^\T\widehat{\mbf{U}}_\omega\left[\widehat{\mbf{U}}_\omega^\T\widehat{\mbf{U}}_\omega + \mbf{I}\right]^{-1}\widehat{\mbf{U}}_\omega^\T\widehat{\mbf{U}}_j,
\end{align*}
and, by Sherman-Morrison-Woodbury identity, that
\begin{align*}
\mbf{S}_{\omega \cup \{j\}}^{-1} - \mbf{S}_\omega^{-1}
&= \widehat{\mbf{U}}_{\omega} \left(\widehat{\mbf{U}}_{\omega}^\T \widehat{\mbf{U}}_{\omega} + \mbf{I} \right)^{-1} \widehat{\mbf{U}}_{\omega}^\T - \widehat{\mbf{U}}_{\omega \cup \{j\}} \left( \widehat{\mbf{U}}_{\omega \cup \{j\}}^\T \widehat{\mbf{U}}_{\omega \cup \{j\}} + \mbf{I} \right)^{-1} \widehat{\mbf{U}}_{\omega \cup \{j\}}^\T.
\end{align*}

As shown in our theoretical analyses, this Gibbs sampler will deal with $|\omega| \le (M_0+1)s$ in most time. The computation of terms in the posterior probability ratio is numerically stable as the Gram matrices involved in the computation has small size. The computation

It is also time-efficient with complexity $\bigO(n|\omega|^2) \le \bigO(ns^2)$. The overall time complexity running $T$ iterations of Gibbs samplers in our Bayesian method is then $\bigO(Tpns^2)$. In contrast, the factor-adjusted lasso method costs $\bigO(p^3)$ time. In the simulation studies, we choose $T = 20$, $n=200$, $p=500$, $s = 5$ as the typical setting, and observe that our Bayesian method runs as fast as its lasso analogue.

\end{document}